\documentclass[12pt]{article}

\usepackage[utf8]{inputenc}
\usepackage{amsfonts, amsthm, amsmath, graphicx, enumerate, verbatim, amssymb, hyperref, mathrsfs}
\usepackage[authoryear]{natbib}
\usepackage{setspace}
\usepackage{xcolor}

\hypersetup{
    colorlinks,
    linkcolor={red!50!black},
    citecolor={blue!50!black},
    urlcolor={blue!80!black}
}

\usepackage {tikz}
\usetikzlibrary{decorations.pathreplacing}
\tikzset{
mybrace/.style={decorate,decoration={brace,aspect=#1,amplitude=10pt}}
}

\tikzset{
mybrace2/.style={decorate,decoration={brace,aspect=#1,amplitude=20pt}}
}

\makeatletter
\newcommand{\neutralize}[1]{\expandafter\let\csname c@#1\endcsname\count@}
\makeatother

\newtheorem{thrm}{Theorem}
\newtheorem{prop*}{Proposition}
\newtheorem{proposition}{Proposition}
\newtheorem{cor}{Corollary}
\newtheorem{thm*}{Theorem}
\newtheorem{lemma}{Lemma}
\theoremstyle{definition}
\newtheorem{defin}{Definition}

\usepackage[top=1in, bottom=1in, left=1in, right=1in]{geometry}
\usepackage{booktabs}

\newtheoremstyle{named}{}{}{\itshape}{}{\bfseries}{.}{.5em}{\thmnote{#3}#1}
\theoremstyle{named}
\newtheorem*{namedtheorem}{}

\newenvironment{thrmbis}[1]
  {\renewcommand{\thethrm}{\ref*{#1}$'$}%
   \neutralize{thrm}\phantomsection
   \begin{thrm}}
  {\end{thrm}}

\title{Innovation and Strategic Network Formation}
\author{Krishna Dasaratha\thanks{Yale University. Email: krishnadasaratha@gmail.com. I am deeply grateful to Drew Fudenberg, Benjamin Golub, Matthew Rabin, and Tomasz Strzalecki for their guidance and support. I would also like to thank Andrea Galeotti and three anonymous referees for their very helpful suggestions and Daron Acemoglu, Mohammad Akbarpour, Leonie Baumann, Yochai Benkler, Laura Doval, Edward Glaeser, Kevin He, Nir Hak, Scott Kominers, Shengwu Li, Jonathan Libgober, George Mailath, Michael Rose, Stefanie Stantcheva, Eduard Talam\`{a}s, Omer Tamuz, Rakesh Vohra, and Alexander Wolitzky for valuable comments and discussion.}}

\date{\today}

\begin{document}

\maketitle

\begin{center}\textbf{Abstract}
\end{center}

We study a model of innovation with a large number of firms that create new technologies by combining several discrete ideas. These ideas are created via private investment and spread between firms. Firms face a choice between secrecy, which protects existing intellectual property, and openness, which facilitates learning from others. Their decisions determine interaction rates between firms, and these interaction rates enter our model as link probabilities in a learning network. Higher interaction rates impose both positive and negative externalities, as there is more learning but also more competition. We show that the equilibrium learning network is at a critical threshold between sparse and dense networks. At equilibrium, the positive externality from interaction dominates: the innovation rate and welfare would be dramatically higher if the network were denser. So there are large returns to increasing interaction rates above the critical threshold. Nevertheless, several natural types of interventions fail to move the equilibrium away from criticality. One effective policy solution is to introduce informational intermediaries, such as public innovators who do not have incentives to be secretive. These intermediaries can facilitate a high-innovation equilibrium by transmitting ideas from one private firm to another.
\pagenumbering{gobble}
\newpage
\pagenumbering{arabic}
\setcounter{page}{1}

\setstretch{1.5}

\section{Introduction}

A growing body of empirical research suggests that interactions between inventors are an important part of innovation.\footnote{The benefits from interactions between inventors and movement of inventors have been quantified empirically by \cite*{akcigit2018dancing}, \cite*{kerr2008ethnic}, \cite*{samila2011venture}, among others.} New technologies are often produced by combining individual insights with learning from peers, which confers large benefits on firms and inventors engaged in such learning. When highly-connected clusters of firms emerge in a location, as in the technology industry in Silicon Valley, inventors in these areas are much more productive.

But frequent collaboration and learning are not assured even when inventors in a given industry co-locate (\citealp*{saxenian1996regional} gives a well-known example). Rather, interaction patterns depend on firms' decisions, such as how much to encourage their employees to interact with employees from other firms. These interactions let a firm learn from other companies and inventors. There are also downsides for the firm, as employees may share valuable information. A more secretive approach allows firms to prevent potential competition by protecting intellectual property. In making these types of decisions, firms and inventors must choose between openness and secrecy.

We develop a theory of firms' endogenous decisions about how much to interact with other firms, and the consequences for information flows and the rate of innovation. Firms' choices determine the probabilities of links in a \emph{learning network}; we thus contribute to a theoretical literature on network formation.\footnote{An important feature of these choices is that firms combine different ideas to create new technologies (see also \citealp*{acemoglu2020endogenous} and \citealp*{chen2021capability} for related combinatorial production functions). Strategic complementaries play a crucial role in many network games (e.g., \citealp*{bramoulle2014strategic} and \citealp*{ballester2006who}), and in our setting complementarities arise endogenously from this process of combining ideas.} We show that at equilibrium, the  learning network has a special structure that would be unlikely to arise exogenously: it is at a \emph{critical threshold} between sparse and dense networks. This implies extreme inefficiencies arise at equilibrium, and we analyze the welfare and policy implications.

We study a framework where each firm chooses two intensities: how open to be as well as how much to invest in R\&D. The choices of levels of openness determine interaction rates between firms, and the probability that one firm learns from another is equal to the interaction rate between the two firms. If a given firm is more open, that firm  is more likely to learn \emph{ideas} from  other firms---but other firms are more likely to learn its ideas. These ideas are valuable because they can be combined to make new \emph{technologies}, which are finite sets of ideas (as in  \citealp*{weitzman1998recombinant}). Firms can generate profits by producing technologies, but the profits from a technology are erased by competition if another firm also knows the component ideas in that technology.

Our first contribution, which is methodological, is to develop a theory of endogenous formation of random networks in the context of our economic application. Learning opportunities are random events, and their realizations determine a learning network. We therefore consider link formation decisions with uncertainty in the matching process, while the leading approach in the literature on network-formation games focuses on deterministic models that allow agents to choose particular links (\citealp*{jackson1996strategic} and \citealp*{bala2000noncooperative}).  Since we take actions to be continuous choices that translate to interaction rates, optimal behavior satisfies first-order conditions rather than a high-dimensional system of combinatorial inequalities.

A key feature of our model is that ideas can spread several steps through this network: when one firm learns from another, the information transferred can include ideas learned from a third firm. We refer to this as indirect learning.\footnote{By contrast, existing work on strategic formation of random networks largely focuses on direct connections (see, e.g., \citealp*{currarini2009economic}).} Under indirect learning, firms' incentives depend on the global structure of the network. Each firm would like to learn many ideas, since then the firm could combine these ideas to produce a large number of new technologies, and much of this learning can be indirect.

Analyzing the global structure of the network leads to our second contribution, which is to establish the \emph{criticality of equilibrium}.  When there are many firms, learning outcomes depend dramatically on whether the learning network is sparsely connected or densely connected. If firms' interaction rates are below a critical threshold, the learning network consists of many small clusters of firms who learn few ideas. Above the threshold, the learning network has a {giant component} asymptotically: a large group of firms who learn a large number of ideas and can incorporate these ideas into many new technologies. To determine where equilibrium lies with respect to this threshold, we analyze an individual firm's decision problem in each of these two domains, i.e. when other firms form a sparse or dense network.

The main result is that the equilibrium interaction rates are at the critical threshold between sparse and dense networks. Firms would deviate to interact more if the network were likely to be sparse, and deviate to interact less if the network were likely to be dense. Intuitively, in sparse networks firms would increase interaction rates to fill central positions in the network, known in sociology as `structural holes', which enable the firm to combine ideas learned via different interactions.\footnote{The concept of structural holes, introduced by \cite*{burt1992structural}, refers to network positions allowing agents to combine information from different connections or spread information between groups.} As others interact more, these structural holes disappear, and indeed firms tend to learn the same ideas repeatedly from different interactions. So the incentives to be more open are weaker relative to the incentives to be secretive.

We next consider the implications of the main result for welfare and policy. Since equilibrium is at the critical threshold, a giant component would emerge if all firms shifted interaction rates slightly above equilibrium levels. Firms learn relatively few ideas at the threshold, but could learn many more with a bit more interaction. More learning has benefits but also leads to more competition, which may discourage private investment. We show there are still unboundedly large improvements in innovation and welfare (as the number of firms grows large) from interventions that increase interaction rates above equilibrium levels. A consequence is that increasing interaction rates is a first-order concern in designing policy. By contrast, policies targeting  decisions about private investment rather than interaction, such as subsidies to R\&D, have minimal effect at equilibrium---but can be valuable if paired with interventions to increase openness.

Intervening to improve equilibrium learning networks is difficult, and policies such as broadly subsidizing interaction will be undermined by firms' equilibrium responses. We discuss one type of policy change that does induce more productive interaction patterns, which is to introduce \textit{public innovators} who do not have incentives to be secretive. For example, governments could fund academic researchers who are especially willing to interact with other researchers, including those in industry. The key is that public innovators can serve as informational intermediaries, transmitting ideas between private firms. They play a valuable role even after considering the equilibrium response of the profit-maximizing firms, who may adjust to be more secretive.

We next explore which features of the baseline model are needed to obtain a critical equilibrium. One assumption is that learning probabilities are symmetric across pairs of firms. We show that equilibrium remains critical even when firms have different propensities to learn from others, which allows some firms to be better at protecting ideas than others. The key feature driving the main result is that there is a margin along which firms can acquire incoming links at the cost of a higher probability of outgoing links.

The criticality of equilibrium is robust to alternate specifications of the benefits from learning from others but more sensitive to the costs of outgoing links. The baseline model assumes firms' profits are additive across technologies, but equilibrium remains critical if there are increasing or slightly decreasing returns to producing many technologies. Indeed, the crucial property is that payoffs are convex in the number of ideas learned by a firm. Equilibrium outcomes do depend, however, on how much firms stand to lose from outgoing links.  In particular, the results described above assume zero profits in competitive markets. If profits under competition are instead positive (e.g., because of a first-mover advantage or markets with collusion between firms), then incentives toward secrecy will be weaker and so equilibria will be above the critical threshold. These results give testable predictions about the relationship between market structure and outcomes such as the innovation rate.

 
Our final results ask how formal intellectual property rights change the incentives to interact. Consider the consequences of granting patents to a positive fraction of ideas, e.g., allowing hardware but not software ideas to be patented. Patents mitigate firms' incentives to be secretive, but can also discourage exchange of ideas. Firms with patents are more open but are also less desirable partners in interactions (at least when ideas are only transmitted directly). The resulting adverse selection in interaction can deter firms from collaborating with others. We show that patent rights can therefore prevent any productive interactions at equilibrium. If indirect learning is important, firms with patents will be informational intermediaries, like the public innovators above. In this case there are benefits to allowing patents, but it turns out that the optimal policy is often to only allow patents for a very small fraction of ideas.

%

At a technical level, this paper develops tools for studying incentives in random network settings. These tools are most applicable to analyzing decisions in network models with complementarities between indirect connections. Classical results in graph theory characterize the component structure of large random networks, and thus in our context the number of ideas firms will learn (\citealp*{karp1990transitive} and \citealp*{luczak1990phase}). But due to the complementarities between ideas, firms' incentives also depend on \emph{how} these ideas are learned, e.g., via many interactions or a few interactions. To capture these complementarities, we prove a key lemma relating a firm's equilibrium action to the extent to which technologies combine ideas from distinct interactions. An additional challenge is that to understand incentives, it is not enough to analyze ``leading terms''. Vanishing-probability events and lower-order terms in link probabilities could substantially affect payoffs. Our analysis therefore requires careful treatment of the graph branching process governing the number of ideas learned from each interaction.


\subsection{Related Literature}

This paper relates to research in network theory, especially network formation, and to models of innovation.

At a methodological level, we develop a theory of strategic network formation with probabilistic links. A large literature since \cite*{jackson1996strategic} and \cite*{bala2000noncooperative} considers endogenous network formation assuming that agents can choose their links exactly.\footnote{The pairwise stability solution concept from \cite*{jackson1996strategic} and variants have been applied to network formation in many settings, including innovation (\citealp*{goyal2001r}, \citealp*{konig2011recombinant}).} Because equilibrium is then characterized by a large system of inequalities, these models illustrate key externalities in special cases but remain largely or entirely intractable in many others. By instead considering agents making a single choice about openness under uncertainty, we obtain a smooth model of link formation that can be solved via basic optimization techniques combined with analyses of random graphs.

Under this random-network model of network formation, incentives to form links depend on the `phase transitions' between sparse and dense networks.\footnote{\cite*{golub2010strategic} also study network formation with phase transitions, and allow payoffs to depend on distance one and two connections. An important feature of our model is that firms' decisions depend on the global network structure rather than only local connections.} Economic models involving  phase transitions have been recently explored in the context of diffusion processes by \cite*{campbell2013word}, \cite*{akbarpour2018diffusion}, and \cite*{sadler2020diffusion}. These models let adoption and/or seeding decisions depend on component structure in an underlying network. We instead study equilibria of a game in which agents endogenously make decisions about how much to interact with others, and find there is a subtle interplay between strategic incentives and the global network structure.

An alternate approach to smooth network formation is to consider weighted networks, so that each link has an intensity (as in \citealp*{cabrales2011social}, \citealp*{baumann2021model}, and \citealp*{griffith2019continuous}).\footnote{Relatedly, \cite{erol2020civil} consider a continuum model where agents endogenously and randomly form links with a positive fraction of peers.} In some settings with unweighted links, this approach yields a deterministic approximation of a game on the underlying random network. But for discrete processes such as the diffusion of an idea, the random and deterministic models can be very different. Indeed, we show that there is a critical threshold corresponding to important discontinuities in network structure and outcomes that do not arise in deterministic models such as \cite*{cabrales2011social}.\footnote{\cite{berliant2008knowledge} study a related deterministic model of knowledge creation with individuals choosing to producing knowledge alone or with a partner. Their analysis focuses on interaction patterns in small disconnected groups, while we find social networks can be more connected in important ways.}

In existing literature on innovation, approaches incorporating interactions between firms generally model these interactions as either mechanical spillovers or learning via imitation. A common approach is to choose a convenient functional form for spillovers, usually motivated by tractability within a macroeconomic (e.g., \citealp*{kortum1997research}) or network-theory (\citealp*{konig2012efficiency}) framework.\footnote{In simulations, \cite*{baum2010network} study the formation of innovation networks via a mechanical process.} By microfounding these spillovers, which arise endogenously within the innovative process, we can study how spillovers respond to policy interventions.

In a different approach, which relies on a quality-ladders framework, interactions give firms a chance to catch up as innovation proceeds vertically through improvements in the quality of existing technologies (e.g., \citealp*{perla2014equilibrium}, \citealp*{akcigit2018dancing}, and \citealp*{konig2016innovation}). We instead explicitly model innovation horizontally as a process of combining distinct ideas, which serve as building blocks for new technologies. Related models appear in \cite*{weitzman1998recombinant} and \cite*{acemoglu2020endogenous}, which focus on the evolution of the total amount of innovation over time and do not involve learning or informational spillovers between firms. We find that when new technologies can be created in this way, changes in interaction patterns can have much larger consequences for the rate of innovation than in quality-ladders models.\footnote{By assuming a continuum of firms, macroeconomic models of imitation often implicitly restrict to network structures without a giant component.}

\section{Model}\label{sec:model}

We first describe our model formally and provide an example. We will then discuss interpretations of the model and its assumptions.

\subsection{Basic Setup}

The model includes two key concepts: ideas and technologies. Ideas are the components in technologies, and each technology combines several ideas. Firms will obtain profits from selling new technologies, and must acquire ideas to produce these technologies.

There are $n>1$ firms $1,\hdots,n$. Each firm $i$ can potentially discover a distinct \textbf{idea}, also denoted by $i$. We assume each firm can discover a single idea for simplicity, and Appendix~\ref{sec:size} allows firms to potentially discover multiple ideas. We let $I \subset \{1,\hdots,n\}$ be the set of ideas that are discovered.

Each firm $i$ chooses a probability $p_i \in [0,1)$ of discovering this idea and pays investment cost $c(p_i)$. We will assume that $c$ is continuously differentiable, increasing, and convex with $c(0)=0$ and $\lim_{p \rightarrow 1^-} c(p)=\infty$. The realizations of discoveries are independent.

A \textbf{technology} $t =\{i_1,\hdots,i_k\}$ is a set of $k$ ideas $i_1,\hdots,i_k \in I$, where $k>1$ is an exogenous parameter capturing the complexity of technologies.\footnote{In the baseline model, the parameter $k$ is the same for all firms.} Each idea $i \in t$ must be discovered by the corresponding firm to be included in a technology. There are therefore $\binom{n}{k}$ potential technologies, and a firm $i$ can produce more than one technology. An equivalent interpretation, which we describe in Section~\ref{sec:discussion}, is that firms can explore many potential technologies and only a small fraction are feasible to bring to market.

Each firm $i$ chooses a level of \textbf{openness} $q_i\in [0,1]$. Given choices $q_i$ and $q_j$, the interaction rate between $i$ and $j$ is $\iota(q_i,q_j)$. We will assume the multiplicative interaction rate $$\iota(q_i,q_j)=q_iq_j$$
except when a more general interaction rate is explicitly stated (in Theorem~\ref{thm:interaction}). The interaction rate determines the probability that firm $i$ learns from firm $j$ and vice versa, as described below.

The timing of the model is simultaneous: firms choose actions $p_i$ and $q_i$ and then all learning occurs. We denote the vectors of actions by $(\mathbf{p},\mathbf{q})$. When actions are symmetric, we will refer to $p_i$ by $p$ and $q_i$ by $q$.

Given actions $\mathbf{p}$ and $\mathbf{q}$, we denote the set of ideas that firm $i$ learns from others by $I_i(\mathbf{p},\mathbf{q}) \subset I$. This is a random set depending on realizations of learning and discoveries. We now describe how learning occurs.

With probability $\iota(q_i,q_j),$ firm $i$ learns directly from firm $j$. In this case, firm $i$ learns idea $j$ if $j\in I$. If firm $i$ learns directly from firm $j$, then with probability $\delta \in [0,1]$, firm $i$ also learns indirectly through firm $j$. In this case, firm $i$ also learns all ideas in $I_j(\mathbf{p},\mathbf{q})$. All realizations of direct and indirect learning are independent, and in particular, firm $i$ can learn from firm $j$ without $j$ learning from $i$.

When $\delta=0$ there is only \textbf{direct learning}, while when $\delta>0$ \textbf{indirect learning} can also occur. When $\delta>0$ we define a directed network, which we call the indirect-learning network, with nodes $1,\hdots,n$ and a link from node $j$ to node $i$ if firm $i$ learns indirectly through firm $j$.

A firm $i$ receives payoff $1$ from each \textbf{proprietary technology} $t$. A technology $t$ is proprietary for firm $i$ if (1) $i \in t$ and (2) $i$ is the unique firm such that $j \in \{i\} \cup I_i(\mathbf{p},\mathbf{q}) $ for all $j \in t$. In words, the technology contains firm $i$'s idea and firm $i$ is the unique firm that knows all ideas in the technology, so firm $i$ has a monopoly over the technology.

If $t$ is not a proprietary technology for firm $i$, then firm $i$ receives payoff $0$ from the technology $t$. In Section \ref{sec:generalprofits}, we will consider more general payoff structures in which (1) payoffs are not additive across technologies and (2) firms can receive non-zero payoffs $f(m)$ if $m>0$ other firms know all ideas contained in a technology.

\subsection{Payoffs}

Given actions $(\mathbf{p},\mathbf{q})$ and a firm $i$, we define $T_i(\mathbf{p},\mathbf{q})$ to be the set of technologies $t$ such that (1) $i \in t$ and (2) firm $i$ knows all ideas $j\in t$. The \textbf{proprietary technologies} $PT_i(\mathbf{p},\mathbf{q})  \subset T_i(\mathbf{p},\mathbf{q})$ for $i$ are then the subset of technologies $t \in T_i(\mathbf{p},\mathbf{q})$ such that no other firm knows all ideas in $t$. Note that these sets are random objects depending on link realizations. The expected payoff to firm $i$ is
$$U_i(\mathbf{p},\mathbf{q})=\mathbb{E}\left[|PT_i(\mathbf{p},\mathbf{q})|\right] -c(p_i).$$

To further illustrate payoffs, we write the cardinality of $PT_i(\mathbf{p},\mathbf{q})$ explicitly when $\delta=1$. Recall that $I_i(\mathbf{p},\mathbf{q})$ is the set of ideas learned by firm $i$ given actions $(\mathbf{p},\mathbf{q})$. Like $PT_i(\mathbf{p},\mathbf{q})$, this is also a random object.

When $\delta = 1$, the expected payoffs to firm $i$ are
$$U_i(\mathbf{p},\mathbf{q})=p_i \cdot \mathbb{E}\left[\binom{|I_i(\mathbf{p},\mathbf{q})|}{k-1}\right] \cdot \prod_{j \neq i}(1-\iota(q_i,q_j)) - c(p_i).$$
A technology $t$ that $i$ profits from consists of $i$'s private idea, which is developed with probability $p_i$, and a choice of $(k-1)$ other ideas known to $i$. The firm $j$ faces competition if and only if some firm learns all of $i$'s ideas, and the probability that this does not occur is $\prod_{j \neq i}(1-\iota(q_i,q_j))$. Finally, the private investment cost is $c(p_i)$.

In general, a firm can face competition for a technology $t$ in two ways. First, a firm $j$ can learn all of firm $i$'s ideas via indirect learning. Second, a firm $j$ can learn $i$'s private idea directly and then the other ideas in the technology $t$ from links with firms other than $i$. The probability of the second possibility is more difficult to express in closed form, and in general depends on the technology $t$. We will show that when there is not too much interaction, most competition comes via the first channel.

\subsection{Example}

To illustrate the mechanics of the model, we describe a simple example with $n=4$ firms and complexity $k=3$. Suppose that firms choose some actions $(\mathbf{p},\mathbf{q})$. As an example, we consider the particular realizations such that (1) ideas are discovered by firms in $I = \{1,3,4\}$ and (2) firm $1$ learns indirectly through firm $2$ and directly from firm $3$, firm 3 learns indirectly through firm $1$, and firm $3$ learns directly from firm 4.

The network and ideas are shown in Figure \ref{fig:example1}. Black circles correspond to firms with ideas $i \in I$, i.e., firms that discover ideas, while white circles correspond to firms with ideas $i \notin I$, i.e., firms that do not discover ideas. Solid arrows denote indirect learning links, while dashed arrows indicate only direct learning occurred.

\begin{figure}
\begin{center}
\includegraphics[scale=.4, trim = 0cm 2cm 0cm 2cm]{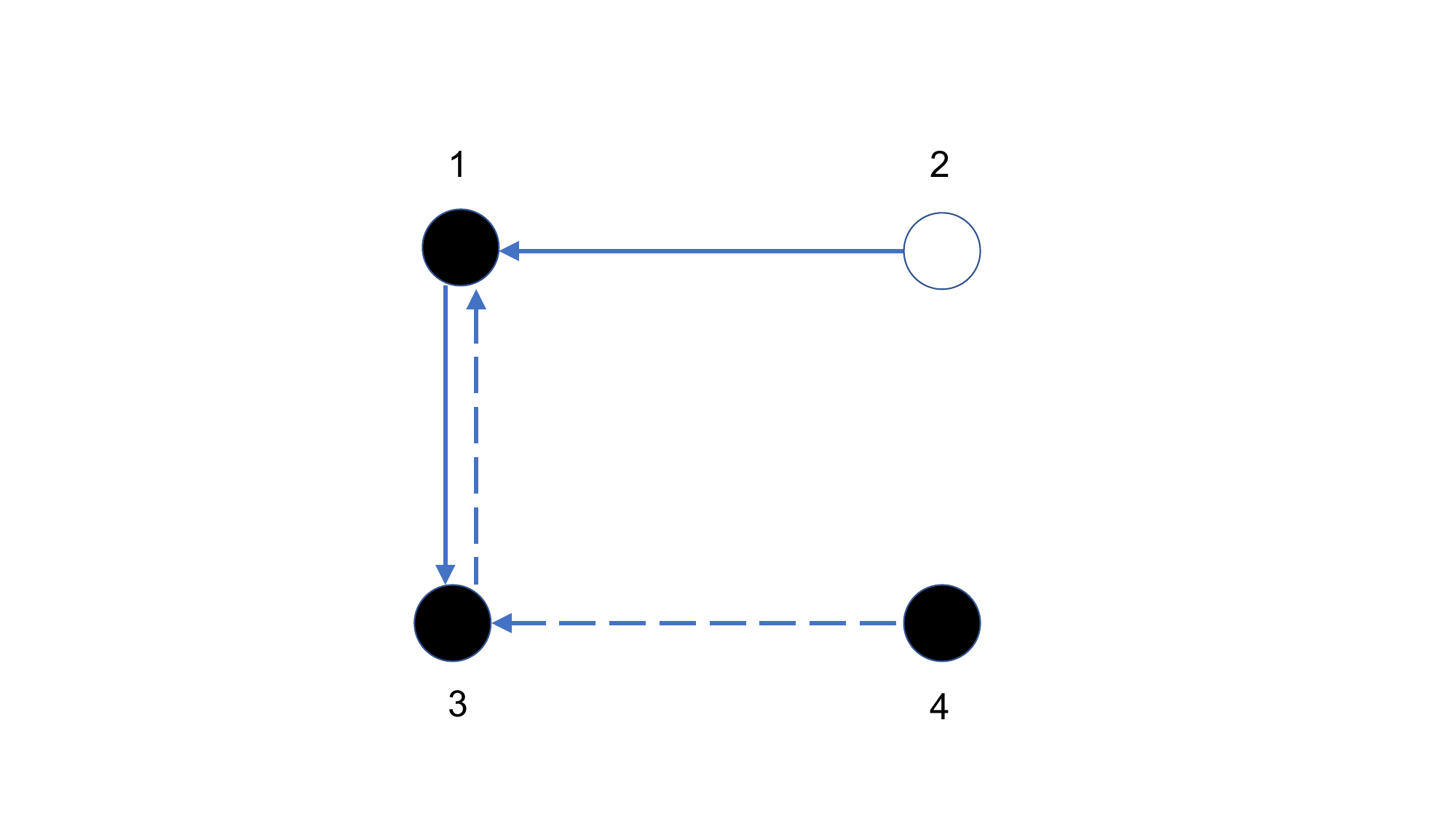}

\caption{Network with four firms and $k=3$. Black circles are firms that discover ideas while white circles do not discover ideas. Arrows are in the direction of information flow, and dashed arrows indicate direct learning while solid arrows indicate indirect learning. The only technology produced is $t=\{1,3,4\}$ and firm $3$ receives monopoly profit.}\label{fig:example1}
\end{center}
\end{figure}

Since $k=3$, the unique technology $t$ consisting of ideas in $I$ is $t = \{1,3,4\}$. The realizations of the sets $I_i(\mathbf{p},\mathbf{q})$ of ideas learned from others are
$$I_1(\mathbf{p},\mathbf{q})= \{3\}, I_2(\mathbf{p},\mathbf{q}) = \emptyset, I_3(\mathbf{p},\mathbf{q})=\{1,4\},I_4(\mathbf{p},\mathbf{q})=\emptyset.$$
Because firm $3$ is the unique firm such that $t \subset I_i(\mathbf{p},\mathbf{q}) \cup \{i\}$ and we have $3 \in t$, firm $3$ produces the technology $t$ and receives monopoly profit of $1$ for that technology. There are no profits from any other technologies.

\begin{figure}
\begin{center}
\includegraphics[scale=.4, trim = 0cm 2cm 0cm 2cm]{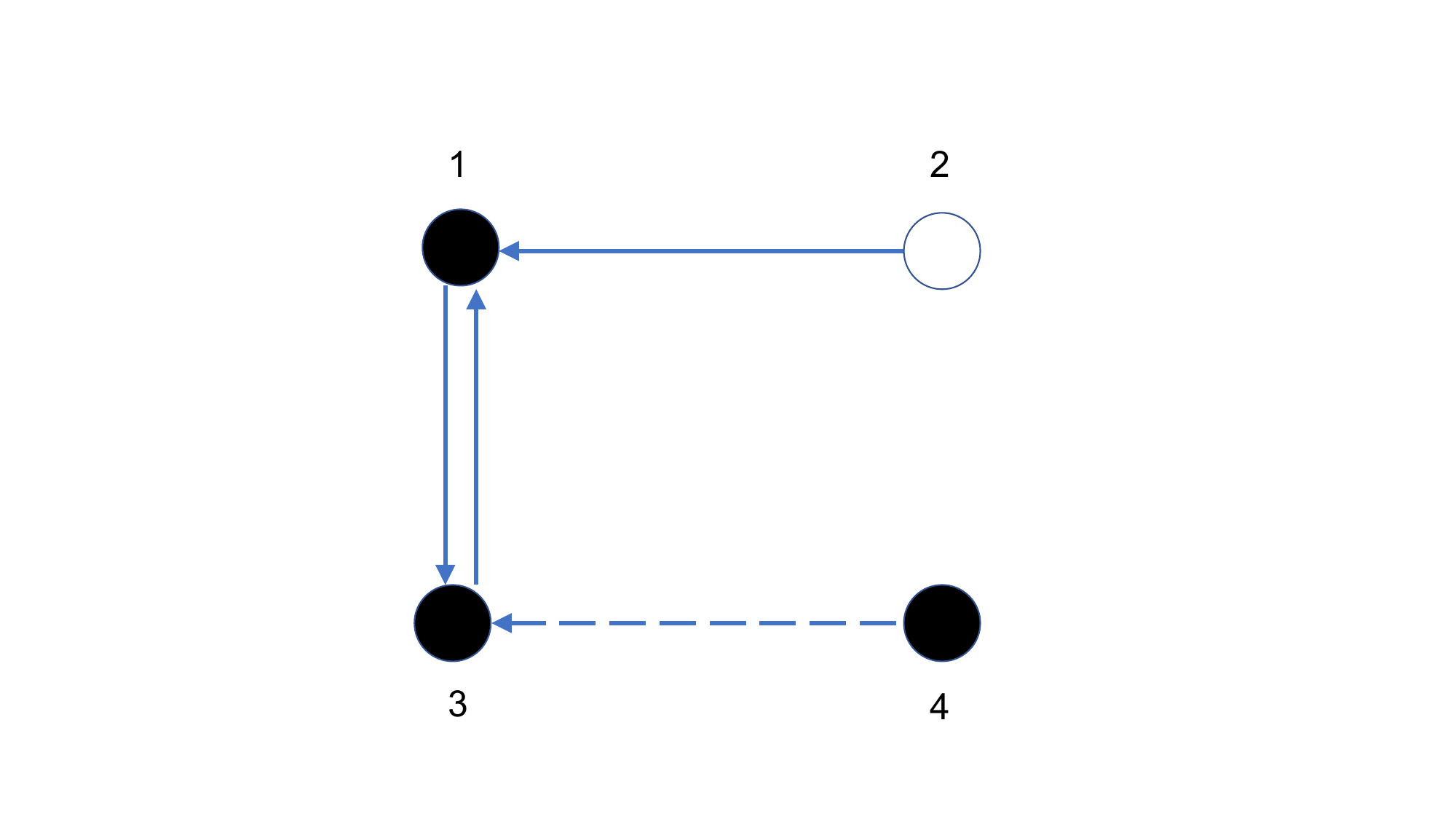}

\caption{Network with four firms and $k=3$. Black circles are firms that discover ideas while white circles do not discover ideas. Arrows are in the direction of information flow, and dashed arrows indicate direct learning while solid arrows indicate indirect learning. Firm $1$ now learns indirectly from firm $3$, unlike in Figure~\ref{fig:example1}. The only technology produced is $t=\{1,3,4\}$, and there are no profits because firms $1$ and $3$ both produce $t$.}\label{fig:example2}
\end{center}
\end{figure}

Suppose instead that firm $1$ also learns indirectly through firm $3$, as shown in Figure \ref{fig:example2}. Then we have
$$I_1(\mathbf{p},\mathbf{q})= \{3,4\}, I_2(\mathbf{p},\mathbf{q}) = \emptyset, I_3(\mathbf{p},\mathbf{q})=\{1,4\},I_4(\mathbf{p},\mathbf{q})=\emptyset.$$
The only potential technology remains $t=\{1,3,4\}$. We now have $t \subset I_i(\mathbf{p},\mathbf{q}) \cup \{i\}$ for both firm $1$ and firm $3$, so both receive the competitive profit of zero for that technology. There are also no profits from other technologies.

\subsection{Interpretation and Discussion}\label{sec:discussion}

Before expressing expected payoffs of firms and defining equilibrium, we discuss interpretation and assumptions in the model.

\textbf{Actions:} Firm actions are choices $(p_i,q_i)$. The first component $p_i$ corresponds to a level of investment in R\&D. A small probability of a discovery is cheap, while probabilities close to one are very expensive.

The second component $q_i$ corresponds to a level of openness or secrecy in interactions with other firms. As one example, consider a technology company's decision about whether to locate in an area with many other technology firms. Locating near other firms will lead to more casual interactions between the employees of the firm making the choice and employees of other firms (e.g., at bars and restaurants), and information can be shared in either direction in these interactions.\footnote{A large body of work describes the role of geographical proximity in driving innovation (e.g., \citealp*{storper2004buzz}). \cite*{kelly2009technological} considers the impact of geography on information-sharing and describes a phase transition in exogenously determined networks; our approach differs in that we model firm decisions including factors such as location as an endogenous choice.} In addition to a firm's choice of location, the action $q_i$ could include decisions such as  whether to send employees to conferences and how much disclose ongoing R\&D to employees.

An important feature of the model is that increasing $q_i$ increases the probability that firm $i$ learns from other firms but also increases the probability that other firms learn from $i$.\footnote{\cite{stein2008conversations} gives a microfoundation for bilateral communication in the context of innovation.} The baseline model assumes that learning probabilities are symmetric: firm $i$ learns from firm $j$ with the same probability that firm $j$ learns from firm $i$. In Section~\ref{sec:secrecy}, we allow firms to have heterogeneous propensities to learn across firms and find this symmetric structure does not drive results.

The downside to interaction for a firm $i$ is the increased probability of outgoing links, not an exogenous link-formation cost. Because the costs of links are an endogenous feature of the model, our equilibrium characterization does not depend on functional forms of costs, as it would with exogenous link costs separate from the innovation process.\footnote{\cite*{acemoglu2017privacy} consider a similar link cost in a setting where the benefits depend only on direct connections and the network-formation game is deterministic. Outgoing links are undesirable in their model because of a primitive preference for privacy.}

\textbf{Formal and Informal Interactions:} The model is meant to primarily describe informal interactions between employees or firms, rather than more formal arrangements such as licensing agreements or joint R\&D ventures. As such, our results are most applicable to industries where formal property rights are imperfectly enforced (Section~\ref{sec:patent} discusses the interplay between informal interactions and more formal property rights). Because information transmitted via informal interactions can often spread several steps, an analysis considering global network structure is particularly relevant.

In Appendix~\ref{sec:size}, we compare the payoffs to firms with different numbers of private ideas. This analysis can also be interpreted as measuring the value of formal contracting arrangements allowing multiple firms to share ideas frictionlessly. We find that as the number of firms grows large, the benefits to such an arrangement are a vanishing fraction of a firm's expected profits.

\textbf{Interaction Rate:} The multiplicative interaction rate $\iota(q_i,q_j)=q_iq_j$ has the feature that firm $i$'s probability of learning from another firm  and that firm's probability of learning from $i$ are both proportional to $q_i$. Thus, this is the (unique up to rescaling) interaction rate that arises from a random matching process in which all agents choose a search intensity and the probability of learning in each direction is proportional to that intensity.

We will show in Section~\ref{sec:equilibrium} that under a symmetry assumption, the main result does not reply on the multiplicative functional form, and extends to any $\widetilde{\iota}:[0,1]\times[0,1] \rightarrow [0,1]$ satisfying several mild properties.

\textbf{Learning Network:} A useful assumption is that if firm $i$ learns indirectly through firm $j$, then firm $i$ learns \emph{all} ideas known to $j$. This ensures that there is a well-defined learning network, and this network is a central object in our analysis. If indirect learning were not perfectly correlated across ideas, there would be a separate learning network for each idea.
%

\textbf{Firm Profits:} The positive payoffs from producing proprietary technologies correspond to monopoly payoffs, which we normalize to $1$. Formally, all technologies give the same monopoly profits and these profits are deterministic. It would be equivalent to take monopoly profits to be randomly drawn from any distribution with finite mean, as long as firms have no information about the realizations a priori. For example, only a small constant fraction of technologies could actually be profitable enough to produce.

 If multiple firms know all ideas contained in $t$, then there is a competitive market and firms receive zero profits. This baseline payoff structure, which we generalize in Section \ref{sec:competition}, corresponds to Bertrand competition.

Our setup requires that monopolist firms must have privately developed one of the ideas in a technology to produce that technology, but competitors need not. To start a new market, some expertise and/or confidence in the quality of the relevant idea is needed. Once a market exists, however, entrants do not require this expertise, perhaps because relevant details can be obtained from the competitor's technology.


\section{Equilibrium}\label{sec:equilibrium}

In this section, we define and then characterize equilibrium in our model. The characterization first briefly describes investment equilibria under direct learning ($\delta=0$). The remainder of the section shows that equilibria are at a critical threshold under indirect learning ($\delta>0$) and discusses implications for innovation and policy.

Two assumptions that simplify the initial analysis are that firms are homogeneous (which we will relax in several ways, including Section \ref{sec:secrecy}) and that profits are equal to the number of proprietary technologies (which we relax in Section~\ref{sec:generalprofits}).

\subsection{Solution Concept}

We begin by defining our solution concept:
\begin{defin}
An \textbf{equilibrium} $(\mathbf{p}^*, \mathbf{q}^*)$ is a pure-strategy Nash equilibrium. An equilibrium $(\mathbf{p}^*,\mathbf{q}^*)$ is an \textbf{investment equilibrium} if $p_i^* > 0$ for all $i$.
\end{defin}

Because all choices $p_i$ and $q_i$ are probabilities of discoveries or interactions, we restrict to pure strategies.

If $p_i=0$ for all $i$, then any $\mathbf{q}$ will give an equilibrium: if no other firms are investing, there is no reason to invest and so payoffs are zero. It is easy to see these trivial equilibria always exist, and we will focus on investment equilibria.

For some of our results, it will also be useful to make the stronger assumption that private investment is non-vanishing asymptotically. We consider a sequence of equilibria as the number of firms $n \rightarrow \infty$.
\begin{defin}
A sequence of equilibria $(\mathbf{p}^*, \mathbf{q}^*)$ has \textbf{non-vanishing investment} if $$\liminf_n \min_i p_i^* > 0.$$
\end{defin}
Depending on $c(\cdot)$, there may be equilibria at which all firms choose very low levels of private investment because others are investing very little. The definition excludes these partial coordination failures as well.

\subsection{Direct Learning}\label{sec:directbrief}

We briefly summarize results with $\delta=0$ here, and give a full analysis in Appendix \ref{sec:direct}. In this case, ideas can spread at most one step.

There exists a symmetric investment equilibrium for $n$ large, and at any sequence of symmetric investment equilibria the interaction rate is
$$\iota(q^*,q^*) \approx \left( \frac{k-1}{n} \right)^{\frac1k}.$$
Since the interaction rate is of order $n^{-\frac1k}$, the probability that a generic firm knows all the ideas in a given technology is of order $\frac{1}{n}$. It follows that the probability that there exists competition on a given technology is constant.

For $n$ large, each firm learns from a large number of other firms with high probability. We will see that interaction rates are much lower in the indirect-learning case. With only direct learning much more interaction is needed to generate a substantial risk of competition, so the interaction rate must be higher for potential competition to meaningfully deter openness.

%

\subsection{Main Result}

Our main focus is the indirect learning case ($\delta>0$) in which ideas can spread multiple steps. We now show that when $\delta>0$, equilibrium networks are at the critical threshold asymptotically.

We begin by defining this critical threshold. Let the number of firms $n\rightarrow \infty$ and consider outcomes under a sequence of symmetric actions. We say that an event occurs a.a.s. (asymptotically almost surely) if the probability of this event converges to $1$ as $n \rightarrow \infty$. To simplify notation, we often omit the index $n$ (e.g., from the actions $(p_i,q_i)$.)

\begin{defin}
A sequence of symmetric actions with openness $\mathbf{q}$ is:
\begin{itemize}
\item \textbf{Subcritical} if $\limsup_n \iota({q},{q}) \delta n < 1$
\item \textbf{Critical} if $\lim_n\iota({q},{q})  \delta n = 1$
\item \textbf{Supercritical} if $\liminf_n \iota({q},{q}) \delta n > 1$
\end{itemize}
\end{defin}

The expected number of firms with links to $i$ in the indirect-learning network is $$\iota(q,q)\delta (n-1),$$ so the three cases distinguish networks where each firm learns indirectly less than once, approximately once, and more than once in expectation. In the subcritical case, it follows that the expected number of firms that learn a given idea is a finite constant. In the supercritical case, there is a positive probability that a given idea is learned by a large number of firms (i.e., a number growing linearly in $n$).

This intuition is formalized by results from the theory of random directed graphs (\citealp*{karp1990transitive} and \citealp*{luczak1990phase}). Adapting their results to this setting, we have the following result.  A \textbf{component} of a directed network  is a strongly-connected component, i.e., a maximal set of nodes such that there is a path from any node in the set to any other.

\begin{lemma}[Theorem 1 of \cite*{luczak1990phase}]\label{lem:luczak} Suppose $\mathbf{q}$ is symmetric.

(i) If the indirect-learning network is subcritical, then a.a.s. every component has size $O(\log n)$.

(ii) If the indirect-learning network is supercritical, then a.a.s. there is a unique component of size at least $\widetilde{\alpha} n$ for a constant $\widetilde{\alpha} \in (0,1)$ depending on $\lim_n \iota(q,q) \delta n$, and all other components have size $O(\log n)$.
\end{lemma}

These asymptotic results each imply that large finite graphs have the component structures described with high probability. It follows from the lemma that in a subcritical sequence of equilibria, all firms learn at most $O(\log n)$ ideas a.a.s. In a supercritical sequence of equilibria, there is a positive fraction of firms learning a constant fraction of all ideas a.a.s. At a critical equilibrium, the number of ideas learned lies between the subcritical and supercritical cases.

To discuss asymmetric strategies and later heterogeneity in firms, we now generalize the notion of criticality to arbitrary strategies. Consider the matrix $(\iota(q_i,q_j)\delta )_{ij}$. The entry $(i,j)$ is equal to the probability that firm $i$ learns indirectly from firm $j$. Let $\lambda$ be the spectral radius of this matrix, i.e., the largest eigenvalue.

\begin{defin}\label{def:crit}
A sequence of actions with openness $\mathbf{q}$ is:
\begin{itemize}
\item \textbf{Subcritical} if $\limsup_n \lambda < 1$
\item \textbf{Critical} if $\lim_n\lambda = 1$
\item \textbf{Supercritical} if $\liminf_n \lambda > 1$
\end{itemize}
\end{defin}

We will see that, as in Lemma~\ref{lem:luczak}, the critical threshold corresponds to the emergence of a giant component. To show this, we will combine the results of \cite*{bloznelis2012birth}  with analysis of multi-type branching processes.

Our existence result establishes that there are equilibria with non-zero investment and communication. Our characterization result shows that asymptotically, equilibrium is on the threshold between sparse and dense networks:
\begin{thrm}\label{critical}
For $n$ sufficiently large, there exists a symmetric investment equilibrium. Any sequence of investment equilibria is critical.
\end{thrm}

Theorem~\ref{critical} makes a sharp prediction about equilibrium. The theorem assumes symmetric firms and payoffs that are linear in the number of monopoly technologies produced by a firm. We will show that equilibrium remains critical with heterogeneity in firms (Section \ref{sec:secrecy}) and non-linear payoffs (Section \ref{sec:concavity}), but does depend on the structure of competition (Section \ref{sec:competition}).

At a sequence of symmetric investment equilibria, the theorem implies that $\iota(q^*,q^*) \rightarrow \frac{1}{\delta n},$ and in particular symmetric investment equilibria are asymptotically unique. While it is easy to see there must exist a symmetric equilibrium, a priori there need not be an equilibrium with non-zero interaction and investment. In fact, we show that there exists a sequence of symmetric equilibria with non-vanishing investment. The existence result relies on analysis of firms' best responses in each region to show a fixed point theorem applies. 

We are able to drop the assumption of symmetric strategies, which is standard in settings involving random networks (e.g., \citealp*{currarini2009economic},  \citealp*{golub2010strategic}, and \citealp*{sadler2020diffusion}), and show any equilibrium is at the critical threshold. Asymmetric equilibria could feature firms with $\iota(q_i^*,q_i^*)$ above and below $\frac{1}{\delta n}$.

The proof of Theorem~\ref{critical} builds on existing mathematical results on large random graphs, and generalizes them to allow complementarities between ideas and endogenous link probabilities. The first obstacle to applying existing results is that the combinatorial structure of technologies generates complementarities between ideas, so payoffs and incentives do not simply depend on the expected number of ideas learned. A second issue is that link probabilities are endogenous, so lower-order terms in link probabilities and vanishing-probability events can matter asymptotically. We now discuss the key ideas in the proof, including how we address these challenges.

\begin{proof}[Proof Intuition] We describe the basic idea of the proof in the case $\delta=1$, and the general argument is similar. We also begin by discussing symmetric strategies.

The first-order condition for the action $q_i$ says that at any best response, the expected cost to firm $i$ of allowing a firm $j$ to learn from $i$ is equal to the expected benefit from learning from an additional firm $j$. A key feature is that the cost and benefit both depend on the distribution of the number of ideas learned from a given link. The proof exploits this symmetry between costs and benefits to solve for $q^*n$. We are able to do so because of the endogenous downside to outgoing links, which depends on the number of ideas that firm $i$ learns.
 
We use the first-order condition at a symmetric equilibrium to obtain an expression for $q^*n$ in terms of the number of incoming links used to learn the ideas in an average proprietary technology. Consider a technology $t$ such that $i$ produces $t$ and gets monopoly profits. This technology is a combination of ideas learned from different links. For example, if $k=4$, an example technology could consist of $i$'s private idea, two ideas learned indirectly from firm $j$, and one idea learned directly from firm $j''$. In this example, the technology would combine ideas from three different links.

More generally behavior will depend on the number of links utilized in learning the ideas in a technology $t$. We refer to this number of links as $\tau(t)$, so that $\tau(t) = 3$ in the example in the previous paragraph. The key tool, which we state in the subcritical region, is:
\begin{lemma}\label{lem:tauFOC}
Along any sequence of symmetric investment equilibria with $\limsup \delta \iota(q^*,q^*)n < 1$,
$$\delta \iota(q^*,q^*)n \sim \mathbb{E}_{t \in PT_i(\mathbf{p}^*,\mathbf{q}^*)}[\tau(t)]$$
for all $i$.\footnote{The distribution over technologies $t$ is defined in Appendix~\ref{sec:mainproof}.}
\end{lemma}
Lemma~\ref{lem:tauFOC} says that the expected number of other firms from whom $i$ learns is equal to the expected value of $\tau(t)$ for a random proprietary technology $t$. We give a brief intuition for the lemma. If $\tau(t)$ is higher, then there are stronger complementarities between links, because produced technologies combine ideas from more links. In this case, if a firm has a few existing links, an additional link will be more valuable than an existing link due to these complementarities. Since additional links are relatively more valuable, firms are willing to interact more.

Since $\tau(t)$ is always at least one, Lemma~\ref{lem:tauFOC} implies that $\lim_n \iota(q^*,q^*)n \geq 1,$ so there cannot be a subcritical equilibrium. 

In the supercritical region, almost all proprietary technologies $t\in PT_i(\mathbf{p}^*,\mathbf{q}^*)$ are created by combining a private idea with $(k-1)$ ideas learned from observing the giant component. In particular, payoffs are determined up to lower order terms by whether firm $i$ has a link that provides a connection to the giant component. Given such a link, additional links add little value. Thus there are not complementarities between links; indeed, links are substitutes due to the potential redundancies.

But because firms have more to lose from an outgoing link in the supercritical region, complementarities between links are needed to sustain high interaction rates. Since these complementarities are not present, there is not a supercritical equilibrium either. We check this intuition formally by a computation.

Extending results to asymmetric equilibria presents several additional technical obstacles. One is that existing mathematical results, e.g., \cite*{bloznelis2012birth}, prove the component structure has certain properties asymptotically almost surely. But this does not remove the possibility that vanishing-probability events distort incentives in an unknown direction. To rule this out, we show that an arbitrary subcritical sequence of equilibria, the asymptotic probability $\lim_{n \rightarrow \infty} \mathbb{P}[|I_i(\mathbf{p},\mathbf{q})|=y]$ that firm $i$ learns $y$ ideas decays exponentially in $y$. The proof bounds $|I_i(\mathbf{p},\mathbf{q})|$ above with the number of nodes in a multi-type Poisson branching process and then analyzes this branching process.
\end{proof}

The analysis of the symmetric equilibria extends to more general interaction rates. The proof of Theorem~\ref{critical} also shows:\begin{thrmbis}{critical}\label{thm:interaction}
Suppose the interaction rate is a strictly increasing and continuously differentiable function $\widetilde{\iota}:[0,1]\times [0,1] \rightarrow [0,1]$ 
satisfying $\widetilde{\iota}(q,q') = \widetilde{\iota}(q',q)$ for all $q$ and $q'$ and $\widetilde{\iota}(q,0) = 0$ for all $q$.\footnote{The same result holds, with minor modifications to the proof, for the additive interaction rate $\widetilde{\iota}(q_i,q_j)=q_i+q_j$. We assume $\widetilde{\iota}(q,0)=0$ to avoid a technicality: for the additive interaction rate, the first-order condition for $q_i$ may not have a solution if other firms choose sufficiently large $q_j$.} Then for $n$ sufficiently large, there exists a symmetric investment equilibrium. Any sequence of symmetric investment equilibria is critical.
\end{thrmbis}

\subsection{Discovery Rate and Policy Implications}

We next discuss consequences of Theorem~\ref{critical} for innovation and welfare. There are large gains to exogenously increasing interaction, but designing policies to realize these gains is subtle.

We first define a measure of the innovation rate, which is the fraction of possible technologies that are produced. Recall that $T_i(\mathbf{p},\mathbf{q})$ is the set of technologies known to firm $i$, so the set of potential technologies that are produced is the union
$\bigcup_i T_i(\mathbf{p},\mathbf{q}).$ This union is a subset of the $\binom{n}{k}$ possible technologies.
\begin{defin}
Given actions $(\mathbf{p},\mathbf{q})$, the \textbf{discovery rate} 
$$D(\mathbf{p},\mathbf{q})= \frac{\mathbb{E}[|\bigcup_i T_i(\mathbf{p},\mathbf{q})|]}{\binom{n}{k}}$$
is the expected fraction of potential technologies that are produced.
\end{defin}

Given any subcritical or critical sequence of actions, each firm learns $o(n)$ ideas asymptotically almost surely, so the discovery rate converges to zero. Along any sequence of supercritical actions, there are a positive fraction of firms learning some fraction $\alpha$ of ideas, so the discovery rate is non-vanishing.

Theorem~\ref{critical} therefore lets us characterize the discovery rate at and near equilibrium:
\begin{cor}\label{cor:open}
Let $(\mathbf{p}^*,\mathbf{q}^*)$ be a sequence of equilibria with non-vanishing investment. Then the discovery rate  vanishes along this sequence: $\lim_n D(\mathbf{p}^*,\mathbf{q}^*)=0$. For any $\epsilon>0$, the discovery rate with openness $(\mathbf{p}^*,(1+\epsilon)\mathbf{q}^*)$  is non-vanishing: $\liminf_n D(\mathbf{p}^*,(1+\epsilon)\mathbf{q}^*) >0$.
\end{cor}

The discovery rate is vanishing under equilibrium interaction patterns, but would be non-vanishing with slightly more interaction. An immediate consequence is that increasing openness by \emph{any} multiplicative factor has a very large effect on payoffs asymptotically:
$$\lim_n \frac{D(\mathbf{p}^*,(1+\epsilon)\mathbf{q}^*)}{{D}(\mathbf{p}^*,\mathbf{q}^*)} =\infty.$$

Corollary~\ref{cor:open} relates to \cite*{saxenian1996regional}'s study of the Route 128 and Silicon Valley technology industries, which found that Silicon Valley had much more open firms and grew faster. In the terminology of our model, Route 128's secrecy corresponds to equilibrium behavior. But institutional features of Silicon Valley (including non-enforcement of non-compete clauses and common ownership of firms by venture capital firms) may have constrained firms' actions to prevent high levels of secrecy (subcritical or critical choices of $q_i$). Such constraints would imply a much higher innovation rate.

A natural question is whether these gains can be realized via policy interventions other than directly restricting firms' strategy spaces. One common policy approach is to subsidize private research and development. We find that increasing R\&D spending has relatively little effect on its own, but can be effective given a collaborative culture:
\begin{cor}\label{cor:investment}
Let $(\mathbf{p}^*,\mathbf{q}^*)$ be a sequence of equilibria with non-vanishing investment. Then the derivative of the equilibrium discovery rate in private investment vanishes along this sequence:   $\lim_n \frac{\partial D(\mathbf{p}^*+x\mathbf{1},\mathbf{q}^*)}{\partial x}(0)=0$. For any $\epsilon>0$, the derivative of the discovery rate with openness $(1+\epsilon)\mathbf{q}^*$ in private investment is non-vanishing: $\liminf_n \frac{\partial D(\mathbf{p}^*+x\mathbf{1},(1+\epsilon)\mathbf{q}^*)}{\partial x}(0) >0$.
\end{cor}
Under equilibrium interaction patterns, a higher level of R\&D does not increase the discovery rate. But if openness is already above equilibrium levels, then increasing R\&D will have a much larger impact on the number of ideas discovered.

To summarize the implications of Corollaries~\ref{cor:open} and~\ref{cor:investment}, at or near equilibrium outcomes, there are large gains to policies (e.g., non-enforcement of non-compete clauses, establishing innovation clusters) that encourage or require more interaction between firms and thus shift outcomes to the supercritical region. Policies to increase private investment (e.g., subsidies for R\&D), however, will not shift outcomes to the supercritical region, and thus have much smaller benefits at equilibrium. But once outcomes are in the supercritical region, policies to increase private investment will have large benefits.

\begin{table} 
\begin{center}
\begin{tabular}{|c|c|c|c|}
 \hline
  &  \textbf{Subcritical} & \textbf{Critical}  & \textbf{Supercritical}  \\ 
  \hline

\textbf{Best Response $q_i$} & High & Intermediate & Low \\ 
 \hline

\textbf{Discovery Rate} & Vanishing & Vanishing & Non-vanishing \\ 
 \hline
 
\textbf{Increasing $q$} & Large Benefit & Large Benefit &  Ambiguous \\ 
 \hline
 
\textbf{Increasing $p$} & Small Benefit & Intermediate & Large Benefit  \\ 
\hline
\end{tabular}
\caption{Best responses and policy implications when firms choose symmetric strategies $(\mathbf{p},\mathbf{q})$ in the subcritical, critical, and supercritical regions.}\label{fig:phasetable}
\end{center}

\end{table}

\subsection{Public Innovators}\label{sec:public}

Corollary~\ref{cor:open} showed there are large gains to increasing interaction rates above equilibrium levels. This section shows that these gains can be realized via targeted interventions that change interaction patterns. 

We now show that introducing \textbf{public innovators} who are not concerned with secrecy leads to learning and innovation at the same rate as in the supercritical region. In particular, there exists a giant component of the learning network containing these public innovators. Public innovators could correspond to academics, government researchers,  open-source software developers, or other researchers with incentives or motivations other than profiting from producing and selling technologies.

A public innovator $i$ pays investment cost $c(p_i)$ and receives a payoff of one for each technology $t$ such that: (1) $i \in t$ and (2) $j \in \{i\} \cup I_i(\mathbf{p},\mathbf{q}) $ for all $j \in t$. We will rely on the fact that for public innovators there is no downside to interactions, but not on the exact incentive structure.

All firms have the same incentives as in the baseline model, and public innovators and firms interact as in the baseline model. We now call an equilibrium symmetric if all public innovators choose the same action and the same holds for all private firms.

\begin{proposition}\label{prop:public}
Suppose a non-vanishing share of agents are public innovators. Then there exists a sequence of symmetric equilibria with non-vanishing investment, and at any sequence of equilibria with non-vanishing investment the discovery rate is non-vanishing: $\liminf_n D(\mathbf{p^*},\mathbf{q}^*)>0$.
\end{proposition}

The proposition says that at equilibrium, a positive fraction of possible technologies are discovered. The ratio between the equilibrium discovery rates with and without public innovators grows unboundedly large as $n\rightarrow \infty$. The proof shows that a giant component forms around the public innovators. This holds for any positive share of public innovators, and indeed could be extended to a slowly vanishing share of public innovators.

Proposition~\ref{prop:public} assumes that firms cannot direct interactions toward public innovators or private firms. In Appendix~\ref{sec:publicdirected}, we show the same result holds when interactions can be directed toward public innovators or private firms. Because public innovators are more likely to be in the giant component, private firms are willing to interact with them.

Public innovators are valuable primarily as informational intermediaries rather than for their private ideas. Because public innovators do not face costs to interaction, they will choose $q_i=1$ at equilibrium. Therefore, public innovators can learn many ideas via interactions and transmit these ideas to other public innovators or to private firms (e.g, academics learning ideas from conferences and collaborations and then consulting for private industry). Conversely, the proposition would remain unchanged if all public innovators instead choose $p_i=0$ and $q_i=1$.

Empirical research on collaboration between academia and industry supports the value of academic researchers as informational intermediaries between firms. \cite*{azoulay2012diffusion} study movement of star academics, and find that moves increase patent-to-patent and patent-to-article citations locally. Federal funding for universities is also linked to higher entrepreneurship locally \citep*{tartari2021more}.

\subsection{Welfare}

The preceding analysis focused on the discovery rate as an outcome measure. We now extend the results to a more general measure of welfare that allows for consumer and producer surplus. We take a reduced-form approach, assuming fixed consumer surplus from each  technology sold by a monopolist and each technology sold in a competitive market.

For each $i$, we define the \textbf{competitive technologies} $CT_i(\mathbf{p},\mathbf{q})\subset T_i(\mathbf{p},\mathbf{q})$ for $i$ to be the set of technologies $t \in T_i(\mathbf{p},\mathbf{q})$ such that at least one other firm knows all ideas in $t$. Thus each technology $t\in T_i(\mathbf{p},\mathbf{q})$ is either proprietary or competitive.

We define $$W(\mathbf{p},\mathbf{q}) =( w_{PT} +1)\cdot  \left|\bigcup_i PT_i(\mathbf{p},\mathbf{q})\right| + w_{CT} \cdot \left|\bigcup_i CT_i(\mathbf{p},\mathbf{q})\right| -\sum_{i=1}^n c(p_i),$$
where the weights $w_{PT}$ and $w_{CT}$ are positive. Welfare is the sum of producer surplus (monopoly profits minus R\&D costs) and  consumer surplus ($w_{PT}$ from each proprietary technology and $w_{CT}$ from each competitive technology).

\begin{proposition}\label{prop:welfare}
Fix a sequence of actions $(\mathbf{p},\mathbf{q})$ with $\liminf_n \min_i p_i>0$ and non-negative expected profits for each firm $i$. The following are equivalent:

(1) There exists  $\widetilde{\alpha} \in (0,1)$ such that there is a unique component of the indirect-learning network of size at least $\widetilde{\alpha} n$ a.a.s.,

(2) The discovery rate is non-vanishing: $\liminf_n D(\mathbf{p},\mathbf{q})>0$, and

(3) There exists $C>0$ such that $\mathbb{E}[W(\mathbf{p},\mathbf{q})] > Cn^{k-1}$ for $n$ sufficiently large.

\end{proposition}
When there is a giant component, a positive fraction of ideas are discovered and welfare is within a constant factor of the social optimum. In the subcritical and critical regions, a vanishing share of ideas are discovered and welfare is only a vanishing share of the socially optimal level. The basic idea behind the proof is that consumer surplus grows at the same rate as the number of technologies discovered.

An immediate consequence is that Corollaries~\ref{cor:open} and~\ref{cor:investment} and Proposition~\ref{prop:public} imply similar statements about any welfare measure $W(\mathbf{p},\mathbf{q})$. For example, equilibrium  welfare $W(\mathbf{p}^*,\mathbf{q}^*)$ grows at rate $o(n^{k-1})$ in the baseline model but grows at a rate proportional to $n^{k-1}$ with a positive share of public innovators. Therefore, the ratio between equilibrium welfare with and without public innovators grows unboundedly large as $n\rightarrow \infty$.

\section{Asymmetric Learning Probabilities}\label{sec:secrecy}

The baseline model assumes that information flows are symmetric across pairs of firms. In practice, firms may have heterogeneous probabilities of learning from others, even given a fixed interaction rate. We next show that equilibrium remains critical with heterogeneous propensities to learn.

Suppose that firms have propensities to learn $\beta_i \in (0,1)$, where there are finitely many propensities to learn and the fraction of firms with each propensity $\beta_i$ converges as the number of firms grows large. Firm $i$ now directly learns from firm $j$ with probability
$$\beta_i \iota(q_i,q_j).$$ Learning otherwise occurs as in the baseline model, including indirect learning.

It is straightforward to extend Definition~\ref{def:crit} to allow heterogeneous secrecy. We now let $\lambda$ be the spectral radius of the matrix $(\beta_i\iota(q_i,q_j)\delta )_{ij}$. As before entry $(i,j)$ is equal to the probability that firm $i$ learns indirectly from firm $j$. Given the modified definition of $\lambda$, Definition~\ref{def:crit} continues to define a critical threshold and this threshold again corresponds to the emergence of a giant component.

\begin{thrm}\label{secrecy}
Suppose firms have propensities to learn $\boldsymbol{\beta}$. There exists an investment equilibrium for $n$ large, and any sequence of investment equilibria is critical.
\end{thrm}

Equilibria remain critical even when the directed link probabilities are asymmetric across pairs. The characterization result extends immediately to the case in which $\beta_i$ are chosen endogenously at a cost $\widetilde{c}_i(\beta_i)$, which can vary across firms.\footnote{This choice can be made simultaneously with or prior to the choice of $q_i$.} In this case, firms can now control the likelihood of learning along two dimensions. First, higher interaction rates allow a firm to learn more from from others at the expense of a higher probability of its ideas leaking. Second, firms can pay an exogenous cost to increase the probability of learning from others at a given interaction rate, and some firms may be able to do so more cheaply than others.

The proof of Theorem~\ref{secrecy} shows that decisions with asymmetric learning probabilities are similar to decisions in the baseline model. Recall that at equilibrium, the first-order condition for the openness $q_i$ relates the value of the ideas already known to firm $i$ with the value of increasing the interaction rate. Fixing $q_i$, a higher $\beta_i$ increases  both sides of this first-order condition because firms with higher propensities to learn have already learned more ideas but also will learn more from an additional interaction.

At potential equilibria in the subcritical region, these two forces cancel out and a firm's optimal choice of openness $q^*_i$ is approximately independent of that firm's propensity to learn $\beta_i$. The proof of Theorem~\ref{critical} used the symmetry between the value of existing links and an additional link to solve for $q^*n$. Even if each of these links only realizes with probability $\beta_i$, the symmetry persists and so $q^*n$ is unchanged.

The two opposing effects would not entirely cancel in the supercritical region, because there are potential redundancies between multiple links to the giant component. These redundancies matter more for firms with higher $\beta_i$. Nevertheless, we can bound the average interaction rate when all firms choose $q_i$ to respond optimally to the giant component size.

\section{Benefits and Costs of Links}\label{sec:generalprofits}

In Sections \ref{sec:model} and \ref{sec:equilibrium}, we studied equilibrium when expected payoffs were
$$U_i(\mathbf{p},\mathbf{q})=\mathbb{E}\left[|PT_i(\mathbf{p},\mathbf{q})|\right] -c(p_i).$$
Firms' utility functions had two properties: 
\begin{enumerate}
\item Payoffs are linear in the number of proprietary technologies, and
\item Payoffs do not depend on technologies for which the firm faces competition.
\end{enumerate}

The first assumption determines the benefits from incoming links, while the second determines the costs of outgoing links.

We now relax each of these assumptions. We find that equilibrium remains critical when the returns to producing more technologies are increasing. More generally, we show that equilibrium is critical in a setting where profits are a convex function of the number of ideas learned. Changing the profit structure in competitive markets, however, leads to supercritical or subcritical equilibria. So outcomes depend on the specification of the costs of outgoing links, but are less sensitive to the specification of the gains from learning.

\subsection{Concavity of Profits}\label{sec:concavity}

The baseline model assumed that a firm's profits are linear in the number of proprietary technologies. We now extend the main result to allow more general firm profits.

In practice, there may be increasing or decreasing returns to producing more technologies. Suppose that the payoffs to firm $i$ are instead
$$|PT_i(\mathbf{p},\mathbf{q})|^{\rho} - c(p_i),$$
where $\rho > 0$.

The baseline model is the case $\rho=1$. When $\rho>1$, there are increasing returns to controlling more monopolies. When $\rho < 1$, there are decreasing returns to controlling more monopolies. Note that these increasing or decreasing returns to scale are not determined by the innovative process, but rather by production costs or other market conditions.

\begin{proposition}\label{prop:concavity}
There exists $\underline{\rho} \leq 1$ such that for any $\rho \geq  \underline{\rho}$, any sequence of symmetric investment equilibria is critical. When $k>2$, we have $\underline{\rho} < 1$.
\end{proposition}

The proposition shows that the prediction of critical equilibria is not knife-edge with respect to $\rho$. In particular, increasing returns to scale cannot move interactions above the critical threshold. As long as $k\neq 2$, slightly decreasing returns to scale will not move interactions below the critical threshold either.

Consider a firm $i$ that does not face competition. We show that under the conditions of the proposition, the firm's profits are convex in $|I_i(\mathbf{p},\mathbf{q})|$. As a result, learning additional ideas is more appealing relative to protecting existing ideas, so openness will not decrease below the critical region. Checking convexity is delicate when $\rho < 1$, because in this case firm profits are the composition of the binomial coefficient $\binom{|I_i(\mathbf{p},\mathbf{q})|}{k-1}$, which is convex, and the polynomial, $|PT_i(\mathbf{p},\mathbf{q})|^{\rho}$, which is concave.

We also show that for any $\rho$, openness will not increase enough to push equilibrium into the supercritical region either. At a potential supercritical sequence of symmetric investment equilibria, profits are driven by the event that firm $i$ learns $\alpha n$ ideas from the giant component and produces approximately $(p^*)^k\binom{\alpha n}{k-1}$ proprietary technologies. Firm $i$ chooses $q_i$ to maximize the probability of this event. Asymptotically the optimal $q_i$ is independent of the payoffs from this event since these payoffs are very large, and therefore the optimal $q_i$ is independent of $\rho$. Given this, the calculation is the same as in the case $\rho=1$ (Theorem~\ref{critical}), where there is no supercritical sequence of investment equilibria.

More generally, the proof shows that our criticality result relies on two features of the payoff function. First, payoffs for a firm $i$ that does not face competition are convex in the number of ideas $|I_i(\mathbf{p},\mathbf{q})|$ learned by $i$. Second, payoffs grow at a polynomial rate in the number of ideas learned by $i$.

We can state this formally when $\delta=1$, so that when firm $i$ learns from $j$ it will learn all ideas known to firm $i$. In this case, we let the profits for firm $i$ be $\phi(|I_j(\mathbf{p},\mathbf{q}|)$ when firm $i$ discovers its private idea ($i \in I$) and no firm learns from $i$, and $0$ otherwise. We will assume that $\phi(\cdot)$ is strictly increasing and continuously differentiable.

\begin{proposition}\label{prop:convexgeneral}
 Suppose $\delta=1$ and payoffs when no firm learns from $i$ and $i\in I$ are equal to $\phi(|I_i(\mathbf{p},\mathbf{q}|)$, where $\phi(x)$ is convex and $\frac{\phi(x_j)}{\phi(x_j')}\rightarrow 1$ along any sequence of $(x_j,x_j')$ such that $\frac{x_j}{x_j'}\rightarrow 1$. Then any sequence of symmetric equilibria  with non-vanishing investment is critical.
\end{proposition}

The assumption that $\frac{\phi(x_j)}{\phi(x_j')}\rightarrow 1$
 along any sequence of $(x_j,x_j')$ such that $\frac{x_j}{x_j'}\rightarrow 1$ bounds the rate of growth of $\phi(\cdot)$.
 In particular, this assumption holds if $\phi(x)=Cx^d + O(x^{d-1})$
 for any $C >0$ and any real $d\geq 1$. If payoffs instead grow at an exponential rate in the number of ideas, then a supercritical equilibrium is possible because an additional idea may be very valuable (see \citealp*{acemoglu2020endogenous} for a related effect).

A special case is that firms can produce technologies of multiple complexities $k$, perhaps with different payoffs for technologies of different complexities. The growth condition will hold as long as the allowed complexities are bounded (independent of $n$). A consequence of the proposition in this case is that equilibrium profits are driven by technologies of the highest feasible complexity for $n$ large. At a critical equilibrium most profits come from rare events where a firm learns many ideas, and when this occurs the firm can produce many more technologies of higher complexities. This suggests that if firms can choose the complexity $k$ of their products, firms will produce more complex technologies (at least in large markets).

\subsection{Profits Under Competition}\label{sec:competition}

We found in Theorem~\ref{critical} that equilibrium lies on the critical threshold. This result is robust to different payoffs structures for monopolist firms. We now show that Theorem~\ref{critical} does depend on the structure of competition, and altering the payoffs from competitive markets can lead to supercritical or subcritical outcomes.

To generalize the payoff from technologies, we will now assume that firm $i$ receives payoffs $f(m)$ from a technology $t$ such that $i \in t$, firm $i$ learns all other ideas in $t$, and $m$ other firms learn all ideas in $t$. We assume $f(\cdot)$ is weakly decreasing and maintain the normalization $f(0)=1$.

A simple case is $f(m)=a<1$ for all $m > 0$. The parameter $a$ represents the profits in competitive markets for firms $i$ such that the idea $i$ is included in the technology. These profits could correspond to a first-mover advantage or a higher quality product due to the firm's expertise. For example, firms whose idea is included in a technology could bring the product to market first and obtain positive profits before competitors arrive (see \citealp*{boldrin2008perfectly} for details). Alternately, these firms could produce higher quality technologies and make positive profits due to this product differentiation.
 The analysis in previous sections corresponded to the case $a=0$.

We can also allow $f(m)<0$, which could correspond to a fixed cost of production that must be paid before competition is known. We assume that firms make a single decision about whether to produce the technologies that they learn.\footnote{If firms can condition their production decision on the flow of ideas, the analysis becomes more complicated.}

\begin{proposition}\label{prop:competition}
(i) If $0 < f(1)< 1$ and $f(m)\geq 0$ for all $m$, then any sequence of symmetric investment equilibria is supercritical.

(ii) If $f(m) < 0$ for all $m > 0$, then any sequence of symmetric investment equilibria is subcritical.
\end{proposition}

We can restate the proposition in terms of the discovery rate, and find that there is more innovation when there are positive profits under competition than in the Bertrand case:
\begin{cor}\label{cor:competition}
(i) If $0 < f(1)< 1$ and $f(m)\geq 0$ for all $m$, then the discovery rate  is non-vanishing along any sequence of symmetric equilibria with non-vanishing investment: $\liminf_n D(\mathbf{p}^*,\mathbf{q}^*) >0$.

(ii) If $f(m) < 0$ for all $m > 0$, then the discovery rate vanishes along any sequence of symmetric equilibria with non-vanishing investment: $\lim_n D(\mathbf{p}^*,\mathbf{q}^*)=0$  . 
\end{cor}

Part (i) of the proposition says that if the potential downside to enabling competitors is not as large, then firms will be more willing to interact. This pushes the equilibrium from the critical threshold into the supercritical region. Proposition~\ref{prop:competition}(i) introduces an additional force to classic debates on whether firms are more innovative in more competitive markets (see \citealp*{cohen1989empirical} for a survey). While much of this literature considers how competition changes firms' private incentives to conduct R\&D, the proposition considers its effect on interaction and learning between firms.

Part (ii) of the proposition says that increasing the costs of competition discourages interaction, and pushes the equilibrium to the subcritical region. There need not be an investment equilibrium if payoffs under competition are sufficiently negative.

The proof of (ii) is more involved, as we must characterize payoffs at a potential critical sequence of equilibria. The key lemma shows that at the critical threshold, we have $\mathbb{E}_{t \in PT_i(\mathbf{p},\mathbf{q})}[\tau(t)]\rightarrow 1$, i.e., most of a firm $i$'s proprietary technologies only include ideas learned from one other firm. To prove this lemma, we use a pair of coupling arguments to show that most profits come from rare events in which a single link (indirectly) lets a firm learn many ideas. 


Proposition~\ref{prop:competition} describes how the equilibrium learning network depends on competition qualitatively, but we can also relate equilibrium outcomes to the innovation rate within the supercritical region. We now show that increasing the payoffs from competitive outcomes will increase the innovation rate even within the supercritical region.

Consider a function $g(m)$ such that $g(m) \geq 0$ for all $m$, $g(m)$ is weakly decreasing in $m$, and $g(m)\rightarrow 0$ as $m\rightarrow \infty$. Set $f_x(0)=1$ and $f_x(m) = xg(m)$ for $m\geq 1$, so that monopoly payoffs are constant and competitive profits are weakly increasing in $x$. For example, higher values of $x$ can correspond to a larger first-mover advantage.

\begin{proposition}\label{prop:competitionquant}
There exists $x^*>0$ such that for $x \in [0,x^*]$, with payoffs $f_x(m)$:

(i) There exists an investment equilibrium for $n$ large;

(ii) The limit of the discovery rate $\lim_n D(\mathbf{p}^*,\mathbf{q}^*)$ along any sequence of symmetric equilibria with non-vanishing investment $(\mathbf{p}^*,\mathbf{q}^*)$ exists and does not depend on the sequence of equilibria;

(iii) The limit $\lim_n D(\mathbf{p}^*,\mathbf{q}^*)$ is weakly increasing in $x$.

\end{proposition}
When payoffs are higher in competitive markets, the upside to learning from others is potentially higher and the downside of information leakage is lower. So individuals are willing to interact more, and thus the discovery rate is higher, as $a$ increases. The proposition shows this comparative static locally, because for $x$ small we can leverage our analysis of equilibrium under Bertrand payoffs to show that an investment equilibrium exists and is well-behaved.

\section{Patent Rights}\label{sec:patent}

In the baseline model, technologies could only be protected via secrecy. We now consider the possibility that a positive fraction of firms receive patents on their ideas. As motivation, suppose that firms discover different types of ideas and patent law determines which types are patentable. For example, \cite*{bessen2007empirical} discuss the boundaries of patent law in the software industry and how those boundaries have changed over time.

\subsection{Equilibrium}

Suppose a fraction $b \in (0,1)$, of firms receive a \textbf{patent} on their private ideas. In this case, other firms cannot use this private idea, either as monopolists or competitors. Formally, a firm $i$ receives payoff $1$ from each technology $t$ such that (1) idea $i \in t$; (2) firm $i$ knows all $j \in t$ and no other $j \in t$ receive patents; and (3) either $i$ receives a patent or $i$ is the unique firm that knows all $j \in t$. Else the firm receives payoff $0$ from the technology $t$.

To focus on how patents relate to informal interactions, we analyze model patents in a very simple way. In particular, the model will not include imperfect patent rights and/or licensing of patents, and licensing can increase the innovation rate.

A firm now chooses either a level of openness $q_i(0)$, which is the action without a patent, or $q_i(1)$, which is the action with a patent. We will refer to the choices at symmetric equilibria as $q^*(0)$ and $q^*(1)$.

For the first part of the following result ($\delta=0$), we will also assume that each firm $i$ pays cost $\epsilon > 0$ for each realized link. The purpose of this cost is to break near-indifferences in favor of lower interaction rates. We observe in Appendix \ref{sec:direct} that without patents, a small link cost $\epsilon$ has little effect on the equilibrium.

\begin{proposition}\label{prop:patent} Suppose a fraction $b \in (0,1)$ of firms receive patents. If $\delta=0$, then  with $k=2$ and any link cost $\epsilon > 0$, there does not exist an investment equilibrium for $n$ large. If $\delta>0$, then $\iota(q^*(0),q^*(0))$ is $o(1/n)$ along any sequence of symmetric investment equilibria.
\end{proposition}

With only direct learning ($\delta=0$), the proposition says that positive investment cannot be sustained at equilibrium for $n$ large. This is because of an adverse-selection effect that discourages social interactions.

Because firms receiving patents have no need for secrecy, firms with patents choose very high interaction rates $q_i(1)$. Thus, most interactions are with firms with patents. On the other hand, firms with patents are undesirable to interact with because their ideas cannot be used by others. Because of this adverse selection in the matching process, firms without patents will have much lower expected profits than in the model without patents. When $k=2$ and there is an arbitrarily small cost to links, this has the effect of shutting down all interaction and investment.

This contrasts with our results on direct learning with no patent rights (Section \ref{sec:directbrief} and Appendix \ref{sec:direct}), where there is an investment equilibrium with substantial interaction. With $k > 2$ and patent rights, the adverse selection effect persists but no longer prevents any equilibrium investment. In this case, the interaction rate between firms without patents is much lower asymptotically than in the result with no patent rights (Appendix~\ref{sec:patentapp}).

The direct learning result also suggests a more general adverse-selection effect in strategic network formation. Suppose that agents with lower link formation costs are also less valuable partners for connections. If agents cannot discriminate in their link formation decisions, the composition of the pool of potential partners will discourage connections.

With indirect learning, firms with patents still do not provide private ideas to others, but can now serve as informational intermediaries (like the public innovators in Section~\ref{sec:public}). A giant component forms around the firms with patents. Firms without patents now have some interactions with firms without patents, who can transmit ideas from other firms without patents. Interactions between pairs of firms without patents, however, are rare.

Proposition~\ref{prop:patent} relates to literature on firms' choices between formal and informal intellectual property protections, particularly patents versus secrecy (e.g., \citealp*{anton2004little} and \citealp*{kultti2006simultaneous}). We focus not on the choice between formal and intellectual property rights but on the interplay between the two.\footnote{A second contrast is to theoretical findings on patents and follow-up innovation (e.g., \citealp*{scotchmer1991standing}, \citealp*{scotchmer1990novelty}, \citealp*{bessen2009sequential}). This literature investigates when granting patent rights for an idea decreases follow-up innovations involving that idea. In our random-interactions setting, patent rights can not only decrease follow-up innovations involving patented ideas but also decrease follow-up innovations involving \textit{other} unprotected ideas.}

The proposition assumes that firms cannot direct interaction toward or away from firms with patents, but we could also allow firms to choose separate interaction rates for peers with and without patents. Under direct learning ($\delta=0$), firms without patents would only interact with each other and the adverse selection effect would disappear.\footnote{The analysis of equilibrium interaction patterns among these firms would be the same as in Appendix~\ref{sec:direct}.} Under indirect learning, firms with patents can still play a similar role as informational intermediaries when directed interaction is permitted. The analysis is very similar to Appendix~\ref{sec:publicdirected}, so we omit details here.

\subsection{Discovery Rate}

We can use Theorem~\ref{critical} and Proposition~\ref{prop:patent} to ask when patent rights produce more innovation and what the optimal share $b^*$ of patentable ideas would be. In the direct-learning case, the proposition gives conditions under which patents decrease innovation and welfare.

In the indirect-learning case, the equilibrium discovery rate and social welfare are higher (for $n$ large) with interior patent rights $b \in (0,1)$ than without patent rights because firms with patents are valuable as intermediaries. Under indirect learning, we can ask what value of $b$ maximizes the discovery rate. Any positive $b$ provides the benefits of information intermediaries, and so there is a tradeoff between the higher private profits obtained by firms with patents and the social benefits provided by firms without patents, whose ideas can be used by others. For high $k$ the optimal value of $b$ converges to zero as $n$ grows large. 

This is easiest to see when $\delta=1$. In this case, we can compute that the share of firms without patents whose ideas are learned by the giant component is 
$e^{-bq^*(0)n} \approx \frac12.$ The number of technologies produced is
\begin{equation}\label{eq:discoverypatents} (p^*)^{k}\binom{\frac12 (1-b) n}{k-1}(b+\frac12\cdot(1-b))n+o(n^{k-1}),\end{equation}
because approximately $b+\frac12\cdot(1-b)$ firms learn from the giant component and each learns approximately $\frac12 (1-b) n$ unpatented ideas. The leading term in expression (\ref{eq:discoverypatents}) is maximized by $b=0$, so the patent share $b^*$ maximizing the discovery rate converges to zero.

\section{Conclusion}

We have studied strategic network formation in large random graphs in the context of an economic application to innovation and social learning. The model is particularly suited to analysis of informal interactions, e.g., between employees of firms, which cannot be fully governed by formal contracts. We find that in these settings, if there are many firms and ideas can travel multiple steps, the global structure of the learning network has stark consequences for incentives and payoffs. In particular, expected payoffs and welfare are much higher when there is enough interaction to support a giant component.

The techniques developed in this paper could be extended or applied in several directions. For example, we have studied a static model, but in a dynamic version of the model past events would create new incentives toward openness or secrecy. More broadly, while we have focused on a network-formation game with a tradeoff between secrecy and learning, we have developed more widely applicable tools for questions in network economics, particularly those with complementarities between connections. Outside of network formation, the same techniques can also be applied to problems such as optimal seeding for diffusion processes or search and matching models where indirect connections are important.

\setstretch{1.3}

\bibliographystyle{ecta}
\bibliography{Innovation}

\begin{thebibliography}{50}
\newcommand{\enquote}[1]{``#1''}
\expandafter\ifx\csname natexlab\endcsname\relax\def\natexlab#1{#1}\fi

\bibitem[\protect\citeauthoryear{Acemoglu and Azar}{Acemoglu and
  Azar}{2020}]{acemoglu2020endogenous}
\textsc{Acemoglu, D. and P.~D. Azar} (2020): \enquote{Endogenous production
  networks,} \emph{Econometrica}, 88, 33--82.

\bibitem[\protect\citeauthoryear{Acemoglu, Makhdoumi, Malekian, and
  Ozdaglar}{Acemoglu et~al.}{2017}]{acemoglu2017privacy}
\textsc{Acemoglu, D., A.~Makhdoumi, A.~Malekian, and A.~Ozdaglar} (2017):
  \enquote{Privacy-constrained network formation,} \emph{Games and Economic
  Behavior}, 105, 255--275.

\bibitem[\protect\citeauthoryear{Akbarpour, Malladi, and Saberi}{Akbarpour
  et~al.}{2018}]{akbarpour2018diffusion}
\textsc{Akbarpour, M., S.~Malladi, and A.~Saberi} (2018): \enquote{Diffusion,
  Seeding, and the Value of Network Information,} \emph{Working paper}.

\bibitem[\protect\citeauthoryear{Akcigit, Caicedo, Miguelez, Stantcheva, and
  Sterzi}{Akcigit et~al.}{2018}]{akcigit2018dancing}
\textsc{Akcigit, U., S.~Caicedo, E.~Miguelez, S.~Stantcheva, and V.~Sterzi}
  (2018): \enquote{Dancing with the stars: Innovation through interactions,}
  \emph{Available at SSRN 2647049}.

\bibitem[\protect\citeauthoryear{Alon and Spencer}{Alon and
  Spencer}{2004}]{alon2004probabilistic}
\textsc{Alon, N. and J.~H. Spencer} (2004): \emph{The Probabilistic Method},
  John Wiley \& Sons.

\bibitem[\protect\citeauthoryear{Anton and Yao}{Anton and
  Yao}{2004}]{anton2004little}
\textsc{Anton, J.~J. and D.~A. Yao} (2004): \enquote{Little patents and big
  secrets: managing intellectual property,} \emph{RAND Journal of Economics},
  1--22.

\bibitem[\protect\citeauthoryear{Athreya and Ney}{Athreya and
  Ney}{1972}]{athreya1972branchingprocesses}
\textsc{Athreya, K. and P.~Ney} (1972): \emph{Branching Processes},
  Springer-Verlag.

\bibitem[\protect\citeauthoryear{Azoulay, Graff~Zivin, and Sampat}{Azoulay
  et~al.}{2012}]{azoulay2012diffusion}
\textsc{Azoulay, P., J.~S. Graff~Zivin, and B.~N. Sampat} (2012): \enquote{The
  Diffusion of Scientific Knowledge across Time and Space,} \emph{The Rate and
  Direction of Inventive Activity Revisited}, 107.

\bibitem[\protect\citeauthoryear{Bala and Goyal}{Bala and
  Goyal}{2000}]{bala2000noncooperative}
\textsc{Bala, V. and S.~Goyal} (2000): \enquote{A noncooperative model of
  network formation,} \emph{Econometrica}, 68, 1181--1229.

\bibitem[\protect\citeauthoryear{Ballester, Calv{\'o}-Armengol, and
  Zenou}{Ballester et~al.}{2006}]{ballester2006who}
\textsc{Ballester, C., A.~Calv{\'o}-Armengol, and Y.~Zenou} (2006):
  \enquote{Who's who in networks. Wanted: The key player,} \emph{Econometrica},
  74, 1403--1417.

\bibitem[\protect\citeauthoryear{Baum, Cowan, and Jonard}{Baum
  et~al.}{2010}]{baum2010network}
\textsc{Baum, J.~A., R.~Cowan, and N.~Jonard} (2010):
  \enquote{Network-independent partner selection and the evolution of
  innovation networks,} \emph{Management science}, 56, 2094--2110.

\bibitem[\protect\citeauthoryear{Baumann}{Baumann}{2021}]{baumann2021model}
\textsc{Baumann, L.} (2021): \enquote{A model of weighted network formation,}
  \emph{Theoretical Economics}, 16, 1--23.

\bibitem[\protect\citeauthoryear{Berliant and Fujita}{Berliant and
  Fujita}{2008}]{berliant2008knowledge}
\textsc{Berliant, M. and M.~Fujita} (2008): \enquote{Knowledge creation as a
  square dance on the Hilbert cube,} \emph{International Economic Review}, 49,
  1251--1295.

\bibitem[\protect\citeauthoryear{Bessen and Hunt}{Bessen and
  Hunt}{2007}]{bessen2007empirical}
\textsc{Bessen, J. and R.~M. Hunt} (2007): \enquote{An empirical look at
  software patents,} \emph{Journal of Economics \& Management Strategy}, 16,
  157--189.

\bibitem[\protect\citeauthoryear{Bessen and Maskin}{Bessen and
  Maskin}{2009}]{bessen2009sequential}
\textsc{Bessen, J. and E.~Maskin} (2009): \enquote{Sequential innovation,
  patents, and imitation,} \emph{The RAND Journal of Economics}, 40, 611--635.

\bibitem[\protect\citeauthoryear{Bloznelis, G{\"o}tze, and Jaworski}{Bloznelis
  et~al.}{2012}]{bloznelis2012birth}
\textsc{Bloznelis, M., F.~G{\"o}tze, and J.~Jaworski} (2012): \enquote{Birth of
  a strongly connected giant in an inhomogeneous random digraph,} \emph{Journal
  of Applied Probability}, 49, 601--611.

\bibitem[\protect\citeauthoryear{Boldrin and Levine}{Boldrin and
  Levine}{2008}]{boldrin2008perfectly}
\textsc{Boldrin, M. and D.~K. Levine} (2008): \enquote{Perfectly competitive
  innovation,} \emph{Journal of Monetary Economics}, 55, 435--453.

\bibitem[\protect\citeauthoryear{Bramoull{\'e}, Kranton, and
  D'amours}{Bramoull{\'e} et~al.}{2014}]{bramoulle2014strategic}
\textsc{Bramoull{\'e}, Y., R.~Kranton, and M.~D'amours} (2014):
  \enquote{Strategic interaction and networks,} \emph{American Economic
  Review}, 104, 898--930.

\bibitem[\protect\citeauthoryear{Burt}{Burt}{1992}]{burt1992structural}
\textsc{Burt, R.~S.} (1992): \emph{Structural Holes: The Social Structure of
  Competition}, Harvard University Press.

\bibitem[\protect\citeauthoryear{Cabrales, Calv{\'o}-Armengol, and
  Zenou}{Cabrales et~al.}{2011}]{cabrales2011social}
\textsc{Cabrales, A., A.~Calv{\'o}-Armengol, and Y.~Zenou} (2011):
  \enquote{Social interactions and spillovers,} \emph{Games and Economic
  Behavior}, 72, 339--360.

\bibitem[\protect\citeauthoryear{Campbell}{Campbell}{2013}]{campbell2013word}
\textsc{Campbell, A.} (2013): \enquote{Word-of-mouth communication and
  percolation in social networks,} \emph{American Economic Review}, 103,
  2466--98.

\bibitem[\protect\citeauthoryear{Chen, Elliott, and Koh}{Chen
  et~al.}{2021}]{chen2021capability}
\textsc{Chen, J., M.~Elliott, and A.~Koh} (2021): \enquote{Capability
  accumulation and conglomeratization in the information age,} \emph{Available
  at SSRN 2753566}.

\bibitem[\protect\citeauthoryear{Cohen and Levin}{Cohen and
  Levin}{1989}]{cohen1989empirical}
\textsc{Cohen, W.~M. and R.~C. Levin} (1989): \enquote{Empirical studies of
  innovation and market structure,} \emph{Handbook of Industrial Organization},
  2, 1059--1107.

\bibitem[\protect\citeauthoryear{Currarini, Jackson, and Pin}{Currarini
  et~al.}{2009}]{currarini2009economic}
\textsc{Currarini, S., M.~O. Jackson, and P.~Pin} (2009): \enquote{An economic
  model of friendship: Homophily, minorities, and segregation,}
  \emph{Econometrica}, 77, 1003--1045.

\bibitem[\protect\citeauthoryear{Erol and Garc{\'\i}a-Jimeno}{Erol and
  Garc{\'\i}a-Jimeno}{2020}]{erol2020civil}
\textsc{Erol, S. and C.~Garc{\'\i}a-Jimeno} (2020): \enquote{Civil Liberties
  and Social Structure,} \emph{Available at SSRN}.

\bibitem[\protect\citeauthoryear{Golub and Livne}{Golub and
  Livne}{2010}]{golub2010strategic}
\textsc{Golub, B. and Y.~Livne} (2010): \enquote{Strategic Random Networks,}
  \emph{Available at SSRN 1694310}.

\bibitem[\protect\citeauthoryear{Goyal and Moraga-Gonzalez}{Goyal and
  Moraga-Gonzalez}{2001}]{goyal2001r}
\textsc{Goyal, S. and J.~L. Moraga-Gonzalez} (2001): \enquote{R\&d networks,}
  \emph{Rand Journal of Economics}, 686--707.

\bibitem[\protect\citeauthoryear{Griffith}{Griffith}{2019}]{griffith2019continuous}
\textsc{Griffith, A.} (2019): \enquote{A Continuous Model of Strong and Weak
  Ties,} \emph{Working paper}.

\bibitem[\protect\citeauthoryear{Jackson and Wolinsky}{Jackson and
  Wolinsky}{1996}]{jackson1996strategic}
\textsc{Jackson, M.~O. and A.~Wolinsky} (1996): \enquote{A strategic model of
  social and economic networks,} \emph{Journal of Economic Theory}, 71, 44--74.

\bibitem[\protect\citeauthoryear{Karp}{Karp}{1990}]{karp1990transitive}
\textsc{Karp, R.~M.} (1990): \enquote{The transitive closure of a random
  digraph,} \emph{Random Structures \& Algorithms}, 1, 73--93.

\bibitem[\protect\citeauthoryear{Kelly}{Kelly}{2009}]{kelly2009technological}
\textsc{Kelly, M.} (2009): \enquote{Technological progress under learning by
  imitation,} \emph{International Economic Review}, 50, 397--414.

\bibitem[\protect\citeauthoryear{Kerr}{Kerr}{2008}]{kerr2008ethnic}
\textsc{Kerr, W.~R.} (2008): \enquote{Ethnic scientific communities and
  international technology diffusion,} \emph{The Review of Economics and
  Statistics}, 90, 518--537.

\bibitem[\protect\citeauthoryear{Klenke and Mattner}{Klenke and
  Mattner}{2010}]{klenke2010stochastic}
\textsc{Klenke, A. and L.~Mattner} (2010): \enquote{Stochastic ordering of
  classical discrete distributions,} \emph{Advances in Applied Probability},
  42, 392--410.

\bibitem[\protect\citeauthoryear{K{\"o}nig, Battiston, Napoletano, and
  Schweitzer}{K{\"o}nig et~al.}{2011}]{konig2011recombinant}
\textsc{K{\"o}nig, M.~D., S.~Battiston, M.~Napoletano, and F.~Schweitzer}
  (2011): \enquote{Recombinant knowledge and the evolution of innovation
  networks,} \emph{Journal of Economic Behavior \& Organization}, 79, 145--164.

\bibitem[\protect\citeauthoryear{K{\"o}nig, Battiston, Napoletano, and
  Schweitzer}{K{\"o}nig et~al.}{2012}]{konig2012efficiency}
---\hspace{-.1pt}---\hspace{-.1pt}--- (2012): \enquote{The efficiency and
  stability of R\&D networks,} \emph{Games and Economic Behavior}, 75,
  694--713.

\bibitem[\protect\citeauthoryear{K{\"o}nig, Lorenz, and Zilibotti}{K{\"o}nig
  et~al.}{2016}]{konig2016innovation}
\textsc{K{\"o}nig, M.~D., J.~Lorenz, and F.~Zilibotti} (2016):
  \enquote{Innovation vs. imitation and the evolution of productivity
  distributions,} \emph{Theoretical Economics}, 11, 1053--1102.

\bibitem[\protect\citeauthoryear{Kortum}{Kortum}{1997}]{kortum1997research}
\textsc{Kortum, S.~S.} (1997): \enquote{Research, patenting, and technological
  change,} \emph{Econometrica}, 1389--1419.

\bibitem[\protect\citeauthoryear{Kultti, Takalo, and Toikka}{Kultti
  et~al.}{2006}]{kultti2006simultaneous}
\textsc{Kultti, K., T.~Takalo, and J.~Toikka} (2006): \enquote{Simultaneous
  model of innovation, secrecy, and patent policy,} \emph{American Economic
  Review}, 96, 82--86.

\bibitem[\protect\citeauthoryear{{\L}uczak}{{\L}uczak}{1990}]{luczak1990phase}
\textsc{{\L}uczak, T.} (1990): \enquote{The phase transition in the evolution
  of random digraphs,} \emph{Journal of Graph Theory}, 14, 217--223.

\bibitem[\protect\citeauthoryear{Perla and Tonetti}{Perla and
  Tonetti}{2014}]{perla2014equilibrium}
\textsc{Perla, J. and C.~Tonetti} (2014): \enquote{Equilibrium imitation and
  growth,} \emph{Journal of Political Economy}, 122, 52--76.

\bibitem[\protect\citeauthoryear{Rosenthal}{Rosenthal}{1970}]{rosenthal1970subspaces}
\textsc{Rosenthal, H.~P.} (1970): \enquote{On the subspaces of $L^p$ ($p > 2$)
  spanned by sequences of independent random variables,} \emph{Israel Journal
  of Mathematics}, 8, 273--303.

\bibitem[\protect\citeauthoryear{Sadler}{Sadler}{2020}]{sadler2020diffusion}
\textsc{Sadler, E.} (2020): \enquote{Diffusion games,} \emph{American Economic
  Review}, 110, 225--70.

\bibitem[\protect\citeauthoryear{Samila and Sorenson}{Samila and
  Sorenson}{2011}]{samila2011venture}
\textsc{Samila, S. and O.~Sorenson} (2011): \enquote{Venture capital,
  entrepreneurship, and economic growth,} \emph{The Review of Economics and
  Statistics}, 93, 338--349.

\bibitem[\protect\citeauthoryear{Saxenian}{Saxenian}{1996}]{saxenian1996regional}
\textsc{Saxenian, A.} (1996): \emph{Regional Advantage}, Harvard University
  Press.

\bibitem[\protect\citeauthoryear{Scotchmer}{Scotchmer}{1991}]{scotchmer1991standing}
\textsc{Scotchmer, S.} (1991): \enquote{Standing on the shoulders of giants:
  cumulative research and the patent law,} \emph{Journal of Economic
  Perspectives}, 5, 29--41.

\bibitem[\protect\citeauthoryear{Scotchmer and Green}{Scotchmer and
  Green}{1990}]{scotchmer1990novelty}
\textsc{Scotchmer, S. and J.~Green} (1990): \enquote{Novelty and disclosure in
  patent law,} \emph{The RAND Journal of Economics}, 131--146.

\bibitem[\protect\citeauthoryear{Stein}{Stein}{2008}]{stein2008conversations}
\textsc{Stein, J.~C.} (2008): \enquote{Conversations among competitors,}
  \emph{American Economic Review}, 98, 2150--62.

\bibitem[\protect\citeauthoryear{Storper and Venables}{Storper and
  Venables}{2004}]{storper2004buzz}
\textsc{Storper, M. and A.~J. Venables} (2004): \enquote{Buzz: face-to-face
  contact and the urban economy,} \emph{Journal of Economic Geography}, 4,
  351--370.

\bibitem[\protect\citeauthoryear{Tartari and Stern}{Tartari and
  Stern}{2021}]{tartari2021more}
\textsc{Tartari, V. and S.~Stern} (2021): \enquote{More than an Ivory Tower:
  The Impact of Research Institutions on the Quantity and Quality of
  Entrepreneurship,} \emph{Working paper}.

\bibitem[\protect\citeauthoryear{Weitzman}{Weitzman}{1998}]{weitzman1998recombinant}
\textsc{Weitzman, M.~L.} (1998): \enquote{Recombinant growth,} \emph{The
  Quarterly Journal of Economics}, 113, 331--360.

\end{thebibliography}
\setstretch{1.5}

\appendix

\section{Proof of Theorem~\ref{critical}}\label{sec:mainproof}

We will first prove Theorem~\ref{critical}. We begin with two lemmas, which bound equilibrium actions and the number of ideas learned at equilibrium. We then prove the theorem for symmetric equilibria. The symmetric case illustrates the main ideas of the proof. We complete the proof by extending the analysis of the symmetric case to show that any sequence of investment equilibria, which need not be symmetric, is critical.

In all appendices, we say that $f(n) \sim g(n)$ if $f(n)/g(n) \rightarrow 1$ as $n \rightarrow \infty$.

\subsection{Preliminary Lemmas}

Our first lemma, which bounds agents' actions uniformly, will be useful for the subcritical and supercritical cases. Given $\mathbf{q}$, we let $\overline{q}=\max_iq_i$ and $\underline{q}=\min_iq_i$ be the maximum and minimum choices of openness.

\begin{lemma}\label{lem:actionbound}
Consider any sequence of investment equilibria $(\mathbf{p}^*,\mathbf{q}^*)$. There exists a constant $\overline{C}$ such that $\iota(\overline{q},\overline{q})n \leq \overline{C}$ for all $n$.
\end{lemma}

\begin{proof}[Proof of Lemma~\ref{lem:actionbound}]
Suppose not. Relabelling firms, we can assume that $q_1^* = \overline{q}$ for each $n$. Passing to a subsequence if necessary, we can take $q_1^*\sqrt{n} \rightarrow \infty$ as $n\rightarrow \infty$.

We first assume for the sake of contradiction that the expected number of times firm $1$ learns directly is unbounded. Passing to a subsequence we can assume that the expected number of times firm $1$ learns directly converges to infinity, i.e., $\sum_{j\neq 1}\iota(q_1^*,q_j^*)\rightarrow \infty$.

We now fix $n$ large and consider the payoffs to firm $1$ after deviating to choose $q_1$. By the Chernoff bound, the probability that no firm learns indirectly from firm $1$ decays exponentially in $q_1$.

We claim that $\binom{|I_i(\mathbf{p}^*,(q_1,q_{-1}^*))|}{k-1}$ grows at most at a polynomial rate in $q_1$. Let $X$ be a random variable equal to the number of ideas learned by learning from a random firm $j$, each chosen with probability proportional to $q_j$. For $\mathbf{q}$ such that $\sum_{j\neq 1}\iota(q_1,q_j^*)$ is an integer $m$, the random variable $\binom{|I_i(\mathbf{p}^*,(q_1,q_{-i}^*))|}{k-1}$ is first-order stochastically dominated by $\binom{\sum_{j=1}^m X_j}{k-1}$, where $X_1,\hdots,X_m$ are i.i.d. random variables distributed as $X$. This in turn within a constant multiple of $\left(\sum_{j=1}^m X_j\right)^{k-1}$. By Rosenthal's inequality \citep*{rosenthal1970subspaces}, the growth rate of the expectation of this sum of moments in $m$ is at most polynomial in $m$.

Therefore, the payoffs to firm $1$ conditional on no firm learning indirectly from $1$ are at most polynomial in $q_1$. The probability of this event decays exponentially, so for $n$ sufficiently large, firm $1$ could profitably deviate to a smaller choice of $q_1$. This gives a contradiction.

The remaining case is that $q_1^*\sqrt{n}\rightarrow \infty$ but $\sum_{j\neq 1}\iota(q_1,q_j^*)$ is bounded. Then there must exist a sequence of $i$ such that $q^*_i/q^*_1\rightarrow 0$, and for some such sequence of $i$ we have $\sum_{j \neq i} \iota(q_i^*,q_j^*) \rightarrow 0$. This gives a contradiction since then for $n$ large, firm $i$ could profitably deviate to choose $q_i = (\sum_{j \neq i} q_j^*)^{-1}.$ This complete the proof of the lemma.
\end{proof}

\begin{lemma}\label{lem:expbound}
Consider a subcritical sequence of actions such that $\iota(q_i,q_i)n$ is bounded above uniformly. Then
$$\lim_{n \rightarrow \infty} \mathbb{P}[|I_i(\mathbf{p},\mathbf{q})|=y]$$
decreases at an exponential rate in $y$.
\end{lemma}
\begin{proof}[Proof of Lemma \ref{lem:expbound}]

Because the sequence of actions is subcritical, we can assume that the matrix $(\iota(q_i,q_j)\delta)_{ij}$ has spectral radius at most $\overline{\lambda}<1$ for all $n$ sufficiently large. Increasing $q_i$ for some $i$ and therefore also $\overline{\lambda}$, we can assume without loss of generality that there are at most $K$ choices of $q_i$ for each $n$. Here the number of distinct actions $K$ can depend on the initial upper bound $\overline{\lambda}$ and distribution of actions. We denote the number of firms choosing $q_i$ by $n(q_i)$.

We will bound the number of firms $j$ with a path from $j$ to $i$ in the indirect learning network above by the number of nodes in a multi-type branching process with the number of successors distributed as Poisson random variables. The types will correspond to the (at most $K$) choices of $q_i$.

For each firm $i$, the number of firms choosing $q_j$ that firm $i$ learns from indirectly is a binomial random variable with success probability $\delta \iota(q_i,q_j)$ and at most $n(q_j)$ trials. We claim that for any $C>1$, for $n$ sufficiently large such a random variable is first-order stochastically dominated by a Poisson random variable with parameter  $C\delta \iota(q_i,q_j)n(q_j)$. By Theorem 1(f) of \cite*{klenke2010stochastic}, this holds if $$(1-\delta \iota({q}_i,{q}_j))^n \geq e^{-C\iota(q_i,q_j)\delta  n}.$$The inequality is satisfied for $n$ sufficiently large since $C>1$.

Choose $1<C < 1/{\overline{\lambda}}$ and assume $n$ is large enough for the previous claim to hold. Then the number of firms that firms $j$ with a path from $j$ to $i$ in the indirect learning network is first-order stochastically dominated by the number of nodes in the multi-type branching process such that the number of successors of a node of type $q_i$ of each type $q_j$ is distributed as a a Poisson random variable with mean $C\delta \iota(q_i,q_j)n(q_j)$. Call this number of nodes $y'$.

We want to show that $y'<\infty$ with probability one and the probability that $y'=y$ decays exponentially in $y$. The (at most $K \times K$) matrix $C\delta \iota(q_i,q_j)n(q_j)$ has spectral radius at most $C\overline{\lambda}<1$ because  $(C\iota(q_i,q_j)\delta)_{ij}$ does. Therefore, by Theorem 2 of Section V.3 of \cite*{athreya1972branchingprocesses}, the probability that $y'=\infty$ is zero.

Let $Z_j$ be the number of nodes in the $j^{th}$ generation of the branching process. By Theorem 1 of Section V.3 of \cite*{athreya1972branchingprocesses}, the probability that $Z_j >0$ is of order at most $(C\overline{\lambda})^i$. So the probability that $Z_T>0$ decays exponentially in $T$.

Dropping nodes with zero probability of interaction if necessary, we can assume that all $q_i>0$. By the Perron-Frobenius theorem, there exists an eigenvector of $(\delta \iota(q_i,q_j)n(q_j))_{ij}$ with positive real entries and eigenvalue equal to the spectral radius of this matrix. We call this eigenvector $\mathbf{v}$.

We claim that the probability that there are more than $Tv_i$ nodes of some type $q_i$ in one of the generations $1,\hdots,T$ decays exponentially in $T$. There is one node in generation zero. Suppose that there are at most $Tv_i$ nodes of each type $q_i$ in generation $j$. Then by our construction of $\mathbf{v}$, the number of nodes of each type $q_i$ in generation $j+1$ is Poisson with mean at most $Tv_i C\overline{\lambda}$. A Poisson random variable of mean $Tv_i C\overline{\lambda}$ is the sum of $T$ Poisson random variables of mean $v_iC\overline{\lambda}$. So by the central limit theorem, the probability that there are at least $Tv_i$ such nodes decays exponentially in $T$, independent of $j$. This implies the claim.

Therefore, the probability that there are more than $Tv_i$ nodes decays exponentially in $T$. We have completed the proof that $y'<\infty$ with probability one and the probability that $y'=y$ decays exponentially in $y$. Finally, $|I_i(\mathbf{p},\mathbf{q})|$ is first-order stochastically dominated by the sum of $y'$ Bernoulli random variables with $n$ trials and success probability $\max_i \iota(q_i,q_i)$. The statement of the lemma now follows by the central limit theorem.
\end{proof}

\subsection{Symmetric Investment Equilibria}

We begin by describing the structure of this section.
\begin{enumerate}
\item We first show that there does not exist a subcritical sequence of symmetric investment equilibria. The proof assumes such a sequence exists for the sake of contradiction and characterizes equilibrium learning and behavior via two lemmas. The first lemma gives a first-order condition for the choice of $q_i$. This first-order condition is used to prove a second lemma, which shows that $\delta\iota(q^*,q^*)n$ is approximately equal to an expectation $\mathbb{E}[\tau(t)]$, where $\tau(t)$ is the number of links needed to learn the ideas in the technology $t$. Because $\tau(t) \geq 1$ for all $t$, this implies that $\liminf_n \delta \iota(q^*,q^*) n$ is at least one. But this contradicts the assumption that the sequence of equilibria is subcritical.

\item We then show that there does not exist a supercritical sequence of symmetric investment equilibria. We show that given a supercritical sequence of symmetric actions, the payoffs for firm $i$ can be computed (to first order) based on whether firm $i$ learns the ideas learned from the giant component and which firms learn from $i$. This is because the number of ideas that firm $i$ learns from outside the giant component is very likely to be small, so such ideas have a small effect on payoffs. We compute the highest-order term in the firm $i$'s payoffs explicitly and show that each firm $i$ would prefer to deviate to a lower choice of $q_i$ for $n$ large.

\item We finally show that there exists a symmetric investment equilibrium for $n$ large. The argument shows that there is a unique best response near the critical region. Given the preceding analysis, the intermediate value theorem then implies the best response function must have a fixed point.
\end{enumerate}

After passing to a subsequence if necessary, we can assume that any sequence of equilibria is either subcritical, critical, or supercritical.

\textbf{Subcritical Case}: 
We begin with a lemma expressing the first-order condition for $q_i$ at a subcritical sequence of investment equilibria.\footnote{Note that we do not require a symmetry assumption, so we can also apply this lemma in the asymmetric case.}
\begin{lemma}\label{lem:basicFOC}
Along any sequence of subcritical investment equilibria,
$$ \delta\sum_{j \neq i}\frac{\partial \iota(q_i,q_j^*)}{\partial q_i}(q_i^*) \cdot \mathbb{E}\left[\binom{|I_i(\mathbf{p}^*,\mathbf{q}^*)|}{k-1}\right] \sim \mathbb{E}\left[\frac{\partial \binom{|I_i(\mathbf{p}^*,(q_i,q_{-i}^*)|}{k-1}}{\partial q_i}(q_i^*)\right]$$
for each $i$.
\end{lemma}
\begin{proof}[Proof of Lemma \ref{lem:basicFOC}]

We first claim that because the sequence of equilibria is subcritical, competition that is not based on learning all of firm $i$'s ideas indirectly is lower order. More formally, let $T_i(\mathbf{p},\mathbf{q})$ be the set of technologies $t$ such that $i \in t$ and firm $i$ learns all other ideas $j \in t$.\footnote{Recall that $PT_i(q_i,q_{-i})\subset T_i(\mathbf{p},\mathbf{q})$ is the subset of technologies $t$ such that no other firm learns all ideas in $t$.}
The claim is that if there does not exist a link from firm $i$ to another firm $j$ in the indirect-learning network, the conditional probability $\mathbb{E}_{t \in T_i(\mathbf{p}^*,\mathbf{q}^*)}[\mathbf{1}_{t \in PT_i(\mathbf{p}^*,\mathbf{q}^*)}]$
that $t \in PT_i(\mathbf{p}^*,\mathbf{q}^*)$ for a technology $t \in T_i(\mathbf{p}^*,\mathbf{q}^*)$ chosen at random converges to one. Here, each technology $t$ with $i\in t$ is chosen with probability proportional to the probability that firm $i$ knows all ideas in $t$, i.e., $t \in T_i(\mathbf{p}^*,\mathbf{q}^*)$.
 
Suppose $t \in T_i(\mathbf{p}^*,\mathbf{q}^*)$ and no firm learns indirectly from $i$. Choose some $j \in t$ distinct from $i$. By Lemma~\ref{lem:expbound}, the probability that a given firm $j'$ learns $y \geq \underline{y}$ ideas decays exponentially in $y$ at a rate independent of $n$. By independence, the probability that firm $j'$ learns ideas $i$ and $j$ is at most $o(\frac{1}{n})$. Therefore, the probability that any firm $j'$ learns ideas $i$ and $j$ is at most $o(1)$. The technology $t \in PT_i(\mathbf{p},\mathbf{q})$ if there is no such $j'$ for any $j \in t$ distinct from $i$, so this proves the claim.

Thus, we can express the expected utility of player $i$ choosing $p_i \in [0,1)$ and $q_i \geq 0$ as
$$p_i \mathbb{E}\left[\binom{|I_i(\mathbf{p}^*,(q_i,q^*_{-i}))|}{k-1}\right]\prod_{j \neq i}(1-\delta \cdot \iota(q_i,q_j^*))-c(p_i)+o(1).$$
We first note that the optimal $q_i$ does not depend on $p_i$, but instead is chosen to maximize
$$\mathbb{E}\left[\binom{|I_i(\mathbf{p}^*,(q_i,q^*_{-i}))|}{k-1}\right]\prod_{j \neq i}(1-\delta \cdot \iota(q_i,q_j^*))+o(1).\footnote{Note that $I_i(\mathbf{p},\mathbf{q})$ does not depend on $p_i$.}$$
The first-order condition gives
$$\delta  \mathbb{E}\left[\binom{|I_i(\mathbf{p}^*,\mathbf{q}^*)|}{k-1}\right] \sum_{j \neq i}\frac{\partial \iota(q_i,q_j^*)}{\partial q_i}(q_i^*)(1-\delta \cdot \iota(q_i,q_j^*))^{-1} \sim \partial \mathbb{E}\left[ \frac{\binom{|I_i(\mathbf{p},(q_i,q_{-i}^*))|}{k-1}}{\partial q_i}(q_i^*)\right].$$

Finally, we have
$$\sum_{j \neq i}\frac{\partial \iota(q_i,q_j^*)}{\partial q_i}(q_i^*)(1-\delta \cdot \iota(q_i,q_j^*))^{-1} \rightarrow \sum_{j \neq i}\frac{\partial \iota(q_i,q_j^*)}{\partial q_i}(q_i^*)$$
because $\limsup \delta \iota(\overline{q},\overline{q}) n$ is bounded by Lemma~\ref{lem:actionbound}.
\end{proof}

Given $t \in PT_i(\mathbf{p}^*,\mathbf{q}^*),$ let $\tau(t)$ be the smallest number of (direct or indirect) links such that firm $i$ would still know all ideas $j \in t$ with only $\tau(t)$ of its links.

We next prove Lemma~\ref{lem:tauFOC}. In the statement, each technology $t \in PT_i(\mathbf{p}^*,\mathbf{q}^*)$ under a given realization of all random variables is chosen with probability proportional to the probability of that realization.
\begin{namedtheorem}[Lemma~\ref{lem:tauFOC}]
Along any sequence of symmetric investment equilibria with $\limsup \delta \iota(q^*,q^*)n < 1$,
$$\delta \iota(q^*,q^*)n \sim \mathbb{E}_{t \in PT_i(\mathbf{p}^*,\mathbf{q}^*)}[\tau(t)]$$
for all $i$.
\end{namedtheorem}

\begin{proof}[Proof of Lemma \ref{lem:tauFOC}]
We will apply Lemma~\ref{lem:basicFOC}, which gives
$${\delta}\cdot \frac{\partial \iota(q_i,q^*)}{\partial q_i}(q^*)\cdot \mathbb{E}\left[\binom{|I_i(\mathbf{p}^*,\mathbf{q}^*)|}{k-1}\right] \sim \frac{1}{n}\frac{\partial \mathbb{E}\left[\binom{|I_i(\mathbf{p}^*,(q_i,q_{-i}^*))|}{k-1}\right]}{\partial q_i}(q^*)$$
at a symmetric equilibrium.

Let $\Gamma$ be the set of weakly increasing tuples $\gamma = (\gamma_1,\hdots,\gamma_{l(\gamma)})$ of integers such that $\sum_{j=1}^{l(\gamma)} \gamma_j = k-1$. We write $l(\gamma)$ for the length of the tuple $\gamma$.

Let $X_j$ be i.i.d. random variables with distribution given by the number of ideas that firm $i$ would learn from a firm $j'$ conditional on learning directly from $j'$. That is, $X_j$ is distributed as the sum of a Bernoulli random variable with success probability $p^*$ (corresponding to direct learning) and a random variable distributed as $|I_{j'}(\mathbf{p}^*,\mathbf{q}^*)|$ with probability $\delta$ and zero otherwise.

Then we claim that
\begin{equation}\label{eq:fewintersections}\mathbb{E}\left[\binom{|I_i(\mathbf{p}^*,\mathbf{q}^*)|}{k-1}\right] = \sum_{\gamma \in \Gamma} \binom{n-1}{l(\gamma)} \iota(q^*,q^*)^{l(\gamma)}\mathbb{E}\left[ \prod_{j =1}^{l(\gamma)} \binom{X_j}{\gamma_j}\right] + o(1).\end{equation}
The right-hand side counts the expected number of choices of $k-1$ ideas learned (directly or indirectly) via different neighbors $j$, allowing for the same idea to be chosen multiple times via distinct neighbors. To show the claim, we must argue that the contribution from choices of $k-1$ ideas including such repetitions is $o(1)$. 

 By Lemma~\ref{lem:actionbound}, the probability that there are links from $j$ to $i$ in the indirect learning network for $l$ firms $j$ decays exponentially in $l$. On the other hand, the probability of learning the ideas in a component via multiple direct connections converges to zero because each component is $o(n)$ a.a.s. By Lemma~\ref{lem:expbound}, the contribution from the this vanishing probability event vanishes asymptotically. This gives the claim in equation (\ref{eq:fewintersections}).

We can express the right-hand side of Lemma~\ref{lem:basicFOC} similarly. Recall the right-hand side counts the number of additional sets of $k-1$ distinct ideas that would be known to $i$ if $i$ added a direct link to an additional random agent. We then have
\begin{equation}\label{eq:fewintersectionsderiv}\frac{1}{n} \mathbb{E}\left[ \frac{\partial \binom{I_i(\mathbf{p}^*,(q_i,q_{-i}^*))}{k-1} }{\partial q_i}(q^*)\right]=\sum_{\gamma \in \Gamma} \binom{n-2}{l(\gamma)-1} \iota(q^*,q^*)^{l(\gamma)-1} \mathbb{E}\left[\prod_{j =1}^{l(\gamma)} \binom{X_j}{\gamma_j}\right] + o(1).\end{equation}
The same argument shows that the contribution from choices of $k-1$ ideas at least one of which is learned via multiple links is $o(1)$.

Substituting equations (\ref{eq:fewintersections}) and (\ref{eq:fewintersectionsderiv}) into Lemma~\ref{lem:basicFOC}, we find that
$$\iota(q^*,q^*)\delta \sim \mathbb{E}\left[\frac{\binom{n-2}{l(\gamma)-1}}{\binom{n-1}{l(\gamma)}}\right],$$
where the expectation is taken over all $(k-1)$-element sets of ideas in $I_i(\mathbf{p}^*,\mathbf{q}^*)$, and for each such set, $l(\gamma)$ is the number of direct links on a path to at least one idea in the set. Thus,
$$\lim_{n \rightarrow \infty} \iota(q^*,q^*)\delta  n = \lim_{n \rightarrow \infty} \mathbb{E}[l(\gamma)]=  \lim_{n \rightarrow \infty} \mathbb{E}_{t \in PT_i(\mathbf{p}^*,\mathbf{q}^*)}[\tau(t)].\qedhere$$\end{proof}

We must have $\tau(t) \geq 1$ for all $t$. So by Lemma~\ref{lem:tauFOC}, we have $\liminf_n \delta\iota(q^*,q^*)n \geq 1$ at any subcritical sequence of symmetric investment equilibria. This contradicts the definition of a subcritical sequence of equilibria.

 \textbf{Supercritical Case}: Suppose there exists a supercritical sequence of symmetric investment equilibria. We have $\liminf \iota (q^*,q^*) \delta n  > 1$ along this sequence, and we can pass to a convergent subsequence under which $\lim \iota(q^*,q^*) n$ exists or is infinite.
 
Theorem 1 of \cite*{karp1990transitive} shows that a.a.s. the number of firms that all firms in the giant component learn from is $\alpha n+o(n)$ for a constant $\alpha$ increasing in $\lim \iota(q^*,q^*) n$ and that the number of agents outside the giant component observed by any agent is $o(n)$.

If firm $i$ chooses $q_i$, the probability of $i$ learning all ideas known to the giant component is $(1-(1-\delta \iota(q_i,q^*))^{\alpha(n-1)+o(n)})$. Conditional on this event, firm $i$ learns $p^*(\alpha n+o(n))$ ideas. Therefore, we have
$$\mathbb{E}\left[\binom{|I_i(\mathbf{p}^*,(q_i,q_{-i}^*))|}{k-1}\right]= (p^*)^{k-1}(1-(1-\delta \iota(q_i,q^*))^{\alpha n+o(n)})((\alpha(n-1))^{k-1}+o(n^{k-1})).$$
In particular, to solve for firm $i$'s choice of $q_i$ to first order, we need only consider technologies consisting of $i$'s private idea and $(k-1)$ ideas learned by the giant component. The probability that such a technology faces competition is
$$(1-\delta \cdot \iota(q_i,q^*) -(1-\delta)\cdot \iota(q_i,q^*) \cdot h(\alpha))^{n-1}+o(1),$$
where $h(\alpha)$ is the fraction of firms $j$ such that  all of $j$'s ideas are learned by some firm that learns all ideas known to the giant component. The term $\delta \cdot \iota(q_i,q^*)$ corresponds to the possibility of a firm $j$ indirectly learning all of firm $i$'s ideas. The term $(1-\delta)\cdot \iota(q_i,q^*) \cdot h(\alpha)$ corresponds to the possibility of a firm learning firm $i$'s idea via a firm $j$ directly learning (but not indirectly learning) from $i$ and indirectly learning the ideas learned by the giant component.

Thus, we are looking for $q_i$ maximizing
\begin{equation}\label{eq:tomax}\mathbb{E}\left[\binom{|I_i(\mathbf{p}^*,(q_i,q_{-i}^*))|}{k-1}\right](1-\delta \cdot \iota(q_i,q^*) -(1-\delta)\cdot \iota(q_i,q^*) \cdot h(\alpha))^{n-1}\end{equation}
asymptotically. We claim the derivative of expression (\ref{eq:tomax}) in $q_i$ is equal to zero at $q^*$ only if $\iota(q^*,q^*)\delta n \leq 1$. A fortiori, we can instead show this for
$$\mathbb{E}\left[\binom{|I_i(\mathbf{p}^*,(q_i,q_{-i}^*))|}{k-1}\right](1-\delta \cdot \iota(q_i,q^*))^{n-1}.$$
This is because for $n$ large, the derivative of this expression in $q_i$ is positive if the derivative of expression (\ref{eq:tomax}) is positive.
 
The first-order condition for the latter expression at $q_i=q^*$ implies that
$$(p^*)^kn(\alpha n)^{k-1} \delta (1-\delta \iota(q^*,q^*))^{n-1}(\delta \alpha (1-\delta  \iota(q^*,q^*)) (1-\delta  \iota(q^*,q^*))^{\alpha(n-1)-1}- (1-(1-\delta  \iota(q^*,q^*))^{\alpha(n-1)})$$
is $o(n^{k-1})$.
Solving for $\iota(q^*,q^*)$ such that this holds, we obtain
$$\lim_n (1-\delta  \iota(q^*,q^*))^{\alpha(n-1)-1} = \frac{1}{1+\delta \alpha}.$$
Thus, 
\begin{equation}\label{eq:giantcompsize}\lim_n e^{-\delta \alpha  \iota(q^*,q^*) n} = \frac{1}{1+\delta \alpha}.\end{equation}
The left-hand side is the asymptotic probability that a firm does not indirectly learn from a firm which learns all ideas known to all firms in the giant component, and this is $1-\alpha$. So $$1-\alpha = \frac{1}{1+\delta\alpha},$$
or equivalently $\alpha(\alpha \delta + 1 - \delta)=0.$ Since $\delta\leq 1$, our sequence of solutions to equation (\ref{eq:giantcompsize}) must have $\alpha = 0$, and therefore cannot be supercritical.

We have now shown that any sequence of symmetric investment equilibria is critical. We next prove that there exists a symmetric investment equilibrium for $n$ large.

\textbf{Existence}: We first show that given private investment $p>0$, there exists a level of openness $q$ that is a best response when all other firms choose $(p,q)$. We then show that when all firms choose openness $q$, there exists an optimal $p$ for all firms.

Let $BR(q)$ be the set of best responses $q$ when all other firms choose $q_j=q$ and $p_j=p>0$. Note that the set $BR(q)$ does not depend on the value of $p>0$. 

\begin{lemma}\label{lem:uniquebr}
There exists $\kappa>0$ such that for $n$ sufficiently large, if $\delta \iota(q,q)n \in (1-\kappa,1+\kappa)$ then $BR(q)$ is a single-valued correspondence.
\end{lemma}
\begin{proof}[Proof of Lemma \ref{lem:uniquebr}]

Let $\epsilon>0$. Taking $\kappa$ sufficiently small, we first claim that we can choose $M>0$ such that $$\lim_n \mathbb{P}\left[|I_i(\mathbf{p},\mathbf{q})| >M \right]<\epsilon$$ for all $q$ such that $\delta \iota(q,q)n \in (1-\kappa,1+\kappa)$.  For any $q<q'$, the random variable $|I_i(\mathbf{p},\mathbf{q})|$ is first-order stochastically dominated by $|I_i(\mathbf{p},\mathbf{q}')|$. So we can assume $q=1+\kappa$. For $M$ and $n$ sufficiently large, the probability of learning at least $M$ ideas is bounded above by twice the probability of learning all ideas learned by the giant component, which is less than $\epsilon$ for $\kappa$ sufficiently small.

Now let $M'>0$. We next claim that we can choose $\kappa>0$ such that $$\liminf_n \mathbb{E}[{|I_i(\mathbf{p},\mathbf{q})|}] >M'$$ for all $q$ such that $\delta \iota(q,q)n \in (1-\kappa,1+\kappa)$.

Let $x >0$ and define $q(x)$ by $\iota(q(x),q(x))=\frac{1-x}{\delta n}$. For any $q'  < q$, the random variable $|I_i(\mathbf{p},\mathbf{q'})|$ is first-order stochastically dominated by $|I_i(\mathbf{p},\mathbf{q})|$. So it is sufficient to show that  $$\lim_{\lambda \rightarrow 0}\lim_{n \rightarrow \infty}\mathbb{E}[{|I_i(\mathbf{p},\mathbf{q}(x))|}] \rightarrow \infty.$$

We can bound $|I_i(\mathbf{p},\mathbf{q}(x))|$ below by the expected number of firms $j$ with a path from $j$ to $i$ in the indirect learning network. By Theorem 11.6.1 of \cite*{alon2004probabilistic}, the limit of this quantity as $n\rightarrow \infty$ is equal to the number of nodes in a Poisson branching process with parameter $1-x$. As $x\rightarrow 0$, this number of nodes converges to infinity. This proves the claim.

Suppose all other firms choose $q_{-i}=q$ for some $q$ such that $\delta \iota(q,q) n \in (1-\kappa,1+\kappa)$. It is sufficient to show that we can choose $\kappa>0$ such that for $n$ large, there is a unique $q_i$ such that the derivative $\frac{\partial U_i(\mathbf{p},(q_i,{q}_{-i}))}{\partial q_i}$ is equal to zero. By the argument in Lemma~\ref{lem:actionbound}, we can restrict to $q_i$ such that $\iota(q_i,q) \leq \frac{\overline{C}}{\delta }$ for some constant $\overline{C}$.

We can again check this first-order condition in the case that no firm has learned indirectly from $i$, as else payoffs are zero. Fix $\epsilon>0$. Let $v$ be the expected value of a single incoming link. We claim that for $n$ sufficiently large, we can take $\kappa$ so that the derivative in $q_i$ of the expected marginal cost of an outgoing link is greater than $\delta v/2$ for $q_i$ such that $\iota(q_i,q) \leq \frac{\overline{C}}{\delta }$.

The marginal cost of an outgoing link is equal to $\delta \binom{|I_i(\mathbf{p},(q_i,q_{-i}))|}{k-1}+o(\mathbb{E}[\binom{|I_i(\mathbf{p},(q_i,q_{-i}))|}{k-1}])$. Fixing $\epsilon>0$ and $M$ sufficiently large, we can assume the probability that firm $i$ has already learned at least $M$ ideas is at most $\epsilon$. For any $M'>0$,  we have $$\liminf_n \mathbb{E}[{|I_i(\mathbf{p},\mathbf{q})|}] >M'$$for $\kappa$ sufficiently small. Taking $M'$ sufficiently large relative to $M$, we can assume that the change in $\binom{|I_i(\mathbf{p},(q_i,q_{-i}))|}{k-1}$ from an additional outgoing link is greater than $v/2$. This proves the claim.

We also claim that for $\kappa$ small enough and $n$ sufficiently large the derivative in $q_i$ of the expected marginal value of an additional incoming link (conditioning on no firm learning indirectly from $i$) is less than $\delta v/4$. 

The probability that firm $i$ learns from $d$ firms decays exponentially in $d$.  By Rosenthal's inequality \citep*{rosenthal1970subspaces}, the payoffs to learning from ${d}$ firms grow at most at a polynomial rate in $d$. Thus we can choose $\overline{d}$ such that an arbitrarily small share of firm $i$'s expected payoffs are from the event that firm $i$ learns from more than $\overline{d}$ other firms for $n$ large.

We can condition on the event that firm $i$ learns from at most $\overline{d}$ other firms. Fixing $\epsilon>0$ and $M$ sufficiently large, we can assume the probability that firm $i$ has already learned at least $M$ ideas is at most $\epsilon$. Taking $M'$ sufficiently large relative to $M$, we can assume the expected number of ideas learned by each other firm is at least $M'$.

We want to evaluate the change in the marginal value of an additional incoming link when $i$ gains an incoming link. We consider two cases depending on whether firm $i$ learns at least $M$ ideas from each of two or more incoming links.

The probability that $i$ learns at least $M$ ideas from each of two or more incoming links is at most $\epsilon^2 \binom{\overline{d}+2}{2}$ and the expected payoffs conditioning on this event are within a constant multiple of the expected payoffs conditioning on learning at least $M$ ideas. So the payoffs from this event are a vanishing fraction of $v$ asymptotically.

Suppose firm $i$ learns at least $M$ ideas from at most one link. The change in the marginal value of an additional incoming link when $i$ gains an incoming link is equal to the number of choices of $k-1$ ideas learned by firm $i$, at least one of which comes from each of these two new incoming links. We have assumed firm $i$ does not learn more than $M$ ideas from both of these links, and therefore the number of such technologies is at most $M \binom{|I_i(\mathbf{p},(q_i,q_{-i})||}{k-2}|$. Taking $M'$ sufficiently large relative to $M$, we can take this quantity to be an arbitrarily small fraction of $v$ for $n$ large. Combining these arguments shows the claim.

We have shown that conditioning on no firm learning indirectly from $i$, the derivative in $q_i$ of the expected marginal cost of an additional outgoing link is at least $\delta v/2$ while the derivative in $q_i$ of the expected marginal value of an additional incoming link is less than $\delta v/4$. So the derivative in $\iota(q_i,q_{-i})$ of expected payoffs conditional on no firm learning indirectly from $i$ is strictly decreasing, and therefore the first-order condition is satisfied uniquely, for $n$ sufficiently large.
\end{proof}

By the lemma, we can choose $\kappa$ such that $BR(q)$ induces a function on the set of $q$ such that $\delta \iota(q,q)n \in (1-\kappa,1+\kappa)$ for $n$ large. Fix such a $\kappa$. Because payoffs are continuous in $\mathbf{q}$, this function $BR(q)$ is continuous. 

We have shown in our analysis of the subcritical region that when $\iota(q,q)= \frac{1-\kappa}{\delta n}$, any element $q'$ of $BR(q)$ has $\delta\iota(q',q)n$ approximately equal to the expectation of $\tau(t)$ over proprietary technologies.\footnote{We stated this result above at equilibrium, but only used that the firm was choosing a  best response.} Because $\tau(t) \geq 1$ for all $t$, we have $$\liminf_n \iota(BR(q),q) n \geq \frac{1}{\delta }.$$

We have shown in our analysis of that supercritical region that if $\iota(q,q)= \frac{1+\kappa}{\delta n}$, then $$\limsup_n \iota(BR(q),q) n < \frac{1}{\delta }.$$
So for $n$ large, we have $BR(q)>q$  when $\iota(q,q)= \frac{1-\kappa}{\delta n}$ while $BR(q) <q$ when $\iota(q,q)= \frac{1+\kappa}{\delta n}$. By the intermediate value theorem, there is a fixed point of the function $BR(q)$ with  $\delta \iota(q,q)n \in (1-\kappa,1+\kappa)$ for $n$ large.

Choose any such fixed point $q^*$ of $BR(q)$. Fix any firm $i$ and let the potential proprietary technologies $PPT_i(\mathbf{q})$ be the set of technologies $t$ such that firm $i$ will receive monopoly profits for $t$ if all ideas in the technology $t$ are discovered. This is a random object depending on the realizations of interactions but not on the realizations of private investment, and $PPT_i(\mathbf{q}) \cap I = PT_i(\mathbf{p},\mathbf{q})$. We have shown that $\mathbf{q}^*$ is critical, so $\mathbb{E}[|PPT_i(\mathbf{q})|]\rightarrow \infty$.

A symmetric equilibrium corresponds to $p^*$ satisfying
$$p^*=\text{argmax}_{p} p(p^*)^{k-1}\mathbb{E}[|PPT_i(\mathbf{q})|]-c(p).$$
Taking the first-order condition, a symmetric equilibrium corresponds to $p^*$ satisfying
$$c'(p^*)=(p^*)^{k-1}\mathbb{E}[|PPT_i(\mathbf{q})|].$$
Because $c(\cdot)$ is continuously differentiable and convex with $c'(0) \geq 0$ and $c(p)\rightarrow \infty$ as $p\rightarrow 1$, while $\mathbb{E}[|PPT_i(\mathbf{q})|]\rightarrow \infty$, there exists a solution for $n$ sufficiently large. So there exists a symmetric investment equilibrium for $n$ sufficiently large.

\subsection{Arbitrary Investment Equilibria}

To complete the proof of Theorem~\ref{critical}, it remains to extend our characterization from symmetric equilibria to arbitrary equilibria. It is again sufficient to show that we cannot have a supercritical sequence of investment equilibria or a subcritical sequence of investment equilibria, and we treat each case separately.

We first consider the supercritical case, and show that there exists a giant component with the same relevant properties as in our analysis of symmetric equilibria. We then consider the subcritical case, which follows the same basic outline as in our analysis of symmetric equilibria.

\textbf{Supercritical Case}: Because the sequence of actions is supercritical, we can assume that the matrix $(\iota(q_i^*,q_j^*)\delta)_{ij}$ has spectral radius at least $\underline{\lambda}>1$ for all $n$ sufficiently large.

We first claim that there exists $\underline{\alpha}>0$ such that for all $n$, there is a component of the learning network containing at least $\underline{\alpha}n$ firms a.a.s. It is sufficient to show this after decreasing $q_i$ for some $i$, and therefore also $\underline{\lambda}$. By Lemma~\ref{lem:actionbound}, we have $q_i^* \leq \overline{C}/\sqrt{n}$ for each $i$. We can therefore assume without loss of generality that there are at most $K$ choices of $q_i$ for each $n$. Here the number of distinct actions $K$ can depend on the initial upper bound $\underline{\lambda}$. We denote the number of firms choosing $q_i$ by $n(q_i)$.

By Theorem 1 of \cite*{bloznelis2012birth}, the largest component has at least $\underline{\alpha}n + o(n)$ nodes a.a.s., where $1-\underline{\alpha}$ is the extinction probability of the multi-type branching process with types corresponding to choices of $q_i$ and the number of successors of type $q_{j}$ of a node of type $q_i$ distributed as a Poisson random variable with mean $\delta \iota(q_i,q_j) n(q_j)$. By Theorem 2 of Section V.3 of \cite*{athreya1972branchingprocesses}, this extinction probability satisfies $1-\underline{\alpha} > 0$ for $n$ large since $\lambda>\underline{\lambda} > 1$. This proves the claim, and we now return to studying the original actions $\mathbf{q}$.

Because there is a component of the learning network containing at least $\underline{\alpha}n$ firms with probability at least $\epsilon$, we can also choose $\underline{C}$ such that $\underline{q}$ is at most $\frac{\underline{C}}{\sqrt{n}}$ for all $n$. To see this, note that the payoffs to choosing $q_i = \frac{1}{\sqrt{n}}$ are of order $n^{k-1}$. On the other hand, the payoffs to choosing $q_i = \frac{\underline{C}}{\sqrt{n}}$ are bounded above by $2(1-e^{-\underline{C}\overline{C}})n^{k-1}$. So for $\underline{C}$ sufficiently small, the expected payoffs to choosing $q_i=\frac{\underline{C}}{\sqrt{n}}$ conditional on any realization of all links between firms other than $i$ are less than the expected payoffs to choosing $q_i=\frac{1}{\sqrt{n}}$ at equilibrium.

As $n$ grows large, for all $i$ and $j$ distinct the probability that $j \in t$ for a uniformly chosen $t \in PT_i(\mathbf{q}^*)$ approaches zero. We will show that for any sequence of best responses $q_i$ for firm $i$, the expected number of interactions $\sum_{j \neq i} \iota(q_i^*,q_j^*)$ has a unique limit which is independent of $i$.

To show this, we next claim that there is at most one component of linear size a.a.s. To do so, we will use equation (5) of \cite*{bloznelis2012birth}. In the notation of \cite*{bloznelis2012birth}, the type space $S$ will be $S=[\underline{C},\overline{C}]$, and the kernel $\kappa(s,s') = ss'$. We will identify the type of an agent $i$ with $q_i^*\sqrt{n}$.

The space of distributions $\Delta(S)$ over types is compact. Fix such a distribution. Relabelling so that $q_i^*$ are increasing in $i$, we can generate a random network for each $n$ by taking the action $q_i^*$ of agent $i$ to be $s/\sqrt{n}$, where $s$ is the $(i/n)^{th}$ quantile of the distribution. As $n\rightarrow \infty$, by equation (5) of \cite*{bloznelis2012birth}, the largest component of the learning network learns $\alpha n + o(n)$ ideas for some $\alpha \in [0,1]$. It follows from Theorem 1 and the same approximation techniques used in that paper that  the second largest component learns $o(n)$ ideas. Because the space of distributions $\Delta(S)$ is compact, this convergence of component sizes is uniform. So passing to a convergent subsequence if necessary, we can assume that there is a unique giant component learning $\alpha n + o(n)$ ideas a.a.s., where $\alpha > 0$.

The payoffs to choosing $q_i$ are then equal to $(p^*)^k\binom{\alpha n}{k-1}$
times the probability that firm $i$ learns all ideas known to the giant component and no firm $j$ learns $i$'s idea and all ideas known to the giant component, plus a term of order $o(n^{k-1})$. Formally, if $G_1$ is the set of firms whose ideas are learned by some firm that learns all ideas known to the giant component, the action $q_i$ is chosen to maximize
$$\binom{\alpha n}{k-1} \left(1-\prod_{j \in G_1}(1-\delta \iota(q_i,q_j^*))\right)\prod_{j \in G_1}(1-\iota(q_i,q_j^*))\prod_{j\notin G_1}(1-\delta \iota(q_i,q_j^*))+o(n^{k-1}).$$
Taking the first-order condition, we must have
$$\delta \sum_{j \in G_1} q_j^*\sim \left(1-\prod_{j \in G_1}(1-\iota(q_i,q_j^*))\right)(\sum_{j \in G_1} q_j^* + \delta \sum_{j \notin G_1}q_j^*).$$
Since the right-hand side is increasing in $q_i$, the solution has a unique limit $\lim_n \sum_{j \neq i} \iota(q_i^*,q_j^*).$
This limit does not depend on $i$. Passing to a subsequence if necessary, we can assume the limit $\lim_{n \rightarrow \infty}\sum_{j} \iota({q}_i,q_j^*)$ exists and is independent of $i$. Moreover, this limit must be greater than $\frac{1}{\delta}$ for equilibrium to be supercritical. But then the same calculation as in the symmetric case shows that the best response $q_i$ for all firms is at most $\frac{1}{\sqrt{\delta n}}$ asymptotically, which gives a contradiction.

\textbf{Subcritical Case}: The largest component of the learning network has at most $o(n)$ nodes a.a.s. By Lemma~\ref{lem:actionbound}, we can choose $\overline{C}$ such that $\overline{q}$ is at most $\frac{\overline{C}}{\sqrt{n}}$ for all $n$. We now proceed to derive a characterization of equilibrium as in Lemma~\ref{lem:tauFOC}. We will then use this characterization to show the result.

We now let $X_j$ be i.i.d. random variables with distribution given by the number of ideas that firm $i$ would learn from a firm $j$ conditional on learning directly from $j$. That is, $X_j$ is distributed as the sum of a Bernoulli random variable with success probability $p_j^*$ (corresponding to direct learning) and a random variable with the distribution of $|I_j(\mathbf{p}^*,\mathbf{q}^*)|$  with probability $\delta$ and equal to zero otherwise.

Therefore, we can express the expected number of technologies firm $i$ can make if it discovers its own idea as
$$\mathbb{E}\left[\binom{|I_i(\mathbf{p}^*,\mathbf{q}^*)|}{k-1}\right] = \sum_{\gamma \in \Gamma} \sum_{j_1,\hdots,j_{l(\gamma)}\neq i} \prod_{r =1}^{l(\gamma)} q^*_iq^*_{j_r}\mathbb{E}\left[ \prod_{r =1}^{l(\gamma)} \binom{X_{j_r}}{\gamma_r}\right] + o\left(\mathbb{E}\left[\binom{|I_i(\mathbf{p}^*,\mathbf{q}^*)|}{k-1}\right]\right).$$
The second summation is over choices of $l(\gamma)$ distinct firms other than $i$.

Taking the first-order condition in $q_i$,
$$\frac{1}{n} \mathbb{E}\left[ \frac{\partial \binom{I_i(\mathbf{p}^*,(q_i,q_{-i}^*))}{k-1} }{\partial q_i}(q^*)\right]=\frac{1}{n}\sum_{\gamma \in \Gamma} \sum_{j_1,\hdots,j_{l(\gamma)}\neq i} \left(\sum_{r =1}^{l(\gamma)} q_{j_r}^* \prod_{r' \neq r} q^*_iq^*_{j_r}\right)\mathbb{E}\left[ \prod_{r =1}^{l(\gamma)} \binom{X_{j_r}}{\gamma_r}\right] + o\left(\mathbb{E}\left[\binom{|I_i(\mathbf{p}^*,\mathbf{q}^*)|}{k-1}\right]\right).$$

By Lemma~\ref{lem:basicFOC}, we have
$$ \delta (\sum_{j \neq i} q_j^*)\sum_{\gamma \in \Gamma} \sum_{j_1,\hdots,j_{l(\gamma)}\neq i} \prod_{r =1}^{l(\gamma)} q^*_iq^*_{j_r}\mathbb{E}\left[ \prod_{r =1}^{l(\gamma)} \binom{X_{j_r}}{\gamma_r}\right]\sim \sum_{\gamma \in \Gamma} \sum_{j_1,\hdots,j_{l(\gamma)}\neq i} \left(\sum_{r =1}^{l(\gamma)} q_{j_r}^* \prod_{r' \neq r} q^*_iq^*_{j_r}\right)\mathbb{E}\left[ \prod_{r =1}^{l(\gamma)} \binom{X_{j_r}}{\gamma_r}\right].$$
Rearranging,
$$\sum_{\gamma \in \Gamma} \sum_{j_1,\hdots,j_{l(\gamma)}\neq i} \prod_{r =1}^{l(\gamma)} q^*_iq^*_{j_r}\mathbb{E}\left[ \prod_{r =1}^{l(\gamma)} \binom{X_{j_r}}{\gamma_r}\right](\delta q_i^*(\sum_{j \neq i} q_j^*)-l(\gamma) ) \sim 0.$$
In particular, we have
\begin{equation}\label{eq:asymaction}\delta q_i^*(\sum_{j \neq i} q_j^*) \sim \mathbb{E}_{t \sim G_i(\mathbf{p}^*,\mathbf{q}^*)}[\tau(t)]\end{equation}
where the expectation is taken with respect to the appropriate distribution $G_i(\mathbf{p}^*,\mathbf{q}^*)$ over technologies.

As in the symmetric case above, this implies that $\liminf_{n\rightarrow \infty} \delta q_i^*(\sum_{j \neq i} q_j^*)  \geq 1$ for each $i$. So the limit inferior of the row sums of $(\delta \iota(q_i^*,q_j^*))_{i,j}$ is at least one. Thus the spectral radius of this matrix also satisfies $\limsup_n \lambda \geq 1,$ which contradicts our assumption that the sequence of equilibria is subcritical.

We conclude that any sequence of investment equilibria must be critical, which proves Theorem~\ref{critical}.
\newpage
\Large
\begin{center}
\textbf{Online Appendix}
\end{center}
\normalsize
\section{Remaining Proofs}

\begin{proof}[Proof of Corollary~\ref{cor:open}]
Consider a sequence of equilibria $(\mathbf{p}^*,\mathbf{q}^*)$ with non-vanishing investment.

By Theorem~\ref{critical}, each firm learns $o(n)$ ideas at the equilibrium $(\mathbf{p}^*,\mathbf{q}^*)$. Thus $\mathbb{E}[| T_i(\mathbf{p},\mathbf{q})|]$ is $o(n^{k-1})$ for each $i$ and $\mathbb{E}[| \bigcup_i T_i(\mathbf{p},\mathbf{q})|]$ is $o(n^k)$. The discovery rate $D(\mathbf{p}^*,\mathbf{q}^*)\rightarrow 0$.
 
The spectral radius $\lambda \rightarrow 1$ at a sequence of investment equilibria by Theorem~\ref{critical}. Let $\lambda'$ be the spectral radius under actions $(\mathbf{p}^*,(1+\epsilon)\mathbf{q}^*)$ for each $n$. Since $\lambda$ is the spectral radius of the matrix $(\delta q_iq_j)_{ij}$ and $\lambda'$ is the spectral radius of the matrix $(\delta (1+\epsilon)^2q_iq_j)_{ij}$, we have $\lambda'=(1+\epsilon)^2\lambda$. In particular, $\lambda'>1+\epsilon$ for $n$ sufficiently large.

Therefore, as shown in the proof of Theorem~\ref{critical}, when actions are $(\mathbf{p}^*,(1+\epsilon)\mathbf{q}^*)$, there is a giant component of firms learning at least $\widetilde{\alpha}n$ ideas for some $\widetilde{\alpha} > 0$. Any firm in this giant component which discovers an idea will produce at least $\binom{\widetilde{\alpha} n}{k-1}$ technologies. There are a non-vanishing share of such firms, so the discovery rate satisfies $\liminf_n D(\mathbf{p}^*,\mathbf{q}^*)>0$.
\end{proof}

\begin{proof}[Proof of Corollary~\ref{cor:investment}]
Consider a sequence of equilibria $(\mathbf{p}^*,\mathbf{q}^*)$ with non-vanishing investment. We first consider the derivative of $D(\mathbf{p}^*+x\mathbf{1},\mathbf{q}^*)$ in $x$. We want to evaluate the change in the discovery rate $D(\mathbf{p}^*\mathbf{q}^*)$ if an additional idea $i$ (chosen uniformly at random) is discovered, which is equal to $\frac{1}{n} \frac{\partial D(\mathbf{p}^*+x\mathbf{1},\mathbf{q}^*)}{\partial x}$.

By Theorem~\ref{critical}, the sequence of equilibria is critical. So a.a.s., the firm $i$ learns $o(n)$ ideas and the firms' idea $i$ is learned by at most $o(n)$ firms. Thus in expectation there are at most $o(n^{k-1})$ additional technologies produced including firm $i$'s idea, and so the change in $D(\mathbf{p}^*+x\mathbf{1},\mathbf{q}^*)$ is $o(\frac{1}{n})$. Therefore the derivative of $D(\mathbf{p}^*+x\mathbf{1},\mathbf{q}^*)$ in $x$ is $o(1)$.

We now consider the derivative of $D(\mathbf{p}^*+x\mathbf{1},(1+\epsilon)\mathbf{q}^*)$ in $x$. As shown in the proof of Corollary~\ref{cor:open}, given actions $(\mathbf{p}^*,(1+\epsilon)\mathbf{q}^*)$ a.a.s. there exists a unique giant component of firms learning at least $\widetilde{\alpha}n$ ideas for some $\widetilde{\alpha} > 0$ and at least $\alpha n$ firms learn all ideas learned by this giant component, where $\widetilde{\alpha}, \alpha>0$.

We want to evaluate the change in the discovery rate $D(\mathbf{p}^*\mathbf{q}^*)$ if an additional idea $i$ is discovered. If firm $i$ learns all ideas learned by the giant component, then at least $\binom{\widetilde{\alpha}n}{k-1}$ new technologies are produced. So for a random $i$, in expectation there are at least $\alpha \binom{\widetilde{\alpha}n}{k-1}$ additional technologies produced including firm $i$'s idea, and so the change in $D(\mathbf{p}^*+x\mathbf{1},\mathbf{q}^*)$ is at least $C/n$ for some $C>0$. Therefore the derivative of $D(\mathbf{p}^*+x\mathbf{1},\mathbf{q}^*)$ in $x$ is bounded below by a positive constant.
\end{proof}

\begin{proof}[Proof of Proposition~\ref{prop:public}]
Let $b(n)$ be the share of public innovators for each $n$.

We first show that $\liminf_n D(\mathbf{p}^*,\mathbf{q}^*)>0$
along any sequence of equilibria with non-vanishing investment.

It is weakly dominant and strictly preferred at any investment equilibrium for all public innovators to choose $q_i=1$. Therefore, all public innovators are in the same component of the learning network. Private investment $p_i$ by public innovators is non-vanishing, so asymptotically almost surely all firms in this component learn at least $\alpha n$ ideas for some $\alpha > 0$.

Let $\overline{q}$ and $\underline{q}$ be the maximum and minimum levels of openness $q_i$ chosen at equilibrium by private firms, respectively. Because the probability that no firm learns indirectly from $i$ vanishes exponentially in $\iota(q_i^*,1)n$ while payoffs are $O(n^{k-1})$, the quantity $\iota(\overline{q},1)n$ must be bounded at equilibrium.

A consequence is that $\iota(\overline{q},\overline{q})n \rightarrow 0$. Therefore, the expected number of links to a firm $i$ from other firms is a vanishing fraction of the expected number of links to $i$ from public innovators. So a.a.s., a given firm $i$'s links are all with public innovators.

Since learning indirectly from a public innovator implies learning at least $\alpha n$ ideas a.a.s., it follows that $\iota(\underline{q},1)n$ does not vanish asymptotically at equilibrium. Therefore, a positive fraction of firms learn at least $\alpha n$ ideas a.a.s. Any such firm which discovers an idea will produce at least $\binom{\alpha n}{k-1}$ technologies, so $\liminf_n D(\mathbf{p}^*,\mathbf{q}^*)>0$. This proves the characterization of equilibria with non-vanishing investment.

It remains to show there exists a sequence of symmetric equilibria with non-vanishing investment. Recall that we now call an equilibrium symmetric if all public innovators choose the same action and the same holds for all private firms. Suppose that all firms other than $i$ choose $(p,q)$ with $p\geq \frac12$ and $\delta qn \leq C$ for some $C>0$ and all public innovators choose $p_i = p' \geq \frac12$ and $q_i=1$.

We claim that the best response $q_i$ for a firm $i$ is unique. To see this, observe that best response $q_i$ must satisfy the first-order condition after conditioning on no firm learning indirectly from $i$. Up to lower order terms, we can also condition on no firm having learned $i$'s idea and all ideas known to the giant component. Then the expected marginal value of an additional incoming link is strictly decreasing in $q_i$ while the expected marginal cost of an outgoing link is strictly increasing in $q_i$. So the first-order condition is satisfied at a unique value of $q_i$, and thus there is a unique best response.

If $q_i$ is the best response for $i$, then $\lim_n q_in$ exists and is independent of $(\mathbf{p},\mathbf{q})$ (given the restrictions on actions assumed above). This is because the probability of interactions between $i$ and other firms vanishes asymptotically, while the best response does not depend on the number of ideas learned by the unique giant component.

Therefore, we can choose $\epsilon>0$ and $C>0$ such that if $q \in [\frac{\epsilon}{\delta n},\frac{C}{\delta n}],$ then for $n$ large so is any best response $q_i$ for firm $i$. We claim that for $n$ large, given $\mathbf{p}$, there exists $q$  that is a best response to $(\mathbf{p},\mathbf{q}).$ This follows from Kakutani's fixed point theorem as in the proof of Theorem~\ref{critical}. We call this choice of openness $q(\mathbf{p})$.

Given such $(\mathbf{p},\mathbf{q}(\mathbf{p}))$, each firm has a non-vanishing probability of learning a linear number of ideas. Therefore, $\mathbb{E}[|I_i(\mathbf{p},\mathbf{q}(\mathbf{p}))|]\rightarrow \infty.$ So for $n$ large, any best response $p_i$ for each public innovator and each firm $i$ has $p_i \geq \frac12$. By Kakutani's fixed point theorem, there exist symmetric actions $(\mathbf{p},\mathbf{q}(\mathbf{p}))$ such that $p_i \geq \frac12$ is also a best response for each $i$. Thus there exists a sequence of symmetric equilibria with non-vanishing investment.
\end{proof}

\begin{proof}[Proof of Proposition~\ref{prop:welfare}]

(1) $\Rightarrow$ (2): Suppose there exists $\widetilde{\alpha} \in (0,1)$ such that there is a unique component of the indirect-learning network of size at least $\widetilde{\alpha}n$ a.a.s. When this event occurs, a positive fraction of firms learn at least $(\liminf_n \min_i p_i/2)\cdot \widetilde{\alpha}n$ ideas for $n$ sufficiently large. For any such firm $i$, any technology $t$ consisting of $i$ and $k-1$ of these ideas is in $T_i(\mathbf{p},\mathbf{q})$. So a positive fraction of technologies are produced, and therefore $\liminf_n D(\mathbf{p},\mathbf{q})>0$.

(2) $\Rightarrow$ (3): Suppose $\liminf_n D(\mathbf{p},\mathbf{q})>0$, so that there exists $C>0$ such that for $n$ large $\mathbb{E}[\left|T_i(\mathbf{p},\mathbf{q})\right|]>Cn^{k-1}$. Since $T_i(\mathbf{p},\mathbf{q})=PT_i(\mathbf{p},\mathbf{q}) \bigcup CT_i(\mathbf{p},\mathbf{q})$, either there exists $C>0$ such that for $n$ large $\mathbb{E}[\left|PT_i(\mathbf{p},\mathbf{q})\right|]>Cn^{k-1}$ or there exists $C>0$ such that for $n$ large $\mathbb{E}[\left|CT_i(\mathbf{p},\mathbf{q})\right|]>Cn^{k-1}$. In either case, expected consumer surplus $$w_{PT}\mathbb{E}[\left|PT_i(\mathbf{p},\mathbf{q})\right|]+w_{CT} \mathbb{E}[\left|CT_i(\mathbf{p},\mathbf{q})\right|]$$ is at least $C'n^{k-1}$ for some $C'>0$ and $n$ large. Since expected firm profits are non-negative, we conclude expected welfare $\mathbb{E}[W(\mathbf{p},\mathbf{q})]$ is at least $C'n^{k-1}$ for $n$ large.

(3) $\Rightarrow$ (1): Suppose there exists $C>0$ such that $\mathbb{E}[W(\mathbf{p},\mathbf{q})]>Cn^{k-1}$ for $n$ large. If (1) does not hold, we have shown in the proof of Theorem~\ref{critical} that each firm learns at most $o(n)$ ideas a.a.s. Then $\mathbb{E}[| T_i(\mathbf{p},\mathbf{q})|]$ is $o(n^{k-1})$ for each $i$. So $\mathbb{E}[| \bigcup_i T_i(\mathbf{p},\mathbf{q})|]$ is $o(n^k)$, and therefore welfare is $o(n^k)$. This gives a contradiction.
\end{proof}

\begin{proof}[Proof of Theorem~\ref{secrecy}]
We first characterize equilibria by showing we cannot have a supercritical and then subcritical sequence of investment equilibria. We then show there exists an investment equilibrium for $n$ large.

\textbf{Supercritical Case}: Suppose there is a supercritical sequence of investment equilibria propensities to learn $\beta_i$ for each $i$.

Passing to a subsequence if necessary, we can assume that all firms in the giant component learn $\widetilde{\alpha} n + o(n)$ ideas for some $\widetilde{\alpha}$ and that the number of firms learning all ideas learned by the giant component is $\alpha n + o(n)$ for some $\alpha$.\footnote{Because link probabilities are no longer symmetric within pairs, we do not assume that $\alpha =\widetilde{\alpha}$.} The argument is the same as in the proof of Theorem~\ref{critical}.

For each $i$, let $\alpha_i$ be the probability that firm $i$ learns all ideas learned by all firms in the giant component. Finally, we let $$\overline{\beta}=\frac{\sum_j \beta_j q_j^*}{\sum_j q_j^*}.$$ 
As $n\rightarrow \infty$, this converges to the derivative of the number of firms that learn from firm $i$ in $q_i$ divided by $\sum_{j\neq i} q_j^*$.

The first-order condition for firm $i$ then implies:
\begin{equation}\label{eq:hetbetaFOC}\delta \alpha_i\overline{\beta} \leq (1-\alpha_i) \beta_i \alpha + o(1).\end{equation}
To see this, suppose we increase $q_i/\sum_{j\neq i} q_j^*$ infinitessimally. We can condition on the event that no firm has learned $i$'s idea and all ideas learned by the giant component. The left-hand side is the probability that firm $i$ has learned indirectly from the giant component times the probability that a firm learns indirectly from $i$ after this increase. The right-hand side is the probability that firm $i$ has not learned indirectly from the giant component times the probability that firm $i$ learns indirectly from the giant component after this increase.

We have $\alpha_i = 1-e^{-\alpha \beta_i \delta \sum_{j \neq i} \iota(q_i^*,q_j^*)}.$ Substituting into equation (\ref{eq:hetbetaFOC}),
$$(1-e^{-\alpha \beta_i \delta \sum_{j \neq i} \iota(q_i^*,q_j^*)}) \overline{\beta} \leq e^{-\alpha \beta_i \delta \sum_{j \neq i} \iota(q_i^*,q_j^*)} \beta_i \alpha + o(1).$$
Therefore,
\begin{equation}\label{eq:totalinteractionineq}\delta \sum_{j \neq i} \iota(q_i^*,q_j^*) \leq \frac{\log(1+\alpha\beta_i/{\overline{\beta}})}{\alpha \beta_i}+o(1).\end{equation}

The expected number of firms from which firm $i$ learns indirectly is $\beta_i \delta \sum_{j\neq i}\iota(q^*_i,q^*_j)$. By equation (\ref{eq:totalinteractionineq}), this probability is bounded above $ \frac{\log(1+\alpha\beta_i/{\overline{\beta}})}{\alpha }+o(1)$. By the standard elementary inequality $\log(1+x)<x$, this is bounded above by $\beta_i/\overline{\beta}+o(1)$.

The expected number of firms that learn indirectly from $i$ is $$\sum_{j\neq i}\beta_j\iota(q^*_i,q^*_j)=\frac{\overline{\beta}}{\beta_i}\sum_{j\neq i}\beta_i\iota(q^*_i,q^*_j).$$
The right-hand side is the product of $\frac{\overline{\beta}}{\beta_i}$ and the expected number of firms from which firm $i$ learns indirectly, and therefore is at most
$$\frac{\overline{\beta}}{\beta_i} \cdot (\beta_i/\overline{\beta}+o(1))=1+o(1).$$
Since firm $i$ is arbitrary, this implies that each column sum of the matrix $(\delta \iota(q_i^*,q_j^*))$ is at most $1+o(1)$. Therefore, the spectral radius of this matrix is at most $1+o(1)$. But this contradicts our assumption that the sequence of equilibria is supercritical.

\textbf{Subcritical Case}: Suppose there is a subcritical sequence of investment equilibria. Introducing choices of secrecy, Lemma~\ref{lem:basicFOC} states that
$$\delta (\sum_{j\neq i} q_j^*\beta_j)\mathbb{E}\left[\binom{|I_i(\mathbf{p}^*,\mathbf{q}^*)|}{k-1}\right]= \frac{1}{n} \mathbb{E}\left[ \frac{\partial \binom{I_i(\mathbf{p}^*,(q_i,q_{-i}^*))}{k-1} }{\partial q_i}(q^*)\right]+o(\mathbb{E}\left[\binom{|I_i(\mathbf{p}^*,\mathbf{q}^*)|}{k-1}\right]).$$
The proof of the lemma remains the same.

We now let $X_j$ be i.i.d. random variables with distribution given by the number of ideas that firm $i$ would learn from a firm $j$ conditional on learning directly from $j$. That is, $X_j$ is distributed as the sum of a Bernoulli random variable with success probability $p_j^*$ (corresponding to direct learning) and a random variable with the distribution of $|I_j(\mathbf{p}^*,\mathbf{q}^*)|$ with probability $\delta$ and equal to zero otherwise.

By the same arguments as in the proof of Theorem~\ref{critical}, we have
$$\mathbb{E}\left[\binom{|I_i(\mathbf{p}^*,\mathbf{q}^*)|}{k-1}\right] = \sum_{\gamma \in \Gamma} \sum_{j_1,\hdots,j_{l(\gamma)}\neq i} \prod_{r =1}^{l(\gamma)} 
\beta_i q^*_iq^*_{j_r}\mathbb{E}\left[ \prod_{r =1}^{l(\gamma)} \binom{X_{j_r}}{\gamma_r}\right] + o(\mathbb{E}\left[\binom{|I_i(\mathbf{p}^*,\mathbf{q}^*)|}{k-1}\right]).$$
and also
$$\frac{1}{n} \mathbb{E}\left[ \frac{\partial \binom{I_i(\mathbf{p}^*,(q_i,q_{-i}^*))}{k-1} }{\partial q_i}(q^*)\right]=\frac{1}{n}\sum_{\gamma \in \Gamma} \sum_{j_1,\hdots,j_{l(\gamma)}\neq i} \left(\sum_{r =1}^{l(\gamma)}\beta_i q_{j_r}^* \prod_{r' \neq r}\beta_i q^*_iq^*_{j_r}\right)\mathbb{E}\left[ \prod_{r =1}^{l(\gamma)} \binom{X_{j_r}}{\gamma_r}\right] + o(\mathbb{E}\left[\binom{|I_i(\mathbf{p}^*,\mathbf{q}^*)|}{k-1}\right].$$

Therefore, $$\delta (\sum_{j \neq i} q_j^*\beta_j)\sum_{\gamma \in \Gamma} \sum_{j_1,\hdots,j_{l(\gamma)}\neq i} \prod_{r =1}^{l(\gamma)} q^*_iq^*_{j_r}\mathbb{E}\left[ \prod_{r =1}^{l(\gamma)} \binom{X_{j_r}}{\gamma_r}\right]\sim \sum_{\gamma \in \Gamma} \sum_{j_1,\hdots,j_{l(\gamma)}\neq i} \left(\sum_{r =1}^{l(\gamma)} q_{j_r}^* \prod_{r' \neq r} q^*_iq^*_{j_r}\right)\mathbb{E}\left[ \prod_{r =1}^{l(\gamma)} \binom{X_{j_r}}{\gamma_r}\right].$$
Note that the $\beta_i$ terms cancel. Rearranging,
$$\sum_{\gamma \in \Gamma} \sum_{j_1,\hdots,j_{l(\gamma)}\neq i} \prod_{r =1}^{l(\gamma)} q^*_iq^*_{j_r}\mathbb{E}\left[ \prod_{r =1}^{l(\gamma)} \binom{X_{j_r}}{\gamma_r}\right](\delta q_i^*(\sum_{j \neq i} q_j^*\beta_j)-l(\gamma) ) \sim 0.$$
In particular, we have
$$\delta q_i^*(\sum_{j \neq i} q_j^*\beta_j) \sim \mathbb{E}_{t \sim G_i(\mathbf{p}^*,\mathbf{q}^*)}[\tau(t)]$$
where the expectation is taken with respect to the appropriate distribution $G_i(\mathbf{p}^*,\mathbf{q}^*)$ over technologies. Because $\tau(t) \geq 1$ for all $t$, the limit superior of the expected number of firms that learn from $i$ is at least one.

So each column sum of the matrix $(\delta \iota(q_i^*,q_j^*))$ is at least $1+o(1)$. Therefore, the spectral radius of this matrix is at least $1+o(1)$. This contradicts our assumption that the sequence of equilibria is subcritical.

\textbf{Existence}: We have finitely many types $\beta_i$. Let $n(\beta_i)$ be the number of firms of type $\beta_i$.

We first fix levels of private investment $\mathbf{p}$ for each $n$ with $p_i \geq \frac12$ for all $i$. We will show that there exist levels of openness $\mathbf{q}$ such that  for all $i$ the choice $q_i$ is optimal given $\mathbf{p}$ and $q_{-i}$. Let $BR_i(\mathbf{p},q_{-i})$ denote the set of optimal $q_i$.\footnote{Note that unlike in the proof of Theorem~\ref{critical}, since the equilibrium is no longer symmetric, the best response $q_i$ can depend on others' levels of private investment.} We will restrict to $(\mathbf{p},q_{-i})$ such that each firm's action depends only on its type $\beta_j$.

We now show that $BR_i(\mathbf{p},q_{-i})$ is single-valued by treating several cases, beginning with a generalization of Lemma~\ref{lem:uniquebr}:
\begin{lemma}\label{lem:uniquebrgeneral}
There exists $\kappa>0$ such that for $n$ sufficiently large, for $(\mathbf{p},q_{-i})$ such that $(p_j,q_j)$ depend only on $j$'s type $\beta_j$ and the spectral radius $\lambda$ is in $(1-\kappa,1+\kappa)$, the correspondence $BR_i(\mathbf{p},q_{-i})$ is single-valued.
\end{lemma}
\begin{proof}[Proof of Lemma \ref{lem:uniquebrgeneral}]
Let $\epsilon>0$. We first claim that $\kappa$ sufficiently small and $M>0$ sufficiently large, we have $$\lim_n \mathbb{P}\left[|I_i(\mathbf{p},\mathbf{q})| >M \right]<\epsilon$$ for all $\mathbf{q}$ such that $\lambda \in (1-\kappa,1+\kappa)$. For any $\mathbf{q}'$ with $\lambda' <1+\kappa$, the random variable $|I_i(\mathbf{p},\mathbf{q}')|$ is first-order stochastically dominated by $|I_i(\mathbf{p},\mathbf{q}'\cdot \frac{1+\kappa}{\lambda })|$. So we can assume $\lambda=1+\kappa$.

We can bound the number of firms $j$ with a path from $j$ to $i$ in the indirect learning network above by the multi-type branching process with types corresponding to types $\beta_i$ and the number of successors of type $\beta_{j}$ of a node of type $\beta_i$ distributed as a Poisson random variable with mean $\beta_i \delta \iota(q_i,q_{j}) n(\beta_{j})$. By Theorem 2 of Section V.3 of \cite*{athreya1972branchingprocesses} and continuity of the extinction probability, we can take the extinction probability of this process to be at least $1-\epsilon/2$ by taking $\kappa$ sufficiently small.

Conditional on non-extinction, we can choose $M_0$ such that the number of nodes in the branching process is greater than $M_0$ with probability at most $\epsilon/4$. We can then choose $M$ such that when there are at most $M_0$ nodes in this branching process, we have $|I_i(\mathbf{p},\mathbf{q})| > M$ with probability at most $\epsilon/4$. Combining these bounds, we have $$\lim_n \mathbb{P}\left[|I_i(\mathbf{p},\mathbf{q})| >M \right]<\epsilon.$$

Now let $M'>0$. We next claim that we can choose $\kappa>0$ such that $$\liminf_n \mathbb{E}[{|I_i(\mathbf{p},\mathbf{q})|}] >M'$$ for all $q$ such that $\delta \iota(q,q)n \in (1-\kappa,1+\kappa)$

Fix any link probabilities $\mathbf{q}$ with $\lambda=1-\kappa$. For any $\mathbf{q}'$ with $\lambda' >1-\kappa$, the random variable $|I_i(\mathbf{p},\mathbf{q}')|$ first-order stochastically dominates $|I_i(\mathbf{p},\mathbf{q}'\cdot \frac{1-\kappa}{\lambda })|$. So it is sufficient to show that  $$\lim_{\lambda \rightarrow 0}\lim_{n \rightarrow \infty}\mathbb{E}[{|I_i(\mathbf{p},\mathbf{q})|}] \rightarrow \infty.$$

We can bound $|I_i(\mathbf{p},\mathbf{q})|$ below by the expected number of firms $j$ with a path from $j$ to $i$ in the indirect learning network. By the same argument as Theorem 11.6.1 of \cite*{alon2004probabilistic}, the limit of this quantity as $n\rightarrow \infty$ is equal to the number of nodes in a multi-type Poisson branching process with the same parameters. By Theorem 1 of Section V.4 of \cite*{athreya1972branchingprocesses}, this number of nodes converges to infinity as $\kappa \rightarrow 0$. This proves the claim.

Given these claims, the remainder of the proof of the lemma proceeds as in the proof of Lemma~\ref{lem:uniquebr}.\end{proof}

Next, we claim that when $n$ is sufficiently large, the best response $BR_i(\mathbf{p},q_{-i})$ for a firm $i$ is single-valued for $(\mathbf{p},\mathbf{q})$ such that $\lambda \geq 1+\kappa$. By compactness, it is sufficient to check this along any sequence of such $(\mathbf{p},\mathbf{q})$. Because $\lambda \geq 1+\kappa$, there exists a giant component of firms learning at least $\alpha n+o(n)$ ideas for some $\alpha>0$. We can assume without loss of generality that the giant component size converges.

Any best response $q_i$ will satisfy $\delta \iota(q_i,q_{-i})n < C$ for some $C$ independent of $n$. Any best response $q_i$ must satisfy the first-order condition after conditioning on no firm learning indirectly from $i$. Up to lower order terms, we can also condition on no firm having learned $i$'s idea and all ideas known to the giant component. Then the expected marginal value of an additional incoming link is strictly decreasing in $q_i$ while the expected marginal cost of an outgoing link is strictly increasing in $q_i$. So the first-order condition is satisfied at a unique value of $q_i$, and thus there is a unique best response.

Finally, we claim that for any $\underline{\lambda}>0$, the best response $BR_i(\mathbf{p},q_{-i})$ for a firm $i$ is single-valued for $n$ sufficiently large for $(\mathbf{p},\mathbf{q})$ such that $\underline{\lambda}<\lambda \leq 1-\kappa$. We can again check this along any sequence of such $(\mathbf{p},\mathbf{q})$. We want to show that the first-order condition for $q_i$ is uniquely satisfied, and we can check this conditioning on the event that no firm has learned indirectly from $i$.

Given action $q_i$, firm $i$ learns $I_i(\mathbf{p},(q_i,q_{-i})$ ideas. The expected marginal cost of an outgoing link is
\begin{equation}\label{eq:costoutgoing}\delta \mathbb{E}[\binom{|I_i(\mathbf{p},(q_i,q_{-i})|}{k-1}]+o(\mathbb{E}[\binom{|I_i(\mathbf{p},(q_i,q_{-i})|}{k-1}]).\end{equation}
As in the proof of Theorem~\ref{critical}, let $X$ be the random variable equal to the number of ideas that firm $i$ would learn from $j$ conditional on learning directly from $j$. The expected marginal value of an incoming link is
\begin{equation}\label{eq:valueincoming}\sum_{j=1}^{k-2}\mathbb{E}[\binom{|I_i(\mathbf{p},(q_i,q_{-i})|}{j} \binom{X}{k-j}]+o(\mathbb{E}[\binom{|I_i(\mathbf{p},(q_i,q_{-i})|}{k-1}]).\end{equation}

We want to show that expressions (\ref{eq:costoutgoing}) and (\ref{eq:valueincoming}) are equal for a unique $q_i$ for $n$ sufficiently large. The cardinality $|I_i(\mathbf{p},(q_i,q_{-i})|$ is first-order stochastically dominated by $|I_i(\mathbf{p},(q'_i,q_{-i})|$ when $q_i<q_i'$. 
 So it is sufficient to show that 
$$\frac{ \sum_{j=1}^{k-2}\mathbb{E}[\binom{y}{j} \binom{X}{k-j}]}{\delta \mathbb{E}[\binom{y}{k-1}]} >\frac{ \sum_{j=1}^{k-2}\mathbb{E}[\binom{y'}{j} \binom{X}{k-j}]}{\delta \mathbb{E}[\binom{y'}{k-1}]} $$
when $y'$ first-order stochastically dominates $y$.\footnote{Here the lower-order terms are well-behaved by the same arguments as in the proof of Theorem~\ref{critical}.}

The numerator is a linear combination of terms $\mathbb{E}[\binom{y}{j}]$ for $j\leq k-2$ with coefficients independent of $y$, so it is sufficient to show that $$ \frac{\mathbb{E}[\binom{y}{j}]}{\mathbb{E}[{\binom{y}{k-1}}]} >\frac{\mathbb{E}[\binom{y'}{j}]}{\mathbb{E}[{\binom{y'}{k-1}}]}$$
for $j<k-1$. This is equivalent to 
$$\mathbb{E}\left[\binom{y}{j}\right]\mathbb{E}\left[{\binom{y'}{k-1}}\right] >\mathbb{E}\left[\binom{y'}{j}\right]\mathbb{E}\left[{\binom{y}{k-1}}\right],$$
which holds because the binomial coefficients are strictly convex in $y \geq k$. We have proven our claim.
 
We have verified there is a unique optimal choice of $q_i$ when $\lambda > \underline{\lambda}>0$ for $n$ sufficiently large, and can now proceed to apply a fixed-point theorem. To do so, we first fix a sequence of $\mathbf{p}$ with $\liminf_n \min_i p_i >0$ for all $i$.

First suppose that a sequence of opponents' actions $(p_{-i}, q_{-i})$ is subcritical.\footnote{We extend our definition of criticality to the restriction of the random network to agents other than $i$.} Our analysis above showed that for $n$ large, the best response $BR(\mathbf{p}, q_{-i})$ has \begin{equation}\label{eq:BRbound}q_i(\sum_{j \neq i} q_j\beta_j) \in \left[\frac{1}{2\delta n}, \frac{k}{\delta n}\right],\end{equation}
where the upper bound follows from the fact that $\tau(t) \leq k-1$. So we can choose a uniform lower bound $\underline{\lambda}<1$ and upper bound $\overline{\lambda}>1$ on the spectral radius corresponding to actions $(BR(\mathbf{p}, q_{-i}))_i$.

Next, suppose that along a sequence of opponents' actions $(p_{-i}, q_{-i})$, the matrix of link probabilities for firms other than $i$ has spectral radius $\frac12 \leq \lambda \leq \overline{\lambda}$. Then firm $i$ can achieve a positive payoff by choosing $q_i$ such that $\beta_i\sum_{j\neq i} \iota(q_i,q_j)=1$, which we can bound below uniformly. On the other hand expected payoffs to firm $i$ vanish or are negative along any sequence such that $\beta_iq_i\rightarrow 0$. Thus we can choose $\epsilon>0$ such that $BR_i(\mathbf{p},q_{-i}) > \epsilon$. So lowering $\underline{\lambda}$ if necessary, we obtain a uniform lower bound $\underline{\lambda} < \frac12$ on the spectral radius corresponding to actions $(BR(\mathbf{p}, q_{-i}))_i$.

Increasing $\overline{\lambda}>1$ if necessary, we also have a uniform upper bound. Expression (\ref{eq:BRbound}) implies that we can increase $\overline{\lambda}>1$ without changing the lower bound.

We can conclude that for $n$ large, whenever the spectral radius of indirect learning probabilities with actions $\mathbf{q}$ is in $[\underline{\lambda},\overline{\lambda}],$ so is the spectral radius of indirect learning probabilities with actions $(BR_i(\mathbf{p},q_{-i})_i$. So by Brouwer's fixed-point theorem, for $n$ sufficiently large there exists a fixed point of $\mathbf{q} \mapsto (BR_i(\mathbf{p},q_{-i}))_i$. We will call this fixed point $\mathbf{q}(\mathbf{p}))$ to indicate the dependence on $\mathbf{p}$. It remains to show that there exists $\mathbf{p}$ such that $p_i$ is a best response under actions $(\mathbf{p},\mathbf{q}(\mathbf{p}))$ for all $i$.

Suppose that $\mathbf{p}_i\geq \frac12$ for all $i$ and all $n$. Our analysis of the subcritical and supercritical regions above extends immediately to this fixed point, as we did not rely on $\mathbf{p}$ being chosen optimally. Therefore, the sequence of outcomes $\mathbf{q}(\mathbf{p})$ must be critical. In particular, expected payoffs at $(\mathbf{p},\mathbf{q}(\mathbf{p}))$ converge to $\infty$ for all such sequences of $\mathbf{p}$.

The best response $p_i$ maximizes
$$p_i \mathbb{E}\left[ \sum_{t \in PPT_i(\mathbf{q}(\mathbf{p}))}\prod_{\substack{j \in t\\ j\neq i}} p_j\right]-c(p_i)$$
and therefore satisfies
\begin{equation}\label{eq:pfoc} c'(p_i) = \mathbb{E} \left[\sum_{t \in PPT_i(\mathbf{q}(\mathbf{p}))}\prod_{\substack{j \in t\\ j\neq i}} p_j\right].\end{equation}
Because $c(p_i)$ is strictly increasing and convex with $c'(0)\geq 0$ and $c(p)\rightarrow \infty$ as $p\rightarrow 1$, the solution set is convex and non-empty. Since $p_j\geq \frac12$ for all $j$ and the sequence $\mathbf{q}(\mathbf{p})$ is critical, the right-hand side of equation (\ref{eq:pfoc}) converges to infinity. So for $n$ large any optimal $p_i \geq \frac12$.

By Kakutani's fixed point theorem, for $n$ large there exists $\mathbf{p}\in [\frac12,1]^n$ such that $p_i$ is optimal for all $i$ under actions $(\mathbf{p},\mathbf{q}(\mathbf{p}))$. This is an investment equilibrium.
\end{proof}

\begin{proof}[Proof of Proposition~\ref{prop:concavity}]

We first show there is no sequence of supercritical symmetric investment equilibria for any $\rho>0$. To do so, we consider firm $i$'s choice of $q_i$ in the supercritical region. As in the proof of Theorem~\ref{critical},
$$\mathbb{E}\left[\binom{|I_i(\mathbf{p}^*,(q_i,q_{-i}^*))|}{k-1}^{\rho}\right]= (p^*)^k(1-(1-\delta \iota(q_i,q^*))^{\alpha(n-1)+o(n)})(\alpha(n-1))^{(k-1)\rho}+o(n^{(k-1)\rho}.$$

In particular, to solve for firm $i$'s choice of $q_i$ to first order, we need only consider technologies consisting of $i$'s private idea and $(k-1)$ ideas learned by the giant component. The probability that such a technology faces competition is
$$(1-\delta \cdot \iota(q_i,q^*) -(1-\delta)\cdot \iota(q_i,q^*) \cdot h(\alpha))^{n-1}+o(1),$$
where again $h(\alpha)$  is the fraction of firms $j$ such that all of $j$'s ideas are learned by some firm that learns all ideas known to the giant component. The term $\delta \cdot \iota(q_i,q^*)$ corresponds to the possibility of a firm $j$ indirectly learning all of firm $i$'s ideas. The term $(1-\delta)\cdot \iota(q_i,q^*) \cdot \alpha$ corresponds to the possibility of a firm learning $i$'s idea via a firm $j$ directly learning (but not indirectly learning) from $i$ and indirectly learning the ideas learned by the giant component.

Thus, we are looking for $q_i$ maximizing:
$$\mathbb{E}\left[\binom{|I_i(\mathbf{p}^*,(q_i,q_{-i}^*))|}{k-1}^{\rho}\right](1-\delta \cdot \iota(q_i,q^*) -(1-\delta)\cdot \iota(q_i,q^*) \cdot h(\alpha))^{n-1}.$$
This expression is equal to
$$\left((p^*)^k(1-(1-\delta \iota(q_i,q^*))^{\alpha(n-1)+o(n)})((\alpha(n-1))^{k-1}+o(n^{k-1}))\right)^{\rho}(1-\delta \cdot \iota(q_i,q^*) -(1-\delta)\cdot \iota(q_i,q^*) \cdot h(\alpha))^{n-1}.$$
Therefore, asymptotically the optimal $q_i$ will be a maximizer of
$$(p^*)^{\rho(k-1)+1}(1-\delta \iota(q_i,q^*))^{\alpha(n-1)+o(n)})((\alpha(n-1))^{k-1}+o(n^{k-1}))^{\rho}(1-\delta \cdot \iota(q_i,q^*) -(1-\delta)\cdot \iota(q_i,q^*) \cdot h(\alpha))^{n-1}.$$
The terms containing $q_i$ do not depend on $\rho$ to first order. Therefore, the optimization problem is the same as in Theorem~\ref{critical}, and the same argument shows there is no supercritical sequence of symmetric investment equilibria.

It remains to define $\underline{\rho}$ suitably and show there is no subcritical sequence of symmetric investment equilibria when for $\rho \geq \underline{\rho}$. We will choose $\underline{\rho}$ to satisfy the conditions in the following lemma:

\begin{lemma}\label{lem:concavity}
If $k=2$ and $\rho \geq 1$, then $\binom{y}{k-1}^{\rho}$ is convex in $y >0$. If $k>2$, there exists $\underline{\rho} <1$ such that $\binom{y}{k-1}^{\rho}$ is convex in $y > 0$ for all $\rho \geq \underline{\rho}.$
\end{lemma}
\begin{proof}[Proof of Lemma \ref{lem:concavity}]We can assume $y \geq k-1$. We have $\binom{y}{k-1}=\frac{ \prod_{j=0}^{k-2}(y-j)}{(k-1)!},$ and the right-hand side is defined for all real $y>0$. We will determine the sign of:
$$
 \frac{d^2 }{dy^2} \left(\left(\frac{ \prod_{j=0}^{k-2}(y-j)}{(k-1)!}\right)^{\rho}\right).$$
This expression has the same sign as
$$\frac{d}{dy} \left( \rho \left( \prod_{j=0}^{k-2}(y-j)\right)^{\rho-1} \sum_{i=0}^{k-2} \prod_{j \neq i} (y-j)  \right).$$
This derivative is equal to
\begin{equation}\label{eq:binomsecondderiv}\rho (\rho-1)\left( \prod_{j=0}^{k-2}(y-j)\right)^{\rho-2}\left(\sum_{i=0}^{k-2} \prod_{j \neq i} (y-j)  \right)^2 +\rho \left( \prod_{j=0}^{k-2}(y-j)\right)^{\rho-1} \sum_{i = 0}^{k-2} \sum_{i' \neq i} \prod_{j \neq i,i'} (y-j).\end{equation}
If $\rho \geq 1$, both the first and second term are non-negative for $y \geq k-1$, so expression (\ref{eq:binomsecondderiv}) is non-negative as well.

Suppose $k>2$. Expression (\ref{eq:binomsecondderiv}) has the same sign as \begin{equation}\label{eq:binomsecondderivalt} (\rho-1)\left(\sum_{i=0}^{k-2} \prod_{j \neq i} (y-j)  \right)^2 + \left( \prod_{j=0}^{k-2}(y-j)\right) \sum_{i = 0}^{k-2} \sum_{i' \neq i} \prod_{j \neq i,i'} (y-j).\end{equation}
The first term may be negative if $\rho < 1$, while the second term is positive for $y \geq k-1$. Both are polynomials of degree $2k-2$ in $y$. Therefore, we can choose $\underline{y}$ and $\underline{\rho} < 1$ sufficiently close to $1$ such that expression (\ref{eq:binomsecondderivalt}) is positive for $\rho > \underline{\rho}$ and $y > \underline{y}$.

We want the expression to be positive for $k-1 \leq y \leq \underline{y}$. There are finitely many values, and for each expression (\ref{eq:binomsecondderivalt})  is positive when $\rho$ is sufficiently close to one or at least one. Therefore, increasing $\underline{\rho}$ if needed, we find that expression (\ref{eq:binomsecondderivalt})  is positive for $\rho > \underline{\rho}$ and $y \geq k-1$. This proves the lemma.
\end{proof}

Let $\rho \geq \underline{\rho}$, where $\underline{\rho}=1$ when $k=2$ and $\underline{\rho}$ is chosen as in Lemma~\ref{lem:concavity} for $k>2$.

Suppose there exists a sequence of symmetric investment equilibria with $\limsup_n \delta  \iota({q}^*,{q}^*) n < 1$. Passing to a subsequence if necessary, we can assume that $\delta  \iota({q}^*,{q}^*) n$ converges.

We claim that for $n$ sufficiently large
\begin{equation}\label{eq:rhobigineq} \delta\frac{\partial \iota(q_i,q^*)}{\partial q_i}(q^*)\cdot \mathbb{E}\left[\binom{|I_i(\mathbf{p}^*,\mathbf{q}^*)|}{k-1}^{\rho}\right] < \frac{1}{n-1}\mathbb{E}\left[\frac{\partial \binom{|I_i(\mathbf{p}^*,(q_i,q^*)|}{k-1}^{\rho}}{\partial q_i}(q^*)\right]\end{equation}
for all $i$. Both sides of the inequality converge because $\delta  \iota({q}^*,{q}^*) n$ converges to a limit less than one.

Let $X_1$ be a random variable equal to $|I_i(\mathbf{p},\mathbf{q})|$ with probability $\delta$ and $0$ with probability $1-\delta$. Then the left-hand side of equation~(\ref{eq:rhobigineq}) is equal to
$\mathbb{E}\left[\binom{X_1}{k-1}^{\rho}\right]$
asymptotically.

Let $X_2$ be the random variable with distribution equal to the change in $|I_i(\mathbf{p},\mathbf{q})|$ if firm $i$ learned from an additional firm $j$ chosen uniformly at random. Then the right-hand side of equation~(\ref{eq:rhobigineq}) is equal to $$\mathbb{E}\left[\binom{|I_i(\mathbf{p},\mathbf{q})|+X_2}{k-1}^{\rho}-\binom{|I_i(\mathbf{p},\mathbf{q})|}{k-1}^{\rho}\right]$$
asymptotically.

In this case, with probability $1-\delta$, the firm $i$ only learns directly from firm $j$. With probability $\delta$, firm $i$ learns indirectly through firm $j$, and then learns
$$| I_j(\mathbf{p},\mathbf{q})| - |I_i(\mathbf{p},\mathbf{q}) \cap I_j(\mathbf{p},\mathbf{q})|$$
additional ideas.

The expected cardinality $|I_i(\mathbf{p},\mathbf{q}) \cap I_j(\mathbf{p},\mathbf{q})|$ is $o(1)$, by the same independence argument given in the proof of Lemma~\ref{lem:basicFOC}. Therefore, we can ignore the intersection term in computing the limit of the right-hand side of equation~(\ref{eq:rhobigineq}). Let $\widetilde{X}_2$ be the random variable with distribution equal to the number of ideas firm $i$ learns from firm $j$, including any ideas firm $i$ already knows, i.e., $X_2$ without this intersection term.

Then $\widetilde{X}_2$ first-order stochastically dominates $X_1$, and is one higher with non-vanishing probability. By Lemma~\ref{lem:concavity}, this implies
$$\mathbb{E}\left[\binom{X_1}{k-1}^{\rho}\right]<\mathbb{E}\left[\binom{|I_i(\mathbf{p},\mathbf{q})|+\widetilde{X}_2}{k-1}^{\rho}-\binom{|I_i(\mathbf{p},\mathbf{q})|}{k-1}^{\rho}\right]$$
for $n$ large. It follows that the same inequality holds for $n$ large with $X_2$ replacing $\widetilde{X}_2$, which proves the claim.

So along any sequence of symmetric investment equilibria with $\limsup \delta \iota(q^*,q^*)n < 1$, for $n$ sufficiently large
$$\delta \iota(q^*,q^*)n > \mathbb{E}_{t \in PT_i(\mathbf{p}^*,\mathbf{q}^*)}[\tau(t)]$$
for all $i$. The proof is the same as the proof of Lemma~\ref{lem:tauFOC}, with the approximate equality from Lemma~\ref{lem:basicFOC} replaced by the inequality from equation (\ref{eq:rhobigineq}).

In particular, $\delta \iota(q^*,q^*)n > 1$ for $n$ large, which contradicts the assumption of subcriticality. So any sequence of symmetric investment equilibria is critical.
\end{proof}

\begin{proof}[Proof of Proposition~\ref{prop:convexgeneral}]
The proof follows the same basic outline as the proof of Proposition~\ref{prop:concavity}, with the function $\binom{|I_i(\mathbf{p},\mathbf{q})|}{k-1}^{\rho}$ replaced by $\phi(|I_i(\mathbf{p},\mathbf{q})|)$.

We first show there is no sequence of supercritical symmetric equilibria with $\liminf_n p^*/n > 0$. To do so, we consider firm $i$'s choice of $q_i$ in the supercritical region. As in the proof of Theorem~\ref{critical}, the payoffs to firm $i$ are
$$\mathbb{E}\left[U_i(\mathbf{p}^*,(q_i,q_{-i}))\right]= (1-(1- \iota(q_i,q^*))^{\alpha(n-1)+o(n)})(1-\iota(q_i,q^*))^{n-1}p^*\phi\left(p^*(\alpha(n-1)+o(n^{k-1}))\right)-c(p^*)$$
when the giant component has size $\alpha n + o(n).$

We will bound $\phi(p^*(\alpha (n-1) + y)),$ where $y$ is $o(n)$. This expression is less than or equal to
$$\phi(p^*\alpha (n-1))+p^*y\phi'(p^*(\alpha (n-1) + y)).$$
By the assumption of non-vanishing investment, we have $\lim_n p^*>0$. By our assumption that $\frac{\phi(x_j)}{\phi(x_j')}\rightarrow 1$ when $x_j/x'_j \rightarrow 1$,
we can conclude
$$\phi(p^*(\alpha (n-1) + y))=\phi(p^*\alpha (n-1))+o(\phi(p^*\alpha (n-1))).$$

Therefore, $q_i$ is chosen to maximize 
$$(1-(1- \iota(q_i,q^*))^{\alpha(n-1)+o(n)})(1-\iota(q_i,q^*))^{n-1}+o(1).$$
The maximization is the same as in Theorem~\ref{critical} with $\delta=1$, and the same calculation shows there is no supercritical sequence of symmetric investment equilibria.

The proof that there is no the subcritical sequence of symmetric equilibria with $\liminf_n p^*/n > 0$ is the same as in Proposition~\ref{prop:concavity}, with $\binom{|I_i(\mathbf{p},\mathbf{q}|}{k-1}^{\rho}$ replaced by $\phi(|I_i(\mathbf{p},\mathbf{q}|)$. We no longer need to prove Lemma~\ref{lem:concavity}, as we assume that $\phi(\cdot)$ is convex.
\end{proof}

\begin{proof}[Proof of Proposition~\ref{prop:competition}]

(i): We begin with a lemma, which we will use to show there cannot be a critical sequence of investment equilibria.
\begin{lemma}\label{lem:criticaltau}
For any critical sequence of symmetric actions with $p>0$,
$$\lim_{n \rightarrow \infty} \mathbb{E}_{t \in PT_i(\mathbf{p},\mathbf{q})}[\tau(t)]=1.$$
\end{lemma}
\begin{proof}[Proof of Lemma~\ref{lem:criticaltau}]
We can assume without loss of generality that $p$ is bounded away from zero, because $\mathbb{E}_{t \in PT_i(\mathbf{p},\mathbf{q})}[\tau(t)]$ does not depend on the value of $p$ as long as $p$ is non-zero.

Let $\epsilon>0$. The probability that firm $i$ learns from $d$ firms decays exponentially in $d$.  By Rosenthal's inequality \citep*{rosenthal1970subspaces}, the payoffs to learning from ${d}$ firms grow at most at a polynomial rate in $d$. Thus we can choose $\overline{d}$ such that the contribution to $\mathbb{E}_{t \in PT_i(\mathbf{p},\mathbf{q})}[\tau(t)]$ from the event that firm $i$ learns from more than $\overline{d}$ other firms is at most $\epsilon$ for $n$ large. Since $\epsilon$ is arbitrary, we can restrict our analysis to the event that firm $i$ learns from at most $\overline{d}$ other firms.

We claim that as $n \rightarrow \infty$, we have $\mathbb{E}[{|I_i(\mathbf{p},\mathbf{q})|}] \rightarrow \infty.$ Let $x >0$ and define $q(x)$ by $\iota(q(x),q(x))=\frac{1-x}{\delta n}$. For any $q  < q'$, the random variable $|I_i(\mathbf{p},\mathbf{q})|$ is first-order stochastically dominated by $|I_i(\mathbf{p},\mathbf{q}')|$. So it is sufficient to show that  $$\lim_{x \rightarrow 0}\lim_{n \rightarrow \infty}\mathbb{E}[{|I_i(\mathbf{p},\mathbf{q}(x))|}] \rightarrow \infty.$$

We can bound $|I_i(\mathbf{p},\mathbf{q})|$ below by the expected number of firms $j$ with a path from $j$ to $i$ in the indirect learning network. By Theorem 11.6.1 of \cite*{alon2004probabilistic}, the limit of this quantity as $n\rightarrow \infty$ is equal to the number of nodes in a Poisson branching process with parameter $1-x$. As $x\rightarrow 0$, this number of nodes converges to infinity. This proves the claim.

The proof of Lemma~\ref{lem:criticaltau} will also use the following lemma, which states that learning a large number of ideas at a critical sequence of actions is rare for $n$ large:
\begin{lemma}\label{lem:rareevents}
Let $\omega(n) \rightarrow \infty$. For any critical sequence of symmetric actions with $p>0$, $$\mathbb{P}\left[|I_i(\mathbf{p},\mathbf{q})| > \omega(n)\right] \rightarrow 0$$
as $n\rightarrow \infty$.
\end{lemma}
\begin{proof}[Proof of Lemma \ref{lem:rareevents}]
Let $\epsilon>0$. We want to prove that $\mathbb{P}\left[|I_i(\mathbf{p},\mathbf{q})| > \omega(n)\right] <\epsilon$ for $n$ large.  This holds by the same argument as in the proof of Lemma~\ref{lem:uniquebr}.
\end{proof}

We now return to the proof of Lemma~\ref{lem:criticaltau}. Choose $\omega(n)\rightarrow \infty$ such that $$\frac{\omega(n)}{\mathbb{E}[{|I_i(\mathbf{p},\mathbf{q})|}]} \rightarrow  0.$$

We have assumed that $i$ learns from at most $\overline{d}$ other firms. We can order these firms from $1$ to $\overline{d}$. For each of these firms $j$, let the additional ideas $AI_{j}$ be the set of ideas that firm $i$ learns from firm $j$ and has not learned from any previous firm $1,\hdots,j-1$ in our ordering.

We claim that as $n \rightarrow \infty$, a vanishing share of proprietary technologies include ideas that firm $i$ learns only from firms $j$ with $AI_j\leq \omega(n).$
The number of such ideas is bounded above by $\omega(n) \overline{d}$. So the number of proprietary technologies including at least one such idea is bounded above by
\begin{equation}\label{eq:notonlyrare}\mathbb{E}\left[\binom{|I_i(\mathbf{p},\mathbf{q})|+\overline{d}\omega(n)}{k-2}\cdot (\overline{d}\omega(n))\right]=\overline{d}\omega(n)\mathbb{E}\left[\binom{|I_i(\mathbf{p},\mathbf{q})|+\overline{d}\omega(n)}{k-2}\right] ,\end{equation}
while the total number of proprietary technologies is on the same order as
\begin{equation}\label{eq:onlyrare}\mathbb{E}\left[\binom{|I_i(\mathbf{p},\mathbf{q})|}{k-1 }\right]\geq\mathbb{E}\left[\binom{|I_i(\mathbf{p},\mathbf{q})|}{k-2 }\right]\frac{\mathbb{E}\left[|I_i(\mathbf{p},\mathbf{q})|\right]}{k-1}.\end{equation}
Since $$\frac{\omega(n)}{\mathbb{E}[{|I_i(\mathbf{p},\mathbf{q})|}]} \rightarrow  0,$$
the quotient of expression (\ref{eq:notonlyrare}) divided by expression (\ref{eq:onlyrare}) vanishes as $n\rightarrow \infty$.

Let $\epsilon>0$. For $n$ sufficiently large, Lemma~\ref{lem:rareevents} implies that the probability that $AI_j > \omega(n)$ is at most $\epsilon$. We will show that the contribution to $\mathbb{E}_{t \in PT_i(\mathbf{p},\mathbf{q})}[\tau(t)]$ from the event that $AI_j > \omega(n)$ for more than one $j$ can be taken to be small.

We condition on the event that $AI_j > \omega(n)$ for at least one $j$. Then the probability that $AI_j > \omega(n)$ for at least one other $j$ is bounded above by $\overline{d}\epsilon$, while the expected number of proprietary technologies increases by at most a constant multiplicative factor (depending on $\overline{d}$) in this case. Since $\epsilon$ can be taken to be arbitrarily small, it follows that the contribution to $\mathbb{E}_{t \in PT_i(\mathbf{p}^*,\mathbf{q}^*)}[\tau(t)]$ from the event that $AI_j > \omega(n)$ for more than one $j$ vanishes as $n\rightarrow \infty$.

The remaining technologies in $PT_i(\mathbf{p},\mathbf{q})$ consist of firm $i$'s private idea and $k-1$ ideas learned from a single firm $j$. We thus have $\tau(t)=1$ for each of the remaining technologies $t \in PT_i(\mathbf{p},\mathbf{q}).$ This shows that $ \mathbb{E}_{t \in PT_i(\mathbf{p},\mathbf{q})}[\tau(t)] \rightarrow 1$, which proves the lemma.
\end{proof}

We will show that if $\liminf_n p>0$ and $\limsup_n \iota(q,q) n \delta\leq 1,$ then given symmetric actions $(\mathbf{p},\mathbf{q})$,
$$\frac{\partial U_i(\mathbf{p},\mathbf{q})}{\partial q_i}(q)>0$$
for $n$ sufficiently large. In words, given symmetric actions in the subcritical or critical region, increasing $q_i$ would increase payoffs. We can assume without loss of generality that $p$ is bounded away from zero, because the sign of this derivative is independent of $p$.

Recall $D_i(\mathbf{q})$ is the set of firms $j$ such that there is a path from $i$ to $j$ in the indirect-learning network. We claim that when $\limsup_n \iota(q,q) \delta n \leq 1$, a.a.s. a random technology $t \in T_i(\mathbf{p},\mathbf{q})$ is known by $|D_i(\mathbf{q})|$ other firms.
 
We can write $\iota(q,q) \delta n=1+\epsilon$, where $\limsup_n \epsilon \leq 0$. Let $\overline{\epsilon} = \max(\epsilon, n^{-1/3}\log n ).$ By the `No Middle Ground' claim from p. 210-211 of \cite*{alon2004probabilistic}, the probability that a given node in an undirected random network with link probability $\frac{1+\epsilon}{n}$ is contained in a component of cardinality at least $\overline{\epsilon}n$ is at most $n^{-2k-1}$ for $n$ large.

A standard correspondence states that the size of the component containing a given node in an undirected random graph first-order stochastically dominates the number of nodes reachable by a path from that node in a directed random graph with the same link probability (see for example \citealp*{luczak1990phase}). So the probability that a given idea is learned indirectly by more than  $\overline{\epsilon}n$ firms is at most $n^{-2k-1}$ for $n$ large. Thus, we can choose a constant such that the probability that any idea is learned indirectly by more than  $\overline{\epsilon}n$ firms is at most $n^{-2k}$ for $n$ large. Since there are $\binom{n}{k}$ potential technologies, we can condition on the event that no idea is learned indirectly by more than   $\overline{\epsilon}n$ firms.

Now, choose $t \in T_i(\mathbf{p},\mathbf{q})$ and let $j \in T_i(\mathbf{p},\mathbf{q})$. The number of firms that learn idea $j$ from each firm that indirectly learns $j$ is bounded above by a Poisson random variable with parameter $1/\delta$ (by the same argument as in the proof of Lemma~\ref{lem:basicFOC}). So we can assume that at most $C \overline{\epsilon}n $ firms learn the idea for some constant $C>0$. For a firm $j' \notin D_i(\mathbf{q})$ to know all ideas in $t$, that firm must learn $j$ and directly learn $i$. Each of the $C \overline{\epsilon} n$ firms that learn $j$ will directly learn $i$ with probability at most $\frac{1+\epsilon}{\delta n}$, so the probability that any of these firms directly learns $i$ vanishes asymptotically. This proves the claim.

Thus, the payoff to firm $i$ is
$$U_i(\mathbf{p},\mathbf{q}) =  \mathbb{E}\left[f(|D_i(\mathbf{q})|)|T_i(\mathbf{p},\mathbf{q})|\right]-c(p_i)+o(\mathbb{E}\left[|PT_i(\mathbf{p},\mathbf{q})\right]).$$

Thus,
$$\frac{\partial U_i(\mathbf{p},\mathbf{q})}{\partial q_i}(q)=  \mathbb{E}\left[\frac{\partial f(|D_i(\mathbf{q})|)|}{\partial q_i}(q)|T_i(\mathbf{p},\mathbf{q})|\right]+ \mathbb{E}\left[\frac{\partial |T_i(\mathbf{p},\mathbf{q})|}{\partial q_i}(q)f(|D_i(\mathbf{q})|)\right]+o(\mathbb{E}\left[|PT_i(\mathbf{p},\mathbf{q})\right]).$$

We claim this is equal to
\begin{equation}\label{eq:payoffderiv}  \mathbb{E}\left[\frac{\partial f(|D_i(\mathbf{q})|)|}{\partial q_i}(q)\right]\mathbb{E}\left[|T_i(\mathbf{p},\mathbf{q})|\right]+ \mathbb{E}\left[\frac{\partial |T_i(\mathbf{p},\mathbf{q})|}{\partial q_i}(q)\right]\mathbb{E}\left[f(|D_i(\mathbf{q})|)\right]+o(\mathbb{E}\left[|PT_i(\mathbf{p},\mathbf{q})\right]).\end{equation}
The relevant random variables are independent conditional on the event that $$I_i(\mathbf{p},\mathbf{q})\cap D_i(\mathbf{q}) = \emptyset$$
and this intersection remains empty after adding an additional incoming or outgoing link.
Because $\limsup_n \iota(q,q) n \delta\leq 1,$
this occurs asymptotically almost surely. We must show the contributions to $T_i(\mathbf{p},\mathbf{q})$ from the vanishing probability event that this intersection is non-empty vanish as $n \rightarrow \infty.$ This follows from the bounds on the probability of this event in the `No Middle Ground' claim from p. 210-211 of \cite*{alon2004probabilistic}, and the argument is the same as above.

Because $f$ is non-negative with $f(0)=1$ and $f(1) > 0$, we can choose $\epsilon > 0$ such that
$$-\frac{1}{q n}\cdot\mathbb{E}\left[\frac{\partial f(|D_i(\mathbf{q})|)|}{\partial q_i}(q)\right] < \mathbb{E}\left[f(|D_i(\mathbf{q})|)\right] - \epsilon$$
for all $n$. Here, the left-hand side is equal to the decrease in $f(|D_i(\mathbf{q})|)$ when an additional firm learns from firm $i$, which is at most $f(|D_i(\mathbf{q})|)$ and will be smaller with non-vanishing probability.

At a subcritical sequence of actions, we have
$$\limsup_n \iota(q,q)  \delta n \leq \liminf_n \mathbb{E}_{t \in PT_i(\mathbf{p},\mathbf{q})}[\tau(t)].$$
By the same argument used to prove Lemma~\ref{lem:tauFOC}, this implies that
$$\mathbb{E}\left[|T_i(\mathbf{p},\mathbf{q})|\right]< \frac{1}{qn}\cdot\mathbb{E}\left[\frac{\partial |T_i(\mathbf{p},\mathbf{q})|}{\partial q_i}(q)\right]$$
for all $n$ sufficiently large. At a critical sequence of actions, Lemma~\ref{lem:criticaltau} shows that
$$\lim_{n \rightarrow \infty} \mathbb{E}_{t \in PT_i(\mathbf{p},\mathbf{q})}[\tau(t)]=1.$$
So at a critical sequence of actions, it follows from Lemma~\ref{lem:criticaltau} and the definition of $\tau(t)$ that
$$\mathbb{E}\left[|T_i(\mathbf{p},\mathbf{q})|\right]\sim  \frac{1}{qn}\cdot\mathbb{E}\left[\frac{\partial |T_i(\mathbf{p},\mathbf{q})|}{\partial q_i}(q)\right].$$

In either case, substituting into expression (\ref{eq:payoffderiv}) we obtain
$$\frac{\partial U_i(\mathbf{p},\mathbf{q})}{\partial q_i}(q)>0$$
in the subcritical or critical region for $n$ sufficiently large. This proves the claim, so any sequence of symmetric investment equilibria is supercritical.

(ii): Suppose there exists a critical or supercritical sequence of symmetric investment equilibria. We claim that 
$$\frac{\partial U_i(\mathbf{p}^*,(q_i,{q}^*_{-i}))}{\partial q_i}(q^*)<0,$$
which will give a contradiction. We can again assume without loss of generality that $p$ is bounded away from zero, because the sign of this derivative is independent of $p$.

First suppose there exists a critical sequence of investment equilibria. Lemma~\ref{lem:criticaltau} shows that
$$\lim_{n \rightarrow \infty} \mathbb{E}_{t \in PT_i(\mathbf{p}^*,\mathbf{q}^*)}[\tau(t)]=1.$$
As a consequence,
$$\mathbb{E}\left[\binom{|I_i(\mathbf{p}^*,\mathbf{q^*})|}{k-1}\right]\sim \frac{1}{qn}\cdot\frac{\partial \mathbb{E}\left[\binom{|I_i(\mathbf{p}^*,{(q_i,q_{-i}^*)})|}{k-1}\right]}{\partial q_i}(q^*).$$
Therefore,
$$\frac{\partial \mathbb{E}\left[|PT_i(\mathbf{p}^*,{(q_i,q_{-i}^*)})|\right]}{\partial q_i}(q^*)\sim0.$$
On the other hand,
$$\frac{\partial \mathbb{E}\left[|T_i(\mathbf{p}^*,{(q_i,q_{-i}^*)})|\right]}{\partial q_i}(q^*)>0$$
for $n$ large, and increasing $q_i$ weakly increases the number of firms $f(m)$ who know each technology. Since firm $i$ receives negative profits from each $t \in T_i(\mathbf{p}^*,{(q_i,q_{-i}^*)})$ that is not proprietary, firm $i$'s profits are decreasing in $q_i$ at $q_i=q^*$.

We next suppose there exists a supercritical sequence of equilibria. Let $m$ be the number of firms other than $i$ who learn idea $i$ and all ideas known to all firms that learn from the giant component. The expected profits from choosing $q_i$ are
$$(p^*)^k\binom{\alpha n}{k-1}(1- (1-\delta \iota(q_i,q^*))^{\alpha n})f(\mathbb{E}_{t \in T_i(\mathbf{p}^*,{(q_i,q_{-i}^*)})}[m])+o(n^{k-1}),$$
where $\alpha$ is the share of ideas learned by all firms in the giant component.

Conditional on the event that $m>0$, increasing $q_i$ will weakly decrease expected payoffs because increasing $m$ and increasing $|I_i(\mathbf{p}^*,(q_i,q^*_{-i}))|$ both weakly decrease payoffs. Moreover, this increase is strict, because the probability of learning the ideas in the giant component is strictly higher under higher $q_i$.

Now consider the event that $m=0$. We showed in the proof of Theorem~\ref{critical} that when $f(m)=0$ for all $m>0$, the increase in expected payoffs from an additional incoming link is less than the decrease in expected payoffs from an additional outgoing link. Since decreasing $f(m)$ for $m>0$ does not change the effect of an additional incoming link but decreases the expected payoffs from an additional outgoing link, it follows that when $q_i=q^*$, increasing $q_i$ will weakly decrease expected payoffs conditional on the event $m=0$.

Therefore, when $q_i=q^*$, increasing $q_i$ will weakly decrease expected payoffs unconditionally, which gives our contradiction. We have checked the critical and subcritical cases, so this completes the proof of Proposition~\ref{prop:competition}.
\end{proof}

\begin{proof}[Proof of Corollary~\ref{cor:competition}]
Recall that the discovery rate $D_n(\mathbf{p}^*,\mathbf{q}^*)\rightarrow 0$ at any subcritical sequence of equilibria while $\liminf_n D_n(\mathbf{p}^*,\mathbf{q}^*)> 0$ at any supercritical sequence of equilibria with non-vanishing investment. Therefore the corollary follows from Proposition~\ref{prop:competition}.
\end{proof}

\begin{proof}[Proof of Proposition~\ref{prop:competitionquant}]
(i): We first show that we can choose $x^*$ such that for $x<x^*$, an investment equilibrium exists for $n$ large.

We first show that for $n$ sufficiently large and $x$ sufficiently small, there exists a unique $q^*$ such that the first-order condition for the level of openness $q_i$ is satisfied when all firms choose $q_j=q^*$.
Given a technology $t$ and a firm $i$, let $m_i(t)$ be the number of other firms that know all ideas in technology $t$. Let $\Delta(\mathbb{E}_{t \in T_i(\mathbf{p}^*,\mathbf{q}^*)}[f_x(m_i(t))])$ be the expected decrease in $\mathbb{E}_{t \in T_i(\mathbf{p}^*,\mathbf{q}^*)}[f_x(m_i(t))]$ from an additional outgoing link from firm $i$.

Firm $i$ learns $\alpha n +o(n)$ ideas if $i$ learns from the giant component and $o(n)$ ideas otherwise. Asymptotically, the first-order condition implies
$$ \delta \alpha (1-\alpha) \mathbb{E}_{t \in T_i(\mathbf{p}^*,\mathbf{q}^*)}[f_x(m_i(t))]  \binom{\alpha n}{k-1} \sim \alpha  \Delta (\mathbb{E}_{t \in T_i(\mathbf{p}^*,\mathbf{q}^*)}[f_x(m_i(t))])\binom{\alpha n}{k-1}  .$$
On the left-hand side, the term $ \delta \alpha (1-\alpha) $ is the probability that firm $i$ has not yet learned the ideas learned by the giant component, but does with an additional incoming link. On the right-hand side, the term $\alpha$ is the probability that firm $i$ has learned the ideas learned by the giant component.

For $\alpha>0$, we obtain
$$ \delta  (1-\alpha) \mathbb{E}_{t \in T_i(\mathbf{p}^*,\mathbf{q}^*)}[f_x(m_i(t))]  \sim   \Delta(\mathbb{E}_{t \in T_i(\mathbf{p}^*,\mathbf{q}^*)}[f_x(m_i(t))]).$$
The contribution to the right-hand side from an outgoing link with only direct learning goes to zero as $\alpha \rightarrow 0$. The decrease in $f_x(m_i(t))$ is at most $f_x(m_i(t))$, and when $x>0$ is less than $f_x(m_i(t))$ with positive probability. So for $n$ sufficiently large, the right-hand side is less than the left-hand side when $x>0$ and $\alpha$ is sufficiently small.

By the assumption that $g(m) \rightarrow 0$, we can choose $\epsilon>0$ such that the marginal value of increasing $q_i$ at a symmetric action $q^*$ will be negative for $\alpha \geq 1-\epsilon$ and $n$ sufficiently large. By the intermediate value theorem, when $n$ is sufficiently large the first-order condition for $q_i$ is satisfied when all firms choose $q_j$ equal to some $q^*$. We next show this $q^*$ is unique.

We claim that for $x$ sufficiently small, the derivative of \begin{equation}\label{eq:margvaluex} \Delta (\mathbb{E}_{t \in T_i(\mathbf{p}^*,\mathbf{q}^*)}[f_x(m_i(t))])- \delta  (1-\alpha) \mathbb{E}_{t \in T_i(\mathbf{p}^*,\mathbf{q}^*)}[f_x(m_i(t))]  \end{equation}
in $\alpha$ is at least $C$ for some $C>0$ when $\alpha<1-\epsilon$ and $n$ is sufficiently large. When $x=0$, expression~(\ref{eq:margvaluex}) is at least
$$( \delta - \delta(1-\alpha)) e^{-\iota(q^*,q^*)h(\alpha)n},$$
where $h(\alpha)$ is the share of firms whose idea is known to some firm that has learned all ideas known to the giant component. So the claim holds for $x=0$ for with $C=\delta\epsilon e^{-\iota(q^*,q^*)h(1-\epsilon)n}$, where $q^*$ is the action such that $\alpha=1-\epsilon$. By continuity, the claim holds for $x$ sufficiently small. Since expression~(\ref{eq:margvaluex}) must be equal to zero (up to lower order terms) whenever the first-order condition holds, the first-order condition is uniquely satisfied for $n$ sufficiently large.

To show the choice of openness $q^*$ satisfying the first-order condition is optimal, we must show that $q^*$ is a best response for $x>0$ sufficiently small when $n$ is sufficiently large. There is not a positive boundary solution for $x>0$ by our assumption that $g(m)\rightarrow 0$, so it is sufficient to show that the first-order condition is uniquely satisfied. It is sufficient to analyze $q_i \leq \overline{q}$ for some $\overline{q}>0$.

Suppose other firms choose $q_{-i}$ and the giant component size is $\alpha$. The marginal value of an additional incoming link is
$$ \delta \alpha e^{- \delta \iota(q_i,q_{-i}) \alpha n} \mathbb{E}_{t \in T_i(\mathbf{p}^*,(q_i,q_{-i}))}[f_x(m_i(t))]  \binom{\alpha n}{k-1} + o(n^{k-1}).$$
The marginal cost of an additional outgoing link is 
$$(1-e^{- \delta \iota(q_i,q_{-i}) \alpha n}) \mathbb{E}_{t \in T_i(\mathbf{p}^*,(q_i,q_{-i}))}[f_x(m_i(t))]  \binom{\alpha n}{k-1} + o(n^{k-1}).$$

We claim that for $x$ sufficiently small, these expressions are equal at a unique $q_i$ when $n$ is sufficiently large. We showed in the proof of Theorem~\ref{critical} that the derivative of their difference is bounded away from zero when $x=0$, and the same is true for $x$ sufficiently small by continuity. This proves the claim and establishes existence.

(ii): The analysis in part (i) showed that asymptotically, the giant component size $\alpha$ is characterized by $$ \delta  (1-\alpha) \mathbb{E}_{t \in T_i(\mathbf{p}^*,\mathbf{q}^*)}[f_x(m_i(t))]  \sim   \Delta(\mathbb{E}_{t \in T_i(\mathbf{p}^*,\mathbf{q}^*)}[f_x(m_i(t))]).$$
So $\alpha$ converges along any sequence of equilibria with non-vanishing investment. Because $p^*\rightarrow 1$ at a sequence of equilibria with non-vanishing investment, the limit of the discovery rate $\lim_n D(\mathbf{p}^*,\mathbf{q}^*)$ exists and is a strictly increasing function of $\lim_n \alpha$.

(iii): We will show that $\lim_n \alpha$ is increasing in $x$ for $x<x^*$, where $x^*$ is chosen so that the properties above hold. We define $\widetilde{f}_x(m) = f_x(m)/x$, so that $\widetilde{f}_x(0)$ varies in $x$ while $\widetilde{f}_x(m)$ is constant in $x$ for $m>0$.\footnote{Note that we no longer require the normalization $f(0)=1$ on payoffs.} The analysis under payoffs $f_x(m)$ and $\widetilde{f}_x(m)$ is the same up to rescaling the cost function $c( \cdot )$, which does not effect $\lim_n \alpha$. So we can show $\lim_n \alpha$ is increasing in $x$ under payoffs $\widetilde{f}_x(m)$.

Fix $x<x^*$ and let $\alpha$ be the giant component size with payoffs $\widetilde{f}_x(m)$, so that $\alpha$ satisfies
$$ \delta  (1-\alpha) \mathbb{E}_{t \in T_i(\mathbf{p}^*,\mathbf{q}^*)}[\widetilde{f}_x(m_i(t))]  \sim   \Delta (\mathbb{E}_{t \in T_i(\mathbf{p}^*,\mathbf{q}^*)}[\widetilde{f}_x(m_i(t))]),$$
Increasing $x$ strictly increases the left-hand side and strictly decreases the right-hand side (because only the monopoly payoffs are affected), so for $n$ sufficiently large and $x<x'<x^*$ the difference
$$ \delta  (1-\alpha) \mathbb{E}_{t \in T_i(\mathbf{p}^*,\mathbf{q}^*)}[\widetilde{f}_{x'}(m_i(t))]  -  \Delta (\mathbb{E}_{t \in T_i(\mathbf{p}^*,\mathbf{q}^*)}[\widetilde{f}_{x'}(m_i(t))])$$
is non-vanishing. At the equilibrium action, this difference must vanish.

Since we showed in part (i) that this difference is decreasing in the equilibrium action for $x<x^*$ and $n$ sufficiently large, the limit of the giant component size at the equilibrium action under payoffs $\widetilde{f}_{x'}(m)$ must be greater than $\alpha$. Thus $\lim_n \alpha$ is increasing in $x$ when $x<x^*$.

When the giant component learns $\alpha+o(n)$ ideas, the discovery rate is $(p^*)^k\alpha^{k}+o(1)$. Since $p^*\rightarrow 1$, the limit of the discovery rate $\lim_n D(\mathbf{p}^*,\mathbf{q}^*)$ is a strictly increasing function of $\lim_n \alpha$. So the limit of the discovery rate is also increasing in $x$.
\end{proof}

\begin{proof}[Proof of Proposition~\ref{prop:patent}]
(i): Suppose there exists a sequence of investment equilibrium with $n \rightarrow \infty$.

We first consider the first-order condition for $q_i(1)$. The expected payoff to firm $i$ with a patent when $k=2$ and link costs are $\epsilon$ is
$$p_ip^* q^*(0)q_i(1) (1-b)(n-1) - \epsilon (q_i(1)(q^*(0)(1-b)+q^*(1)b)(n-1)-c(p_i).$$
The first term now does not depend on whether other firms learn the ideas involved in the technologies that firm $i$ can produce. The second term is the expected link cost.

The expected payoff is linear in $q_i(1),$ so the coefficient of $q_i(1)$ must be non-negative at any investment equilibrium. Therefore
$$(p^*)^2q^*(0)(1-b) \geq \epsilon(q^*(0)(1-b)+q^*(1)b) .$$
If equality holds, then firms with patents are indifferent to all choices of interaction rates. But then firms without patents would not choose positive interaction rates, which they must at any investment equilibrium. So the inequality is strict.

Because payoffs are strictly increasing in $q_i(1)$ on $[0,1]$, we have $q^*(1)=1$. Thus the inequality
$$(p^*)^2q^*(0)(1-b) > \epsilon(q^*(0)(1-b)+q^*(1)b)$$
implies that $\liminf_n q^*(0)$ is positive. But if all interaction rates are bounded below by constants, the probability that a firm $j$ without a patent receives monopoly profits from a given technology $t$ decays exponentially. Since firm $j$'s link costs are linear in $n$, firm $j$'s expected payoff is negative. This cannot occur at equilibrium, so we have a contradiction.


(ii): It is weakly dominant for all firms to choose $q^*(1)=1$, and is strictly optimal at any investment equilibrium. So for any $\delta > 0$, almost surely all firms with patents are in the same component of the indirect-learning network. All firms in this component learn $\alpha n$ ideas for some $\alpha > 0$.

Suppose that $q^*(0)n \rightarrow \infty$ and consider firm $i$ without a patent. The probability that no firm with a patent learns indirectly from firm $i$ decays exponentially in $q^*(0)n$, and so profits must be $o(n^{k-1})$. But firm $i$ could receive higher profits by deviating to choose $q_i(0)=\frac{1}{\delta n}$, as this would give profits of order $n^{k-1}$. So it must be the case that $q^*(0)$ is $O(1/n)$, and thus $$\iota(q^*(0),q^*(0))=(q^*(0))^2 = O(1/n^2)$$
as desired.
\end{proof}

\section{Direct Learning}\label{sec:direct}

We now analyze the model from Section \ref{sec:model} in the case $\delta=0$. Then firms can only learn directly from other firms, and not indirectly.
\begin{proposition}\label{prop:direct}
When $\delta = 0$, there exists a symmetric investment equilibrium for $n$ large, and at any sequence of symmetric investment equilibria $$\lim_n \iota(q^*,q^*)n^{\frac1k}= (k-1)^{\frac1k}.$$ 
\end{proposition}

Interaction rates are much higher than in the indirect-learning case, because without indirect learning much more interaction is needed for competition to be a substantial force. The expected number of ideas learned by each firm is now $O(n^{\frac{k-1}{k}})$, which is still asymptotically lower than in the supercritical case with indirect learning (where the expected number of ideas learned is linear in $n$).

\begin{proof}
For $n$ large, the expected number of potential technologies that firm $i$ produces and which include firm $i$'s private idea is
$$|T_i(\mathbf{p}^*,(q_i,{q}^*_{-i}))| \sim \frac{(p^*)^k}{(k-1)!}(\iota(q_i,q^*)(n-1))^{k-1}$$
The probability that no other firm produces any such technology is
$$(1-\iota(q_i,q^*)\iota(q^*,q^*)^{k-1})^{n-1}+o(1).$$
So $q_i$ is chosen to maximize the number of potential proprietary technologies for $i$:
\begin{equation}\label{eq:directmonopolies}|PT_i(\mathbf{p}^*,(q_i,{q}^*_{-i}))| \sim \frac{(p^*)^k}{(k-1)!}(\iota(q_i,q^*)(n-1))^{k-1}(1-\iota(q_i,q^*)\iota(q^*,q^*)^{k-1})^{n-1}.\end{equation}
The first-order condition for $q_i$ is
$$(k-1)\left(\frac{\partial \iota(q,q^*)}{\partial q}(q_i)\right)(1-\iota(q_i,q^*)\iota(q^*,q^*)^{k-1}) \sim (n-1)\iota(q_i,q^*)\left(\frac{\partial \iota(q,q^*)}{\partial q}(q_i)\right)\iota(q^*,q^*)^{k-1}.$$

We claim that we must have $\iota(q^*,q^*) \rightarrow 0$ at any sequence of equilibria. Else, expected payoffs would vanish asymptotically, but firms could achieve non-vanishing profits by choosing any $q_i$ such that the interaction rate $\iota(q_i,q^*)$ is proportional to $\frac1n$. Thus, the first-order condition implies:
$$\lim_n (n-1)\iota(q_i,q^*)\iota(q^*,q^*)^{k-1} = k-1.$$
At equilibrium, this implies
$$\iota(q^*,q^*) \sim (\frac{k-1}{n-1})^{\frac1k}$$
as desired.
\end{proof}

We can observe from the proof that adding a constant link cost $\epsilon>0$ would not change the result of Proposition~\ref{prop:direct}. This contrasts with Proposition~\ref{prop:patent}(i), where there is no investment equilibrium for any positive $\epsilon$. We next discuss direct learning with patent rights when $k>2$.

\subsection{Patents}\label{sec:patentapp}

Proposition~\ref{prop:patent} considers granting patent rights to some types of ideas when $\delta=0$ and $k=2$. We found that adverse selection interactions prevents the emergence of an investment equilibrium.

We now extend the analysis to $k>2$. There is now an investment equilibrium, but the same adverse selection effect implies that firms with patents choose much lower levels of openness than without patents.

Suppose that a fraction $b\in(0,1)$ of firms receive patents, as in Proposition~\ref{prop:patent}. We maintain the assumption that $\delta=0$. 

\begin{proposition}Suppose a fraction $b\in(0,1)$ of firms receive patents and $\delta=0$. For $k>2$, at any sequence of symmetric invsetment equilibria
$$\lim_n q^*(0) n^{1/k} = \left(\frac{k-1}{b} \right)^{1/k}.$$
\end{proposition}
\begin{proof}
It is weakly dominant for patent rights choose $q^*(1)=1$, and this action is the best response at an investment equilibrium. We claim that at any sequence of symmetric investment equilibria, $q^*(0)n\rightarrow \infty.$
The probability that all ideas in a given technology $t$ including $i$ are known to another firm is
$$(1-q^*(0)^{k})^{bn}+o(1).$$
This converges to zero whenever $q^*(0) n^{1/k} \rightarrow 0$, so if $q^*(0)n$ were bounded along any subsequence then firm $i$ could profitably deviate by increasing $q_i$ when $n$ is large.

Therefore, by the law of large numbers, for a firm $i$ without patents choosing $q_i(0)$ when other firms choose equilibrium actions $p^*_{-i}$ and $q^*_{-i}$, the number of ideas learned from firms without patents is $$|I_i((p_i,{p}^*_{-i}),(q_i,{q}^*_{-i}))| \sim p^*(0) q_iq^*(0)(1-b)n.$$ So the expected number of proprietary technologies for a firm without patents choosing $p_i$ and  $q_i(0)$ is
$$\mathbb{E}[|PT_i((p_i,{p}^*_{-i}),(q_i,{q}^*_{-i}))| ] \sim p_i(0)(p^*(0))^{k-1}(1-q_i(0)q^*(0)^{k-1})^{bn} \binom{q_iq^*(0)(1-b)n}{k-1}.$$
Note that we use our explicit formula $\iota(q_i,q_j) = q_iq_j$  for the interaction rate here.

We can approximate the binomial coefficient with its highest-order term, so taking the first-order condition and cancelling terms gives:
$$bn q^*(0)^{k-1} (q_iq^*(0)(1-b)n)\sim (k-1)(1-b)nq^*(0)(1-q_i(0)q^*(0)^{k-1}) .$$
Taking $q_i=q^*(0)$ and solving, $$q^*(0) \sim  \left(\frac{1}{b}\cdot \frac{k-1}{n}\right)^{1/k}$$
as claimed above.
\end{proof}

As a result, the interaction rate between firms without patents is 
$$\iota(q^*(0),q^*(0)) \sim \left(\frac{1}{b}\cdot \frac{k-1}{n}\right)^{2/k},$$
which is lower order than the interaction rate in Proposition~\ref{prop:direct}. The payoffs to firms without patents are therefore of lower order than in Proposition~\ref{prop:direct}, where no patent rights are granted.

The interaction rate between a firm with a patent and a firm without a patent is
$$\iota(q^*(0),q^*(1)) \sim  \left(\frac{1}{b}\cdot \frac{k-1}{n}\right)^{1/k},$$
which is the same order as the interaction rate in Proposition~\ref{prop:direct}. The payoffs to firms with patents are therefore of the same order as in Proposition~\ref{prop:direct}, where no patent rights are granted.

\section{Firm Size}\label{sec:size}

The baseline model assumed that each firm can discover a single idea. In this section, we consider firms that can instead discover $1 < \sigma < k$ private ideas.\footnote{The assumption that $\sigma<k$ is not essential, but rules out investment equilibria with little or no interaction, e.g., all firms choose $q_i=0$ but private investment $p_i>0$.}

A firm with size $\sigma>1$ can frictionlessly share ideas internally without fear of competition. Similarly, we can interpret a firm of size $\sigma>1$ as the entity created by a licensing agreement between $\sigma$ small firms.\footnote{Decisions about private investment may be different in these two cases, but this will not affect our analysis.}
 
More formally, we continue to let the set of firms be $\{1,\hdots,n\}$ but now allow multiple ideas for each firm. Each of the $\sigma$ ideas corresponding to firm $i$ is discovered independently with probability $p_i$. A firm learning $i$ directly from $j$ will learn all private ideas discovered by firm $j$. The analysis from Sections \ref{sec:equilibrium} and \ref{sec:generalprofits}, including Theorem~\ref{critical} and Propositions~\ref{prop:concavity} and \ref{prop:competition}, extends easily to any firm size $\sigma < k$.

We will now compare the payoffs of firms of two different sizes $\sigma$ and $\sigma'$. Suppose there are fixed positive shares of firms of each size. We can think of the exercise as measuring the value of increasing firm size.

\begin{proposition}\label{prop:firmsize}
If firm $i$ can discover $\sigma$ ideas and firm $i'$ can discover $\sigma'$ ideas, then at any sequence of investment equilibria
$$ \lim_{n \rightarrow \infty} \frac{U_i(\mathbf{p}^*,\mathbf{q}^*)}{U_{i'}(\mathbf{p}^*,\mathbf{q}^*)}=\frac{\sigma}{\sigma'}.$$
\end{proposition}

The proposition says that when payoffs are such that equilibrium is critical or supercritical, then small and large firms obtain the same payoffs per idea asymptotically. An implication is that merging two separate firms would increase their profits by very little for $n$ large.

We can give intuition in the case $\sigma=1$ and $\sigma'=2$. Two separate firms of size one can each potentially produce technologies by combining their private idea with $n-1$ ideas learned from others. A firm of size two can also potentially produce technologies by combining either of its private ideas with $n-1$ ideas learned from others. In addition, the firm of size two can produce technologies by combining both of its private ideas with $n-2$ ideas learned from others. These additional technologies generate any excess profits for the larger firm over the two smaller firms. Because $\binom{y}{k-1}$ is much larger than $\binom{y}{k-2}$ for $y$ large, the additional technologies have a small impact on profits in large markets.

When equilibrium is subcritical, as under the payoff structure of Proposition~\ref{prop:competition}(ii), we have
$$ \lim_{n \rightarrow \infty} \frac{U_i(\mathbf{p}^*,\mathbf{q}^*)}{U_{i'}(\mathbf{p}^*,\mathbf{q}^*)}>\frac{\sigma}{\sigma'}.$$
In this case, because firm profits are bounded asymptotically, technologies using multiple private ideas will generate a non-vanishing share of a firm's profits.

An important assumption in Proposition~\ref{prop:firmsize} is that $\sigma$ and $\sigma'$ do not depend on $n$, so that firms are still small relative to the overall market. A firm that can discover a non-vanishing fraction of all ideas can obtain much higher payoffs per idea than small firms, as such a firm will obtain payoffs of order $n^{k-1}$ even without interacting with other firms.

\begin{proof}[Proof of Proposition~\ref{prop:firmsize}]
We can show that equilibrium is critical or supercritical by a modification of the argument used to prove Theorem~\ref{critical} and Proposition~\ref{prop:competition}, which we now describe.

A version of Lemma~\ref{lem:basicFOC} still applies at any subcritical equilibrium. The statement and proof must be modified, as in the proof of the subcritical asymmetric case of Theorem~\ref{critical}, to accommodate heterogeneity in firms. Because $\sigma <k$, we must have $\tau(t) \geq 1$ for all $t$. Because $\delta q^*n$ is equal to the expectation of $\tau(t)$ with respect to a suitable distribution over technologies $t$, we cannot have a subcritical equilibrium.

Therefore, we have $\lim_{n \rightarrow \infty} U_i(\mathbf{p}^*,\mathbf{q}^*) = \infty$ for firms $i$ of either size.

So for any integer $y>0$ and any $\epsilon>0$, we have $$\mathbb{E}\left[|PT_i(\mathbf{p}^*,\mathbf{q}^*)| \mathbf{1}_{|I_i(\mathbf{p}^*,\mathbf{q}^*)|>y}\right] \geq (1-\epsilon)\mathbb{E}\left[|PT_i(\mathbf{p}^*,\mathbf{q}^*)| \right]$$
for $n$ sufficiently large, where $\mathbf{1}$ is the indicator function. That is, almost all of the profits of firm $i$ are generated in the event that firm $i$ learns at least $y$ ideas.

Because $$\lim_{y \rightarrow \infty} \frac{\binom{y}{k-l}}{\binom{y}{k-1}} = 0$$
for all $l>1$, in expectation at least a share $1-\epsilon$ of technologies in $PT_i(\mathbf{p}^*,\mathbf{q}^*)$ include only one private idea developed by firm $i$ (of either size).

Suppose $\sigma < \sigma'$ and fix firms $i$ of size $\sigma$ and $i'$ of size $\sigma'$. The preceding facts imply that by choosing $(p_i,q_i) = (p_{i'}^*,q_{i'}^*)$, for any $\epsilon>0$ the firm $i$ can guarantee
$$\mathbb{E}\left[|PT_i((p_i,p^*_{-i}),(q_i,q^*_{-i})| \right]\geq (1-\epsilon)^3\mathbb{E}\left[|PT_{i'}(\mathbf{p}^*,\mathbf{q}^*)|\right]$$
for $n$ sufficiently large. Here the first two factors of $(1-\epsilon)$ correspond to the share of proprietary technologies including only one private idea developed by firm $i'$, while we introduce the third because at least a share $1-\epsilon$ of proprietary technologies for firm $i'$ do not include idea $i$. This implies the result.
\end{proof}

This section introduced heterogeneity in firm size. Our results can similarly accommodate heterogeneity in other parameters, such as the complexity $k$ of products produced by a firm or the private investment cost $c(\cdot)$.

\section{Public Innovators and Directed Interaction}\label{sec:publicdirected}

We now show that the result of Proposition~\ref{prop:public} continues to apply if firms can direct their interactions toward private firms or public innovators.

As in Section~\ref{sec:public}, public innovator $i$ pays investment cost $c(p_i)$ and receives a payoff of one for each technology $t$ such that: (1) $i \in t$ and (2) $j \in \{i\} \cup I_i(\mathbf{p},\mathbf{q}) $ for all $j \in t$.

All firms have the same incentives as in the baseline model. Public innovators and firms can now choose two interaction rates $q_{i0}$ and $q_{i1}$, where $q_{i0}$ is the interaction rate with public innovators and $q_{i1}$ is the interaction rate with private firms.

We show payoffs again grow at the same rate as in the supercritical region, up to a constant factor:
\begin{proposition}\label{prop:publicdirected}
Suppose a non-vanishing share of agents are public innovators. Then there exists a sequence of symmetric equilibria with non-vanishing investment, and at any sequence of equilibria with non-vanishing investment
$$\liminf_n \frac{U_i(\mathbf{p}^*,\mathbf{q}^*)}{\binom{n-1}{k-1}}>0$$
for all firms $i$.
\end{proposition}
\begin{proof}

Let $b(n)$ be the share of public innovators for each $n$.

We first show that $$\liminf_n \frac{U_i(\mathbf{p}^*,\mathbf{q}^*)}{\binom{n-1}{k-1}}>0$$
for all $i$ at any sequence of equilibria with non-vanishing investment.

It is weakly dominant and strictly preferred at any investment equilibrium for public innovators to choose $q^*_{i0}=q^*_{i1}=1.$
 Therefore, all public innovators are in the same component of the learning network. Private investment $p_i$ by public innovators is non-vanishing, so asymptotically almost surely all firms in this component learn at least $\alpha n$ ideas for some $\alpha > 0$.

Each private firm can obtain expected payoffs $O(n^{k-1})$ by choosing $q_{i0}=\frac{1}{n}$ and $q_{i1}=0$. This is because then the probability that firm $i$ learns indirectly from a public innovator and no firm $j$ learns from $i$ is non-vanishing, and the payoffs from this event are $O(n^{k-1})$. This shows the desired bound on $\frac{U_i(\mathbf{p}^*,\mathbf{q}^*)}{\binom{n-1}{k-1}}$.

It remains to show there exists a sequence of symmetric equilibria with non-vanishing investment. Suppose that all public innovators choose $p_0\geq \frac12$ and $q_{i0}=q_{i1}=1$ and all firms other than $i$ choose $(p_1,q_0,q_1)$ with $p_1 \geq \frac12,$ $\delta q_{0}n \leq 1$ and $q_1=0$. If $q_{i0}$ is the best response for $i$, then $\lim_n q_{i0}n$ exists and is independent of $p_0$ and $(p_1,q_0,q_1)$ given the constraints in the previous sentence. This is because the probability of interactions between $i$ and other firms vanishes asymptotically, while the best response does not depend on the number of ideas learned by the unique giant component.

Therefore, we can choose $\epsilon>0$ such that if $q_0 \in [\frac{\epsilon}{\delta n},1],$ then so is an arbitrary firm $i$'s best response $q_{i0}$. Because all other private firms choose $q_{j1}=0$, firm $i$ is indifferent to all choices of $q_{i1}$ and in particular $q_{i1}=0$ is a best response. We claim that for $n$ large, given $\mathbf{p}$, there exists $\mathbf{q}$ such that each each $q_i$ is a best response to $(p_0,p_1,q_0,q_1).$ This follows from Kakutani's fixed point theorem as in the proof of Theorem~\ref{critical}. We call this choices of openness $q_0(p_0,p_1)$.

Given such $(p_0,p_1,q_0(p_0,p_1))$, each firm has a non-vanishing probability of learning a linear number of ideas. Therefore, $\mathbb{E}[|I_i(\mathbf{p},\mathbf{q})|]\rightarrow \infty.$ So any best response $p_i$ for each public innovator has $p_i \geq \frac12$, and the same holds for each firm. By Kakutani's fixed point theorem, there exists $(p_0,p_1,q_0(p_0,p_1))$ such that $p_0 \geq \frac12$ and $p_1 \geq \frac12$ are also best responses. Thus there exists a sequence of equilibria with non-vanishing investment.
\end{proof}

\end{document}